\documentclass{article}
\usepackage[utf8]{inputenc} 
\usepackage{float}
\usepackage{graphicx}
\usepackage{amsmath}
\usepackage{blkarray}
\usepackage{amsfonts}
\usepackage{xcolor}

 \usepackage{fancyhdr,graphicx,amssymb}
\usepackage[ruled,vlined]{algorithm2e}
\usepackage{algorithmic}
\include{pythonlisting}

\usepackage[colorlinks, citecolor=blue, linkcolor=black]{hyperref}
\usepackage{fullpage,amssymb,amsmath,amsthm}
\usepackage{amsthm}
\usepackage{comment}
\usepackage{mwe} 
\usepackage{subcaption}
\usepackage{hyperref}

\newtheorem{theorem}{Theorem}[section]
\newtheorem{claim}[theorem]{Claim}

\newtheorem{definition}[theorem]{Definition}
\newtheorem{lemma}[theorem]{Lemma}
\newtheorem{observation}[theorem]{Observation}

\newcommand{\Sps}{\Sigma\Pi\Sigma}

\newcommand{\Spsp}{\Sigma\Pi\Sigma\Pi}
\newcommand{\SMLSps}{\Sps_{\set{\sqcup_j X_j}}}

\newcommand{\sk}{ $\Sps_{\set{\sqcup_j X_j}}(k)$}
\newcommand{\pow}{\Sigma\!\wedge\!\Sigma}
\newcommand{\powk}[1]{\Sigma\!\wedge\!\Sigma({#1})}

\newcommand{\set}[2][]{{\left\{#2\right\}}^{#1}}

\newcommand{\proofsketch}{\vspace*{-1ex} \noindent { \textit{Proof Sketch.} }}

\newlength\myindent
\newcommand{\D}{\Delta}

\newcommand{\condset}[2]{\set{#1 \; \left| \;#2 \right. }}
\newcommand{\restrict}[1]{| _{#1}}
\newcommand{\divs}{\; | \;}

\newcommand{\Fn}{\mathbb{F}[x_{1},x_{2},\ldots ,x_{n}]}

\newcommand{\T}{\mathcal{T}}
\newcommand{\x}{\times}

\newcommand{\BigO}{\mathcal{O}}
\newcommand{\eqdef}{\stackrel{\Delta}{=}}
\newcommand{\nequiv}{\not \equiv}
\newcommand{\simp}{\mathrm{sim}}
\newcommand{\var}{\mathrm{var}}

\newcommand{\wh}{\mathrm{w_H}}

\newcommand{\cF}{\overline{\F}}

\newcommand{\RFn} {\F(x_{1},x_{2},\ldots ,x_{n})}

\newcommand{\Fpar}[2]{\frac{\partial {#1}}{\partial {x_#2}}}

\newcommand{\bools}[1]{\set[#1]{0,1}}
\newcommand{\Hybrid}{\gamma}
\newcommand{\Sys}{\mathrm{Sys}}
\newcommand{\Drank}{\mathrm{\Delta_{rank}}}
\newcommand{\ri}{r_{init}}
\newcommand{\Con}{\mathrm{Con}}

\newcommand{\xb}{\bar{x}}

\newcommand{\ub}{\bar{u}}
\newcommand{\ab}{\bar{a}}
\newcommand{\bb}{\bar{b}}
\newcommand{\cb}{\bar{c}}

\newcommand{\al}{\alpha}
\newcommand{\be}{\beta}

\newcommand{\mc}{\mathcal }
\newcommand{\poly}{\mathrm{poly}}
\newcommand{\qpoly}{\mathrm{quasipoly}}
\newcommand{\rs}{|_{\sigma}}
\newcommand{\abs}[1]{\left|{#1}\right|}
\newcommand{\size}[1]{\abs{#1}}

\newcommand{\rank}{\mathrm{rank}}

\newcommand{\Q}{{\mathbb Q}}
\newcommand{\R}{{\mathbb R}}

\newcommand{\C}{{\mathbb C}}
\newcommand{\N}{{\mathbb{N}}}
\newcommand{\Z}{{\mathbb Z}}
\newcommand{\F}{{\mathbb F}}

\newcommand{\cupdot}{\ensuremath{\mathaccent\cdot\cup}}
\newcommand{\ignore}[1]{}
\newcommand*{\rom}[1]{\expandafter\@slowromancap\romannumeral #1@}

\newtheorem{THEOREM}{Theorem}

\newtheorem{corollary}[theorem]{Corollary}
\newtheorem{remark}[theorem]{Remark}

\newcommand{\G}{\mathbb{G}}

\title{Reconstruction Algorithms for Low-Rank Tensors and Depth-3 Multilinear Circuits. 
}

\author{ Vishwas Bhargava\thanks{Department of Computer Science, Rutgers University, Piscataway, NJ 08854. Research supported in part by the Simons 
Collaboration on Algorithms and Geometry and NSF grant CCF-1909683. Email: {\tt vishwas1384@gmail.com}.}
\and
Shubhangi Saraf\thanks{
Department of Mathematics \& Department of Computer Science, Rutgers University, Piscataway, NJ 08854. Research supported in part by NSF grants
CCF-1350572, CCF-1540634, CCF-1909683, BSF grant 2014359, a Sloan research fellowship and the Simons
Collaboration on Algorithms and Geometry.
Email: {\tt shubhangi.saraf@gmail.com}.}
\and 
Ilya Volkovich
\thanks{Department of Computer Science, Boston College, Chestnut Hill, MA 02467. Email: {\tt ilya.volkovich@bc.edu}.}
}

\date{}

\begin{document}

\maketitle

\begin{abstract}
We give new and efficient black-box reconstruction algorithms for some classes of depth-$3$ arithmetic circuits. As a consequence, we obtain the first efficient algorithm for computing the tensor rank and for finding the optimal tensor decomposition as a sum of rank-one tensors when then input is a {\em constant-rank} tensor. More specifically, we provide efficient learning algorithms that run in randomized polynomial time over general fields and in deterministic polynomial time over $\R$ and $\C$ for the following classes:

\begin{enumerate}
    \item {\it Set-multilinear depth-$3$ circuits of constant top fan-in ($\SMLSps(k)$ circuits)}. As a consequence of our algorithm, we obtain the first polynomial time algorithm for tensor rank computation and optimal tensor decomposition of constant-rank tensors. This result holds for $d$ dimensional tensors for any $d$, but is interesting even for $d=3$.
    \item {\em Sums of powers of constantly many linear forms ($\pow(k)$ circuits)}.  As a consequence we obtain the first polynomial-time algorithm for tensor rank computation and optimal tensor decomposition of constant-rank symmetric tensors.
    \item {\em Multilinear depth-3 circuits of constant top fan-in (multilinear $\Sps(k)$ circuits)}. Our algorithm works over all fields of characteristic 0 or large enough characteristic. Prior to our work the only efficient algorithms known were over polynomially-sized finite fields \cite{KarninShpilka09}. 
\end{enumerate}
Prior to our work, the only polynomial-time or even subexponential-time algorithms known (deterministic or randomized) for subclasses of $\Sps(k)$ circuits that also work over large/infinite fields were for the setting when the top fan-in $k$ is at most $2$ \cite{Sin16, Sin20}.

\end{abstract}

\thispagestyle{empty}
\newpage



\pagenumbering{arabic}
\section{Introduction}

\emph{Arithmetic circuits}  are directed acyclic graphs (DAG) computing  multivariate polynomials succinctly, building up from variables using ($+$) addition and ($\times$) multiplication operations. \textit{Reconstruction} of arithmetic circuits is the following problem: given black-box  (a.k.a oracle/ membership query) access to a polynomial computed by a circuit $C$ of size $s$ from some class of circuits $\cal{C}$, give an efficient algorithm (deterministic or randomized) for recovering $C$ or some circuit $C'$ that computes the same polynomial as $C$. This problem is the algebraic analogue of exact learning in Boolean circuit complexity \cite{Angluin88}. If one additionally requires that the output circuit belongs to the same class $\cal{C}$ as the input circuit, then it is called {\em proper learning}. 

Reconstruction of arithmetic circuits is an extremely natural problem, but also a really hard problem. Thus in the past few years, much attention has focused on reconstruction algorithms for various interesting subclasses of arithmetic circuits \cite{ BBBKV00, KlivansSpielman01, KlivansShpilka06, ForbesShpilka12}. In particular, much attention has focused on depth-$3$ and depth-$4$ arithmetic circuits \cite{KarninShpilka09, GKL12,Sin16,BSV20,Sin20}. Depth-$3$ and depth-$4$ circuits have been intensely studied for the problem of proving lower bounds, deterministic polynomial identity testing as well as polynomial reconstruction (which is probably the hardest of the three). Given the depth reduction results of~\cite{AgrawalVinay08, Koiran10, Tavenas13, GKKS13}, we know that depth-$3$ and depth-$4$ arithmetic circuits are very expressive, and good enough reconstruction algorithms (or even lower bounds or polynomial identity testing) for these models would have major implications for general circuits. Thus perhaps not surprisingly, we are quite far from obtaining efficient reconstruction algorithms even for depth-$3$ circuits. 

In this work, we will focus on some interesting subclasses of depth-$3$ circuits with bounded top fan-in ($\Sps(k)$ circuits) and give efficient {\em proper learning} algorithms for them. A setting of particular interest for us (and which motivated much of this work) is when the underlying field is large or infinite (such as $\R$ or $\C$), since in that setting we have even fewer reconstruction algorithms. Though we state many of our results over all fields, for concreteness it will be convenient to imagine the underlying field being $\R$ or $\C$ or $\F_p$. 

The subclasses of $\Sps(k)$ circuits that we study, already capture some very interesting models, and our result for one of these subclasses implies the first efficient polynomial-time algorithm for tensor rank computation and optimal tensor decomposition of {\em constant-rank tensors}. Before describing the connection to tensors and stating our results, we first give some background on polynomial reconstruction. \\


There is substantial evidence supporting the hardness of arithmetic circuit reconstruction. Deterministic algorithms for reconstruction are at least as hard as deterministic black-box algorithms for polynomial identity testing, which is equivalent to proving lowering bounds for general arithmetic circuits~\cite{KabanetsImpagliazzo03, Agrawal05}.  Randomized reconstruction is also believed to be a hard problem and there are a number of results showing hardness of reconstruction under various complexity-theoretic and cryptographic assumptions \cite{Hastad90,FortnowKlivans09,KlivansSherstov09, Shitov16}. (For more details see the section on hardness-results in \cite{BSV20}). 


Despite reconstruction being a very hard problem, there has been a lot of research focused on efficient reconstruction for restricted classes of arithmetic circuits. Yet, the progress has still been quite slow. Even among the class of constant-depth arithmetic circuits, we only understand reconstruction well for a handful of restricted cases \cite{ KlivansSpielman01,  KarninShpilka09,GKL12, Sin16,BSV20}. 
If one studies {\em average case reconstruction} (a model that has received increased attention in recent years) then we know a number of additional results and they hold for richer circuit classes~\cite{ GKL11a, GKQ14,  KNST17, KNS18, KayalSaha19, GargKayalSaha20}. However we will not discuss this setting much since the focus of this work will be on the worst case setting. 



Before describing the status of what we know about reconstruction for some of the relevant circuit classes, we first define some natural classes of arithmetic circuits that will play an important role in our discussion.

\paragraph{Some Definitions of Relevant Circuit Classes}
The model of  depth-$3$ arithmetic circuits with top fan-in $k$, which we refer as  $\Sps(k)$ circuits, has three layers of alternating $\Sigma$ and $\Pi$ gates and  computes a polynomial of the form
\begin{equation*}
C(\bar{x}) = \sum\limits _{i = 1}^k T_i(\bar{x}) =
\sum\limits _{i = 1}^k \prod\limits_{j=1}^{d_{i}}l_{ij}(\bar{x})
\end{equation*} 
where the $l_{ij}(\bar{x})$-s are linear polynomials.

A multilinear polynomial is a polynomial with individual degree of each variable bounded by 1. We say that a circuit $C$ is multilinear (or syntactically multilinear) if every gate in $C$  computes a multilinear polynomial. Thus, a \emph{multilinear} $\Sps(k)$
circuit is a $\Sps(k)$ circuit in which each multiplication gate $T_i$ computes a multilinear polynomial.  

A more refined subclass of multilinear polynomial is that of \emph{set-multilinear} polynomials. Let $ \sqcup _{j \in [d]} X_j$ be a partition of the set $X$ of input variables. Then a polynomial is set-multilinear under partition $\sqcup_{j \in [d]} X_j$ if  each monomial of the polynomial picks up \textit{exactly} one variable from each part in the partition. 

A \emph{set-multilinear} $\Sps(k)$ circuit under partition $ \sqcup _{j \in [d]} X_j$ (which we denote as $\SMLSps(k)$
circuit) is a $\Sps(k)$ circuit in which each multiplication gate $T_i$ computes a set-multilinear polynomial respecting the partition $ \sqcup _{j \in [d]} X_j$. 
In the Section~\ref{sec:tensorintro} we will discuss this model and its connection to tensor decomposition. 

The final subclass of $\Sps(k)$ circuits that we discuss is the innocuous looking class of sum of power of $k$ linear forms, also referred to as \emph{diagonal} depth-3 circuits with bounded top fan-in  ($\pow(k)$ circuits). These are a subclass of $\Sps(k)$ circuits where instead of using multiplication gates, we are just allowed  powering gates which raise an input linear polynomial to some power. In Section~\ref{sec:tensorintro} we will discuss this model and its connection to {\it symmetric} tensor decomposition.
 
\paragraph{Proper Learning}
The focus of this work will be on {\em proper learning} algorithms for subclasses of $\Sps(k)$ circuits.

Note that in the setting of proper learning, if $\mc{C}'$ is a subclass of $\mc{C}$, then an efficient proper learning algorithm for $\mc{C}$ does not imply an efficient proper learning algorithm for $\mc{C}'$.
Indeed, as some evidence towards this, note that there are efficient algorithms for proper learning of read-once algebraic branching programs (ROABPs) ~\cite{ BBBKV00, KlivansShpilka06,ForbesShpilka12}, but we do not know proper learning algorithms for $\SMLSps$ circuits and $\pow$ circuits (with no bound on the top fan-in), which are both subclasses of ROABPs. In fact, it is known that properly learning $\SMLSps$ circuits or $\pow$ circuits with an optimal bound for the top fan-in is NP-hard~\cite{Hastad90, Shitov16}.






Reconstruction algorithms for $\Sps(k)$ circuits and for subclasses of $\Sps(k)$ circuits have been studied in the past a fair bit. The only proper reconstruction algorithms that we are aware of are for the model of multilinear $\Sps(k)$ circuits by Karnin and Shpilka~\cite{KarninShpilka09}  and for $\Sps(2)$ circuits by Sinha \cite{Sin16, Sin20}. In the case of $\Sps(2)$ circuits, the algorithms are proper (i.e. the output is also a $\Sps(2)$ circuit) only if the ``rank" of the linear forms in the underlying circuit is large enough. 

All three of these results are highly nontrivial and they introduce several beautiful techniques which give insight into the structure of these models. 
The Karnin-Shpilka result is in fact more general and gives reconstruction algorithms for $\Sps(k)$ circuits without the multilinearity constraint, but in this setting the learning algorithms aren't proper (and they do not work over large fields) and we will not discuss it here. For multilinear  $\Sps(k)$ circuits as well, the running time of the Karnin-Shpilka algorithm has a polynomial dependence on the field size $|\F|$. Thus it works only over polynomially-sized finite fields, and in particular it does not work over large or infinite fields (which is the primary focus of this work). We discuss the algorithm from~\cite{KarninShpilka09} in a little more detail in Section~\ref{sec:multintro}.

Our goal is to obtain algorithms that work over infinite fields ($\R$, $\C$) with polynomial dependence on the input bit complexity, and that work over finite fields $\F_q$ with $\poly (\log q)$ dependence on the field size.
In this setting, the only subclasses of $\Sps(k)$ circuits for which we know proper learning algorithms is for $\Sps(2)$ circuits, if the ``rank" of the linear forms in the underlying circuit is large enough ~\cite{Sin16, Sin20}. Both these results use fairly sophisticated tools, and really show why even for the seemingly simple case of $k=2$, reconstruction can be fairly complex.
 
 Some additional classes of bounded depth circuits for which we do know proper learning algorithms that work over large fields are  depth-$2$ ($\Sigma\Pi$) arithmetic circuits (a.k.a sparse polynomials) which have efficient polynomial-time algorithms \cite{Ben-OrTiwari88, KlivansSpielman01}, and multilinear depth-$4$ circuits with top fan-in 2 (multilinear $\Spsp(2)$ circuits) \cite{GKL12}.

\subsection{Connection to the Tensor Rank Problem}
\label{sec:tensorintro}
Tensors, higher dimensional analogues of matrices, are  multi-dimensional arrays with entries from some field $\F$. For instance, a $3$-dimensional tensor can be written as $\mc{T} = (\alpha_{i,j,k} ) \in \F^{n_1 \times n_2 \times n_3}$. 
We will work with general $d$-dimensional tensors $\mc{T}=(\alpha_{j_1,j_2,\ldots, j_d} ) \in \F^{n_1 \times \cdots \times n_d}$.
The \emph{rank} of a tensor $\mc{T}$ can be defined as the smallest $r$ for which $\mc{T}$ can be written as a sum of $r$ tensors of rank $1$, where a rank-$1$ tensor is a tensor of the form $v_1 \otimes \cdots \otimes v_d$ with $v_i \in \F^{n_i}$. Here $\otimes$ is the Kronecker (outer) product a.k.a \emph{tensor product}. The expression of $\mc{T}$ as a sum of such rank-$1$ tensors, over the field $\F$ is called \emph{$\F$-tensor decomposition} or just tensor decomposition, for short. The notion of tensor rank/decomposition has become a fundamental tool in different branches of modern science with applications in  machine learning, statistics, signal processing, computational complexity, psychometrics,
linguistics and chemometrics. We refer the reader to the detailed monograph by Landsberg \cite{Landsberg12} and the references therein for more details on applications of tensor decomposition. \\

\noindent For a tensor $\mc{T}=(\alpha_{j_1,j_2,\ldots, j_d} ) \in \F^{n_1 \times \cdots \times n_d}$  consider the following polynomial 
\begin{equation*}
    f_{\mc{T}}(X) \eqdef \sum_{(j_1,\ldots, j_d) \in [n_1] \times \cdots \times [n_d]} \, \alpha_{j_1,j_2,\ldots, j_d} x_{1,j_1} x_{2,j_2} \cdots x_{d,j_d}.
    \end{equation*}
Let $C(X) = \sum \limits _{i=1}^k \prod \limits_{j=1}^d \ell_{i,j}$ be a set-multilinear depth-$3$ circuit over $\F$ respecting the partition 
$ \sqcup _{j \in [d]} X_j$, and
computing $f_{\T}(X)$.
Then observe that $$\T = \sum \limits_{i=1}^k \bar{v}(\ell_{i,1}) \otimes \cdots \otimes \bar{v}(\ell_{i,d})$$ 
where $\bar{v}(\ell_{i,j})$ corresponds to  the linear form $\ell_{i,j}$ as an $n_j$-dimensional vector over $\F$. Indeed, it is easy to see that a tensor $\mc{T}=(\alpha_{j_1,j_2,\ldots, j_d} ) \in \F^{n_1 \times \cdots \times n_d}$ has rank at most $r$ if and only if $f_{\T}(X)$ can be computed by a $\SMLSps(r)$ circuit.
Therefore, rank of $\T$ is the smallest $k$ for which $f_{\T}(X)$ can be computed by a $\SMLSps(k)$ circuit.\\

Consider the following question. {\bf Question 1:} Given as input a $3$-dimensional tensor $\mc{T} = (\alpha_{i,j,k} ) \in \F^{n_1 \times n_2 \times n_3}$, is there an efficient algorithm for computing its tensor rank?
This problem is known to be NP-hard in general~\cite{Hastad90}. Now consider the following variant of the question. {\bf Question $1'$:} Given as input a $3$-dimensional tensor $\mc{T} = (\alpha_{i,j,k} ) \in \F^{n_1 \times n_2 \times n_3}$ such that the tensor rank is at most some fixed constant. Does the problem still remain hard, or is the rank efficiently computable? One could also ask these same questions for $d$-dimensional tensors where $d$ is large. Let $\mc{T}=(\alpha_{j_1,j_2,\ldots, j_d} ) \in \F^{n_1 \times \cdots \times n_d}$. In such a setting, one might not even be able to efficiently store the entire tensor as an array. However, if the tensor rank is small (say a constant), then there is still a small ``implicit" representation of $\mc{T}$ a sum of rank one tensors. In this setting, one has black-box access to measurements of $\mc{T}$. In particular, given $\bar{\alpha}_i \in \F^{n_i}$ for all $i\in [d]$, the measurement of $\mc{T}$ at $(\bar{\alpha_1}, \ldots, \bar{\alpha}_d)$ equals $\langle \mc{T}, \bar{\alpha}_1 \otimes \cdots \otimes \bar{\alpha}_d \rangle$. The $d$-dimensional question is strictly harder than the three dimensional question, and again we can ask ({\bf $d$-dimensional analog of Question $1'$})- suppose the tensor rank of $\mc{T}$ is at most some fixed constant. Is there an efficient algorithm for computing the tensor rank of $\mc{T}$? 

Observe that each measurement of $\mc{T}$ at $(\bar{\alpha_1}, \ldots, \bar{\alpha_d})$ corresponds to a black-box evaluation of the polynomial $f_{\T}$ at $(\bar{\alpha_1}, \ldots, \bar{\alpha}_d)$. Moreover, finding the optimal decomposition of $\mc{T}$ as a sum of rank-1 tensors is equivalent to the following: Given black-box access to $f_{\T}$, reconstruct it as a set-multilinear $\SMLSps$ circuit with the smallest possible top fan-in. 

The three dimensional version was asked as an open question in the work of Schaefer and Stefankovic~\cite{SchaeferStefankovic16}. In a related setting, a version of the $d$-dimensional variant (efficiently learning an optimal decomposition of a constant-rank tensor by black-box access to the measurements) was also raised in the recent work of Chen and Meka~\cite{ChenMeka20}.
It turns out that the answer to the above question is extremely sensitive to the underlying field. For instance, if the underlying field is the rationals ($\mathbb Q$), then even if the tensor rank is a constant, computing the exact value of the tensor rank over $\mathbb Q$ is not known to be decidable (and is, in fact, believed to be undecidable)~\cite{Shitov16, SchaeferStefankovic16}.

In this paper, we give the first randomized polynomial-time algorithm for computing the tensor rank of a constant-rank, $d$-dimensional tensor $\mc{T}$ \footnote{It is possible that the algorithm of Karnin and Shpilka~\cite{KarninShpilka09} for learning multilinear $\Sps(k)$ circuits can be adapted to also properly learn set-multilinear $\Sps(k)$ circuits. The Karnin-Shpilka algorithm has a polynomial dependence on field size $|\F|$. If there algorithm can be adapted then it would give a polynomial-time algorithm over small finite fields. The algorithms in this paper work over infinite fields as well, and that setting was the primary motivation for this work.}. Over the fields $\R$ and $\C$ we also show how to obtain deterministic polynomial time algorithms. Moreover, our algorithm finds the optimal decomposition of $\mc{T}$ as a sum of rank-1 tensors. Our algorithm works over fields such as $\R$, large enough finite fields, $\C$, and any other algebraically closed fields. Over other fields, we are only able to compute the tensor rank when we view the entries of the tensor as elements of some extension field.



\begin{THEOREM}[Informal]
Let $k$ be any constant. There exists a randomized polynomial-time algorithm that given black-box access to a polynomial $f \in \F[X]$ computable by a $\SMLSps(k)$ circuit over $\F$, and the partition $\sqcup X_i$ of the set of variables $X$, outputs a $\SMLSps(k)$ circuit computing $f$. When $\F$ is $\R$ or $\C$ then our algorithm is deterministic. 
\end{THEOREM}

This implies a polynomial-time algorithm to compute the optimal tensor decomposition (and hence also the tensor rank) of constant-rank tensors for various fields. The formal version of the result is given in Theorem \ref{THEOREM:main1}

Our proof uses various ingredients such as a variable reduction procedure, and setting up and solving a system of polynomial equations. Another important ingredient used is the {\em rank bounds} that were developed in the study of polynomial identity testing for $\Sps(k)$ circuits \cite{DvirShpilka07,KayalSaxena07, KayalSaraf09, SaxenaSeshadhri09, SaxenaSeshadhri10}. These are structural results for identically zero $\Sps(k)$ circuits, and essentially show that under some mild conditions, any $\Sps(k)$ circuit which computes the identically zero polynomial must have its linear forms contained in a
``low-dimensional" space. This understanding led to very efficient deterministic polynomial identity testing results for this class, and then eventually were used in efficient reconstruction algorithms for subclasses of $\Sps(k)$ circuits as well. 


\paragraph{Symmetric Tensors: } Just as we asked the question of tensor rank computation for general tensors, we can also ask the analogous questions for {\em symmetric tensors}. 

A tensor $\mc{T}$ is called symmetric if $X_1=X_2=\cdots =X_d$ and we have $\mc{T}(i_1, i_2, \ldots, i_d ) = \mc{T}(j_1, j_2, \ldots, j_d )$ whenever $(i_1, i_2, \ldots, i_d )$ is a permutation of $(j_1, j_2, \ldots, j_d )$. Thus, a symmetric tensor is a higher order generalization of a symmetric matrix. Analogous to tensor rank, \textit{symmetric rank} is obtained when the constituting rank-1 tensors are imposed to be themselves symmetric, that is $\bar{v}\otimes \bar{v} \cdots \otimes \bar{v}$. 

Just like in the case of general tensors, computing the symmetric rank reduces to  finding the optimal top fan-in of a special class of arithmetic circuits, which is sum of power of linear forms ($\pow$) circuits. The class of $\pow(k)$ circuits computes polynomials of the form $f = \ell_1^d + \cdots \ell_k^d$ where each $\ell_i$ is a linear polynomial over the underlying $n$ variables.

Let $C(X) = \sum \limits _{i=1}^k \ell_{i}^d$ be a $\powk{k}$ circuit over $\F$ computing $f_{Sym,\T}(X)$ for a symmetric tensor $\mc{T}=(\alpha_{j_1,j_2,\ldots, j_d} ) \in \F^{n_1 \times \cdots \times n_d}$.
Then $$\T = \sum \limits_{i=1}^k \bar{v}(\ell_{i}) \otimes \cdots \otimes \bar{v}(\ell_{i})$$ 
where $\bar{v}(\ell_{i})$ is a  $n$-dimensional vector corresponding to  the linear form $\ell_{i,j}$.

Just as in the case of tensor rank, determining the symmetric rank of tensors is also known to be NP-hard~\cite{Shitov16}. One could still ask if there are efficient algorithms for determining the symmetric rank when the rank is constant. In this paper, we give (what we believe to be) the first randomized polynomial-time algorithm for computing the symmetric tensor rank of a  constant-rank $d$-dimensional symmetric tensor $\mc{T}$.



\begin{THEOREM}[Informal]
Let $k$ be any constant. Let $\F$ be any field of characteristic $0$ or sufficiently large characteristic.
There exists a randomized polynomial-time algorithm that given black-box access to a polynomial $f \in \F[X]$ computable by a $\pow(k)$ circuit with constant $k$ over $\F$, outputs a $\pow(k)$ circuit computing $f$. When $\F$ is $\R$ or $\C$ then our algorithm is deterministic.
\end{THEOREM}

This implies a polynomial-time algorithm to compute the optimal symmetric tensor decomposition (and hence also the symmetric tensor rank) of constant-rank symmetric tensors over various fields.
The formal version of the result is given in Theorem \ref{THEOREM:main2}.

Our proof in this case also uses a variable reduction procedure, and setting up and solving a system of polynomial equations. However the proof is overall way simpler than that for general tensors (and actually fits in about half a page!). 

\subsection{Multilinear $\Sps(k)$ circuits}\label{sec:multintro}
Multilinear $\Sps(k)$ circuits are a more general class of circuits than $\SMLSps(k)$ circuits. In the proper learning setting however, a proper learning algorithm for multilinear $\Sps(k)$ circuits does not imply a proper learning algorithm for $\SMLSps(k)$ circuits.

In this paper we also study reconstruction algorithms for multilinear $\Sps(k)$ circuit. Multilinear $\Sps(k)$ circuits were studied by by Karnin and Shpilka~\cite{KarninShpilka09} and they give the first polynomial-time algorithm for this class of circuits. However the running time of the Karnin-Shpilka algorithm has a polynomial dependence on the field size $|\F|$. Thus it works only over polynomially sized finite fields, and in particular it does not work over infinite fields \footnote{The Karnin-Shpilka \cite{KarninShpilka09} result is in fact more general and gives reconstruction algorithms for $\Sps(k)$ circuits without the multilinearity constraint, but in this setting the learning algorithms aren't proper and we will not discuss it.}.

At a very high level, the way the algorithm works in~\cite{KarninShpilka09} is as follows. It finds a suitable projection of the input circuit where only constantly many variables are kept ``alive" and the rest are set to field constants. The new circuit in constantly many variables has only constantly many field elements appearing as coefficients, and hence in time $\poly(|\F|)$ one can efficiently ``guess'' it by going over all possibilities for what the projected circuit looks like. Once the algorithm hits upon the correct guess of the projected circuit, then it  ``lifts" the projected circuit to recover the original circuit. The implementation of the lifting procedure is quite clever and uses a very nice clustering procedure. The only place where the prohibitive dependence on the field size comes up is in guessing the projected circuit. 

In this work we give the first randomized polynomial-time proper learning algorithm for this model that works over large fields (and in particular infinite fields). Our algorithm works over all fields of characteristic 0 or characteristic greater than $d$ (where $d$ is the degree of the circuit). Over $\R$ and $\C$ we show how to derandomize the above algorithm and to obtain {\it deterministic} polynomial time algorithms. 
Several of the ideas in our algorithm are inspired by the algorithm from~\cite{KarninShpilka09} but we need several new ideas as well. 

One similarity we have with~\cite{KarninShpilka09} is that we also project to constantly many variables and try to learn the projected circuit. Instead of ``guessing" or iterating to find the projected circuit, we reduce the problem to solving a suitable system of polynomial equations. The problem is that the projected circuit may not have a unique representation as a multilinear $\Sps(k)$ circuit, and hence the representation learnt by polynomial system solving might be just some representation (not the original representation) and it might not be liftable. This leads to some subtleties and the rest of the algorithm and how we implement the lift is quite different. We give a more detailed overview in Section \ref{sec:multilinearoverview}.

\begin{THEOREM}[Informal]
Let $k$ be a constant. Let $\F$ be any field of characteristic $0$ or sufficiently large characteristic.
There exists a randomized polynomial-time algorithm that given black-box access to a polynomial $f \in \F[X]$ computable by a multilinear $\Sps(k)$ circuit over $\F$, outputs a multilinear $\Sps(k)$ circuit computing $f$. Over $\R$ and $\C$ the algorithms we obtain are in deterministic polynomial time. 
\end{THEOREM}

This implies a polynomial-time algorithm for learning multilinear $\Sps(k)$ circuits over infinite fields.
The formal version of the result is given in Theorem \ref{THEOREM:main3}.

\subsection{Our results}
We now state our results. All our algorithms will be randomized algorithms over general fields, and hence algorithms will output the correctly reconstructed circuit with high (say $\geq 0.9$) probability.  This probability can boosted to $1-o(1)$ by simply doing independent repetitions. Over $\R$ and $\C$, all our algorithms are deterministic. 

Our first main result is a polynomial-time algorithm for proper learning of the class of $\SMLSps(k)$ circuits. 

\begin{theorem}[Proper learning $\SMLSps(k)$ circuits]
\label{THEOREM:main1}
Given black-box access to a degree $d$, $n$ variate polynomial $f \in \F[X]$ computable by a $\SMLSps(k)$ circuit over $\F$, and given the partition $\sqcup X_i$ of the set of variables $X$, there is a randomized $\poly(d^{k^3}, k^{k^{k^{10}}}, n, c)$ time algorithm for computing a $\SMLSps(k)$ circuit computing $f$, where $c = \log q$ if $\F = \F_q$ and $c$ is the maximum bit complexity of any coefficient of $f$ if $\F$ is infinite. When the underlying field $\F$ is $\R$ or $\F_q$ with $q \geq nd\cdot 2^{k+1}$ or algebraically closed, then the output circuit is over $\F$ as well. Otherwise the output circuit is over a degree $\poly(k^{k^{k^{10}}})$ extension of $\F$. Moreover when $\F$  is $\R$ or $\C$, then we show that the above algorithm can be made to run in deterministic time $\poly(d^{k^3}, k^{k^{k^{10}}}, n, c)$.
\end{theorem}




We would like to remark that this is the first proper learning algorithm for $\SMLSps(k)$ circuits, and it works over all fields. We feel this result is particularly interesting in the setting of large or infinite fields such as $\R$ or $\C$, and understanding reconstruction algorithms in that setting was the goal of this work. 
If we didn't require the learning to be ``proper" and were okay with letting the output be a polynomial from a bigger class, then such algorithms were already known (even without the restriction of top fan-in)~\cite{BBBKV00,KlivansShpilka06}. 

By the equivalence described in Section~\ref{sec:tensorintro} (see also Lemma \ref{lem:decomposition}),
we obtain the following immediate corollary to Theorem \ref{THEOREM:main1} which for  constant-rank tensors gives us an efficient tensor decomposition algorithm for expressing the input tensor as sum of rank one tensors.

When $\T$ is a $d$ dimensional tensor, as described in Section~\ref{sec:tensorintro}, even storing all of $\T$ as an array is too inefficient. However if the rank is small, there is still a small implicit description of $\T$. We consider the setting when we have black-box access to measurements of $\T$ (as described in Section~\ref{sec:tensorintro}). This exactly corresponds to having black-box access to the associated polynomial $f_{\T}$. 

\begin{corollary}[Decomposing fixed rank tensors]
Let $\T \in \F^{n_1 \times \cdots \times n_d}$ be a $d$-dimensional tensor of rank at most $k$. Let $n = \sum_{i=1}^d n_i$. Given black-box access to measurements of $\T$ (equivalently to evaluations of $f_{\T}$), there exists a randomized $\poly(d^{k^3}, k^{k^{k^{10}}}, n, c)$ time algorithm for computing a decomposition of $\T$ as a sum of at most $k$ rank 1 tensors, where $c = \log q$ if $\F = \F_q$ and $c$ is the maximum bit complexity of any coefficient of $f$ if $\F$ is infinite. 
When the underlying field $\F$ is $\R$ or $\F_q$ with $q \geq nd \cdot 2^{k+1}$ or algebraically closed, then the decomposition is over $\F$ as well. Otherwise the decomposition will be over (a degree $\poly(k^{k^{k^{10}}})$) extension of $\F$. Moreover when $\F$  is $\R$ or $\C$, then we show that the above algorithm can be made to run in deterministic time $\poly(d^{k^3}, k^{k^{k^{10}}}, n, c)$.
\end{corollary}

Notice that we can use the above result to obtain an efficient algorithm for computing the exact value of the tensor rank of the input tensor (at least over $\R$, $\C$, large finite fields and other algebraically closed fields). Over other fields we can only compute the tensor rank over an extension field. The way one can compute the tensor rank is as follows: run the above algorithm for all values of $k$ starting from $k = 1$, and the smallest $k$ for which the algorithm successfully outputs a tensor decomposition will be the tensor rank of $\T$. (Note that one can test when the output is successful by a simple randomized polynomial identity test.)

\begin{remark} \label{tensorRemark}
 The dependence on k (exponential tower  of size 2) is not optimized in the above theorem and corollary and can be improved to a single exponential in $k$ when $\F=\C, \R,$ (see Section~\ref{sec:syssolving} and Section~\ref{sec:setmultilinear} for details).  However, the single exponential dependence on $k$ is expected as tensor decomposition is NP-hard in general \cite{Hastad90, SchaeferStefankovic16} and not even known to be computable for $\Q$, thus justifying our need to go to extension fields.  See Section \ref{sec:tensorhardness} for more details on hardness of tensor decomposition.  
\end{remark}

Note that in the case of constant dimensional tensors (i.e. when one can actually efficiently look at all the entries), we can simulate black-box access to the polynomial $f_{\T}$,
given access to the entries of the tensor and vice versa. Thus in the constant dimensional setting our algorithm also gives a way for computing tensor rank and obtaining the optimal tensor decomposition given access to the entries of the tensor. This in particular answers an open question asked by Schaefer and Stefankovic \cite{SchaeferStefankovic16}, who asked as an open question the complexity of computing the tensor-rank when the rank is constant. Our proof in the constant dimensional setting is simpler than that for the setting of growing $d$. In the setting of $d$ dimensional tensors (for large or growing $d$) the question of whether one can get improved efficiency when the rank of $\T$ is constant was raised in the work of Chen and Meka \cite{ChenMeka20} (in a slightly different context). Our work addresses and resolves this question in the black-box query setting for worst case tensors. 


Analogous to the result above for tensor decomposition of general tensors, we also obtain efficient algorithms for optimal symmetric tensor decomposition of constant-rank symmetric tensors. The setting of constant-rank symmetric tensors ends up being much simpler than general tensors, and our proofs for this model are much simpler. This result will follow as a corollary of the next result, which is a randomized polynomial-time algorithm for proper learning of $\pow(k)$ circuits.

\begin{theorem}(Proper learning $\Sigma\wedge\Sigma(k)$ circuits)
\label{THEOREM:main2}
Given black-box access to a degree $d$, $n$ variate polynomial $f \in \F[X]$ computable by a $\pow(k)$ circuit over $\F$,
such thar $char(\F) > d \mbox{ or } 0$,  there is a randomized $\poly{\left((dk)^{k^{k^{10}}}, n, c \right) }$ time algorithm for computing a $\pow(k)$ circuit computing $f$, where $c = \log q$ if $\F = \F_q$ and $c$ is the maximum bit complexity of any coefficient of $f$ if $\F$ is infinite. When the underlying field $\F$ is $\R$ or $\F_q$ with $q \geq nd2^{k}$ or algebraically closed, then the output circuit is over $\F$ as well. Otherwise the output circuit is over a degree $\poly((dk)^{k^{k^{10}}})$ extension of $\F$. Moreover when $\F$  is $\R$ or $\C$, then we show that the above algorithm can be made to run in deterministic time $\poly{\left((dk)^{k^{k^{10}}}, n, c \right) }$.
\end{theorem}

By the equivalence described in Section~\ref{sec:tensorintro},
we obtain the following immediate corollary to Theorem \ref{THEOREM:main2} which for  constant-rank tensors gives us an efficient symmetric tensor decomposition algorithm for expressing the input tensor as sum of rank one symmetric tensors.

\begin{corollary}[Decomposing fixed symmetric rank tensors]
Let $\T$ be a \textit{symmetric} $d$-dimensional tensor of side length $n$, with $\F$-entries and symmetric rank at most $k$, such that $char(\F) > d \mbox{ or } 0$. Given black-box access to $f_{Sym, \T}$, there is a randomized $\poly{\left((dk)^{k^{k^{10}}}, n, c \right)}$ time algorithm for computing a decomposition of $\T$ as a sum of at most $k$ rank 1 symmetric tensors, where $c = \log q$ if $\F = \F_q$ and $c$ is the maximum bit complexity of any coefficient of $f$ if $\F$ is infinite.  When the underlying field $\F$ is $\R$ or $\F_q$ with $q \geq nd2^{k}$ or algebraically closed, then the decomposition is over $\F$ as well. Otherwise the decomposition will be over (a degree $\poly \left((dk)^{k^{k^{10}}}\right
)$ extension of $\F$. Moreover when $\F$  is $\R$ or $\C$, then we show that the above algorithm can be made to run in deterministic time $\poly{\left((dk)^{k^{k^{10}}}, n, c \right) }$.
\end{corollary}
 
Again, like in the case of general tensor decomposition, Remark \ref{tensorRemark} holds here as well. 
 
 We next state our result on proper learning of multilinear $\Sps(k)$ circuits. 

\begin{theorem}[Proper learning multilinear-$\Sps(k)$ circuits]
\label{THEOREM:main3}

Given black-box access to a degree $d$, $n$ variate polynomial $f \in \F[X]$ computable by a multilinear $\Sps(k)$ circuit over $\F$,
such that $char(\F) > d \mbox{ or } 0$,  there is a randomized $\poly \left(n^{k^{k^{k^{10}}}}, d,  k^{k^{k^{k^{10}}}}, c\right)$ time algorithm  for computing a multilinear $\Sps(k)$ circuit computing $f$, where $c = \log q$ if $\F = \F_q$ and $c$ is the maximum bit complexity of any coefficient of $f$ if $\F$ is infinite. When the underlying field $\F$ is $\R$ or $\F_q$ with $q \geq nd\cdot 2^{k+1}$  or algebraically closed, then the output circuit is over $\F$ as well. Otherwise the output circuit is over a degree $\poly\left((k)^{k^{k^{k^{10}}}}\right)$  extension of $\F$. Moreover when $\F$  is $\R$ or $\C$, then we show that the above algorithm can be made to run in deterministic time $\poly \left(n^{k^{k^{k^{10}}}}, d,  k^{k^{k^{k^{10}}}}, c\right)$.
\end{theorem}

This is the first efficient proper learning algorithm for multilinear-$\Sps(k)$ circuits that works over large fields, and in particular infinite fields such as $\R$ and $\C$. Even here, the dependence on k (exponential tower  of size 3) is not optimized in the above theorem  and can be improved to a tower of size 2  in $k$ when $\F=\C, \R,$ (see Section \ref{sec:multilinear} for details).

\paragraph{Deterministic vs Randomized Reconstruction Algorithms:} 
The algorithms we give in this paper are randomized over general fields and deterministic over $\R$ and $\C$. Indeed, derandomizing them in  general, will be highly nontrivial for the following reason. In the reconstruction problem for all three subclasses of $\Sps(k)$ circuits being studied, we can embed within them the problem of solving a system of polynomial equations. (See Theorem~\ref{thm:tensorhardness} and the discussion in Section~\ref{sec:tensorhardness}.)
The only efficient algorithms we know for solving systems of polynomial equations over large finite fields (i.e with running time polynomial in $\log q$ for a field $\F_q$) are randomized and it is a very interesting open question to derandomize them.  A derandomized solution to our reconstruction algorithms over large finite fields would have very interesting algorithmic implications for polynomial system solving, see \cite[Problem 15]{AMcC94}. 

Interestingly, the large characteristic case is the only case when low-variate polynomial system solving is hard to derandomize. That is, if the underlying field is not a finite field with large characteristic, then there do exist efficient deterministic algorithms for low-variate polynomial system solving. See Section \ref{Det_system} for details. Also, this turns out to be the only bottleneck for derandomizing our learning/decomposition algorithms. That is, if the underlying field is not a finite field with large characteristic, then the algorithms underlying 
Theorems ~\ref{THEOREM:main1}, ~\ref{THEOREM:main2}, ~\ref{THEOREM:main3} can be derandomized efficiently. Though we do not mention this explicitly, it is easy to see that, when $\F=\F_{p^t}$ then the algorithms mentioned in Theorems ~\ref{THEOREM:main1}, ~\ref{THEOREM:main2}, ~\ref{THEOREM:main3} can be made deterministic with an additional polynomial in $p$ (characteristic) dependence in time complexity. See derandomization remarks in respective sections for details. 

When we present our proofs, for simplicity we will first present the randomized algorithms and then later point out the changes that need to be made in order to derandomize them. 



\subsection{Related/Previous Work}\label{sec:previous }

Reconstruction of $\Sps(k)$ circuits has received a fair amount of attention. The case of $k=1$ is resolved by the black-box factoring algorithm of Kaltofen and Trager \cite{KaltofenTrager90}. The case of $k=2$ is already highly nontrivial and very interesting and thus needed quite a few new ideas. This case was first studied by Shpilka \cite{Shpilka09}, who designed a  reconstruction algorithm for $k=2$ which was later improved by  Karnin and  Shpilka \cite{KarninShpilka09} who gave efficient reconstruction algorithms for ($\Sps(k)$) circuits for any constant top fan-in $k$. When the input is an $n$-variate, degree $d$ polynomial computed by a size $s$ circuit, both  algorithms run in time  $\qpoly(n, d, \size{\F},s)$. The algorithms are not `proper learning' algorithms, and the output is from a larger class of ``generalized"  depth-$3$ circuit. Moreover given the dependence of the running time on the field size, these algorithms aren't efficient over large/infinite fields. 

Over fields of characteristic $0$, the only efficient reconstruction algorithm we know for $\Sps(k)$ circuits is the randomized algorithm by~\cite{Sin16}  which works for $k=2$, and uses lots of new ideas such as quantitative/robust Sylvester-Gallai theorems for high dimensional points. Very recently, in \cite{Sin20}, Sinha studied the case of $k=2$  for finite fields and gave the first algorithm in this setting with $\poly \log$ dependence in field size. These algorithms are mostly proper, but not always. When the rank of the linear forms in the input polynomial is not high dimensional, then the output circuit might not be a $\Sps(2)$ circuit.

When the input is a \textit{multilinear} $\Sps(k)$ circuit, the works of Shpilka \cite{Shpilka09} and  Karnin-Shpilka \cite{KarninShpilka09}  give polynomial-time {\em proper learning} algorithms. The dependence on the field size is still $\poly(|\F|)$, and hence these algorithms do not work over large/infinite fields. Inspired by the work of Karnin and Shpilka, in \cite{BSV20} similar results were obtained for multilinear depth-4 circuits with bounded top fan-in ($\Spsp(k)$ circuits). The running time is however still at least $\poly(|\F|)$, and hence it does not work over large/infinite fields. When the top fan-in is $2$, i.e. for $\Spsp(2)$ circuits, we do know such efficient polynomial-time reconstruction algorithms by the work of Gupta, Kayal and Lokam \cite{GKL12}.

\paragraph{Other Results:}
The class of circuits for which we understand reconstruction really well is the class of depth-$2$ ($\Sigma\Pi$) arithmetic circuits (a.k.a sparse polynomials). We can properly learn sparse polynomials  in  \textit{deterministic} $\poly(s,n,d)$ time over any field \cite{Ben-OrTiwari88, KlivansSpielman01}. Another class for which we understand reconstruction \textit{reasonably} well is the class of read-once oblivious branching programs (ROABPs). Klivans and Shpilka \cite{KlivansShpilka06} gave a randomized reconstruction (proper learning) algorithm for ran in time $\poly(n,d,s)$.  This was later derandomized in \cite{ForbesShpilka12} with time complexity $\qpoly(n,d,s)$. For depth-$3$ circuits, reconstruction algorithms for various other restricted classes have been studied. For instance, for set-multilinear depth-$3$ circuits \cite{BBBKV00, KlivansShpilka06} gave a randomized poly(n,d,s) (improper) learning algorithm which outputs an ROABP. 

Recently, there has been a flurry of activity in average case learning algorithms for various arithmetic  circuit classes \cite{GKL11a, GKQ14,  KNST17, KNS18, KayalSaha19, GargKayalSaha20}. These results can be thought of as worst case reconstruction, given some non-degeneracy condition holds for some implicit  polynomials (which are usually computed by intermediate gates). Interestingly, these results fall under the umbrella of learning from natural lower bounds which is an exciting area of research in arithmetic as well as Boolean circuit complexity \cite{CIKK16, KayalSaha19}.

\section{Proof Overview}

We have three main results in the paper:(1) reconstruction of $\SMLSps(k)$ circuits (equivalent to low rank tensor decomposition), (2) reconstruction of $\pow(k)$ circuits 
(equivalent to low rank symmetric tensor decomposition), and (3) reconstruction of multilinear $\Sps(k)$ circuits.

Our algorithms are randomized over general fields and we show how to derandomize then over $\R$ and $\C$. For simplicity, in the proof overview we will only discuss the randomized algorithms. Later in the paper when we give the formal proof we will show how to derandomize the algorithms. 

A common theme in the proof of each of these results is that all proofs involve a {\em variable reduction procedure} and setting up and {\em solving a suitable system of polynomial equations}, where a solution to the system gives some important information about the circuit being reconstructed. In the case of reconstruction for $\SMLSps(k)$ circuits and multilinear $\Sps(k)$ circuits, the proofs are considerably more involved and also use ``rank bound" techniques that give structural information about $\Sps(k)$ circuits that are identically 0.

For simplicity, we start with a proof overview of the result that was (in hindsight) quite easy to prove, which is coming up with an efficient reconstruction algorithm for $\pow(k)$ circuits. 

\subsection{Reconstruction of $\pow(k)$ circuits}

Let $f$ be a polynomial which has a $\powk{k}$ representation, and let 
 \[C_f \equiv \sum_{i=1}^{k}  (a_{i,1}x_{1} +  a_{i,2}x_{2} + a_{i,3}x_{3}+ \ldots +  a_{i,n}x_{n})^{d}\] be the $\powk{k}$ circuit computing $f$.

An important observation is that if $f$ can be represented by a $\powk{k}$ circuit, then $f$ has only $k$ ``essential variables". In particular one can apply an invertible linear transformation to the variables of $f$ so that the transformed $f$ only depends on $k$ variables. 

What is nice is that such a linear transformation can actually be computed without
actually looking at $C_f$ and its linear forms, but only with black-box access to $f$. 
    This follows from result of Kayal~\cite{Kayal11}, and which built upon a result by Carlini~\cite{Carlini06}.  (The original result by Kayal was not stated or used in the black-box setting, but it is easy to see that the proof an be adapted to black-box setting as well.) Let $g_A(\bar{x}) = f(A\cdot \bar{x})$, where $g_A(\bar{x})$ depends only on $k$ variables. Since the algorithm can compute $A$, hence given black-box access to $f$, it can efficiently simulate black-box access to $g_A$. Moreover, observe that $g_A$ also has a $\pow(k)$ representation. Thus if we can learn a $\pow(k)$ representation of $g_A$, then by simply applying the inverse linear transform, one can recover a $\pow(k)$ representation of $f$.

Thus the new goal is to learn a $\pow(k)$ representation of $g_A$ given black-box access to it. We will do this be reducing the problem of learning the $\pow(k)$ representation of $g_A$ to solving a suitable system of polynomial equations. Recall that $g_A$ only depends on $k$ variables. Thus the monomial representation of $g_A$ only has ${k+d \choose d}$ monomials. Since $k$ is small, this quantity is not too big, and one can invoke black-box reconstruction algorithms for sparse polynomials~\cite{Ben-OrTiwari88,KlivansSpielman01} to learn $g_A$ as a sum of monomials. Let $g_A=  \sum_{\bar{e} }c_{\bar{e}} \cdot \bar{x}^{\bar{e}}$ be the monomial representation of $g_A$.

Let \[C_{g_A} \equiv \sum_{i=1}^{k}  (b_{i,1}x_{1} +  b_{i,2}x_{2} + b_{i,3}x_{3}+ \ldots +  b_{i,k}x_{k})^{d}\] be a $\powk{k}$ representation of $g_A$.

Then notice that 
 \[\sum_{i=1}^{k}  (b_{i,1}x_{1} +  b_{i,2}x_{2} + b_{i,3}x_{3}+ \ldots +  b_{i,k}x_{k})^{d} = \sum_{\bar{e} }c_{\bar{e}} \cdot \bar{x}^{\bar{e}}.\] 

Now for each monomial $\bar{x}^{\bar{e}}$ that appears in $g_A$, we can compare the coefficient of $\bar{x}^{\bar{e}}$ on both sides of the above expression to get a polynomial equation in the variables $b_{i,j}$. Doing this for all monomials gives us
a system of at most ${k+d \choose d}$ polynomial equations in $k^2$ variables, with $b_{i,j}$ as the unknown variables. Observe that any solution to the system of equations would give a $\pow(k)$ representation of $g_A$ an vice versa. By Theorem~\ref{thm:polysystem}, this system can be solved in polynomial time if $k$ is a constant. 

\subsection{Reconstruction of $\SMLSps(k)$ circuits}
We now show how to efficiently reconstruct $\SMLSps(k)$ circuits. Again, variable reduction and setting up and solving polynomial systems of equations play an important role, but several other ingredients (such as rank bound techniques) also go into the proof and the proof is more involved. 

We are given as input black-box access to a degree $d$, $n$ variate polynomial $f \in \F[X]$ computable by a $\SMLSps(k)$ circuit over $\F$, and we are also given the partition $\sqcup X_i$ of the set of variables $X$. Let $ C_f \equiv  \sum_{i=1}^k \prod_{j=1}^d \ell_{i,j}, $ be a $\SMLSps(k)$ representation of $f$, where each $\ell_{i,j}$ is a linear polynomial in $X_j$ variables.  

\subsubsection{Variable reduction:}

As a first step, we show how to reduce the number of variables in each part to at most $k$. Here we cannot directly invoke the result by Kayal~\cite{Kayal11} and Carlini~\cite{Carlini06} for the following reasons. The total number of essential variables is $k \times d$ which is quite large. Though the number of essential variables in every part is at most $k$, there seems to be no straightforward way to apply the result separately to each part\footnote{Since the linear maps might then end up being over the field of rational functions in the remaining variables.}. Even if $kd$ was small, after applying the linear transformation given by the Carlini-Kayal result, the new circuit might not be set-multilinear, and we need to crucially maintain set-multilinearity in order for the other steps of the algorithm to be carried out. 

Instead, we use the structural properties of set-multilinear circuits to come up with a a different black-box algorithm for performing the variable reduction. We essentially come up with $d$ different invertible linear transformations, one for each set of variables in the partition, that reduces the variables in each set to at most $k$. In Section~\ref{sec:widthreduction-setmultilinear} we elaborate more on how we find these transformations using some properties of the underlying class of circuits. After this step is performed, one can essentially assume that the input circuit is such that each set of the partition has at most $k$ variables.

\subsubsection{Reconstructing low degree ($d \leq 2k^2$) $\SMLSps(k)$ circuits:} Once we have the variable reduction established, we proceed along the same lines as the algorithm for reconstructing $\pow(k)$ circuits. Since the degree is small, the number of monomials appearing in $f$ is small, and the total number of variables appearing in $f$ is small. ( Unlike the symmetric case where the number of monomials was small even for high degree circuits). One can invoke black-box reconstruction algorithms for sparse polynomials~\cite{KlivansSpielman01, Ben-OrTiwari88} to learn $f$ as a sum of monomials. Then, similar to the $\pow(k)$ case, we set up a system of polynomial equations in $\poly(k)$ variables such that every solution to the system corresponds to a $\SMLSps(k)$ representation of $f$. For more details, see Section~\ref{sec:lowdegree-setmultilinear}.

\subsubsection{Reconstructing high degree ($d > 2k^2$) $\SMLSps(k)$ circuits:}
The high level plan for reconstructing general high degree $\SMLSps(k)$ circuits is to use induction on $k$. When $k=1$, then the algorithm just invokes a black-box factoring algorithm such as~\cite{KaltofenTrager90}. Now assume $k \geq 2$.

Our first step will be to just learn  any one linear form appearing in $C_f$. (Actually as a first step it will be convenient to learn two distinct linear forms such that each multiplication gate contains at most one of them.) 
In the next step we will use that linear form to learn most of the linear forms of $C_f$. In the final step we will try to learn all the linear forms and obtain a full $\SMLSps(k)$ representation of $f$. 

\paragraph{Learning one (or two) linear forms appearing in $C_f$: }
The algorithm chooses $k^2$ sets of variables in the partition $X = \sqcup X_i$ to keep ``alive" and sets the variables in the remaining sets to random values. Let the resulting restricted polynomial be $f_R$ and the resulting constant degree $\SMLSps(k)$ circuit be $C_{f_R}$.

Now, we already know reconstruction algorithms (from the previous case) for low degree $\SMLSps(k)$ circuits which we could invoke. 
If we could learn $C_{f_R}$, then in particular we would have learnt several linear forms of $C_f$. 
However note that all we have is black-box access to $f_R$, which might not have a unique $\SMLSps(k)$ circuit representation. In fact it might have exponentially many $\SMLSps(k)$ circuit representations, and our reconstruction algorithm would learn one of these representations. Thus it is possible that we do not learn $C_{f_R}$, but some other $\SMLSps(k)$ circuit representation of $f_R$, call it $C'_{f_R}$ . Now a priori it may seem that the linear forms in $C'_{f_R}$ might not have anything in common with the linear forms of $C_{f_R}$ or $C_f$. However using rank bound arguments that have been used extensively in the past to analyze identically $0$ $\Sps(k)$ circuits (for polynomial identity testing and polynomial reconstruction), one can show two distinct $\SMLSps(k)$ representations of the same polynomial must indeed have many linear forms in common (as long as the degree is large enough, which it is in our case).
Thus we get that $C_{f_R}$ and $C'_{f_R}$ (which we learnt) must have many linear forms in common. Though we may not know exactly which linear form of $C'_{f_R}$ also appears in $C_f$, we can come up with a small list of candidate options and then iterate over these options. Any wrong candidate will not lead to a successful output of the final algorithm and we will be able to detect it by a later testing phase. Thus we can effectively assume we know a linear form in $C_f$. In fact if we do things more carefully we can ensure that we know two linear forms $\ell_1$ and $\ell_2$ appearing in $C_f$ such that they are supported on the same subset of variables. 

\paragraph{Learning most of the linear forms from each multiplication gate of $C_f$: }
Once we learn $\ell_1$ and $\ell_2$ appearing in $C_f$, we try to learn more linear forms as follows. (We don't need $f_R$ any more or $C'_{f_R}$)

The algorithm applies a suitable random setting of the variables of $\ell_1$ in the polynomial $f$, that makes $\ell_1$ evaluate to $0$, and results in a circuit with $< k$ multiplication gates. Call the restricted polynomial $f_{R_1}$ and let $C_{f_{R_1}}$ be the restricted version of $C_f$. By the inductive hypothesis, we can learn a $\SMLSps(k-1)$ representation of $f_{R_1}$. Call this $C'_{f_{R_1}}$. If we could actually learn the representation $C_{f_{R_1}}$ then we would have learnt most of the linear forms in all the multiplication gates of $C_f$ that did not get set to zero under the restriction. However we can only learn some other representation, which we called $C'_{f_{R_1}}$. Using rank bound arguments, we will however still be able to argue that $C'_{f_{R_1}}$ and $C_{f_{R_1}}$ have a lot in common. In fact we show that each multiplication gate of $C_{f_{R_1}}$ overlaps almost entirely (in all but $k$ linear forms) with some multiplication gate of $C'_{f_{R_1}}$. 
Repeating this procedure for the other linear form $\ell_2$ as well gives us another restricted circuit $C_{f_{R_2}}$ and the version of it that is learnt which is $C'_{f_{R_2}}$. It is now easy to see that each multiplication gate of $C_f$ overlaps almost entirely (in all but $k$ linear forms) with some multiplication gate of $C'_{f_{R_1}}$ or $C'_{f_{R_1}}$. 

Once we have this, by iterating over all ways of matching up the multiplication gates and choices of overlap, we can make generate a polynomial sized list of $k$-tuples $(G_1, G_2, \ldots G_k)$ which has the following property. 
One of the $k$-tuples $(G_1, G_2, \ldots G_k)$ from the list will have the property that $f = G_1H_1 + \cdots G_kH_k$ and $G_iH_i = T_i$ where $T_i$ was one of the multiplication gates in the original $\SMLSps(k)$ representation of $f$, $C_f$. Each $G_i$ has degree $d-k^2$ and hence each $H_i$ has degree $k^2$. By a little bit of more effort we can also ensure that all the $H_i$ depend on the same sets of the underlying variable partition. The final algorithm will go over all possible $k$-tuples $(G_1, G_2, \ldots G_k)$ from the list in order to find the correct one. All the wrong ones will not lead to a successful reconstruction, and will get eliminated by a later testing phase. 




\paragraph{Learning the full $\SMLSps(k)$ representation of $f$: }
We now assume that we have learnt $k$ polynomials $G_1, G_2, \ldots, G_k$ such that $f = G_1H_1 + \cdots G_kH_k$. $G_iH_i = T_i$ where $T_i$ was one of the multiplication gates in the original $\SMLSps(k)$ representation of $f$. Each $H_i$ is a polynomial in $k^3$ variables of degree at most $k^2$ (since after variable reduction each part had at most $k$ variables)  and all the $H_i$ depend on the same sets of the underlying variable partition. 
We need to now learn the $H_i$, or even some variation of them which will eventually lead to a full $\SMLSps(k)$ representation of $f$.\\

We demonstrate how we do this with some simple examples. 
As a simple case, suppose that the $G_i$ are linearly independent polynomials. By substituting random values into the variables of the $G_i$, we obtain black-box access to a random linear combination of $H_1, \ldots H_k$. Call this linear combination $P_1$. From  black-box access to $P_1$, we can actually obtain the monomial representation of $P_1$ using black-box interpolation for sparse polynomials. We can repeat this process $k$ times to get $k$ different random linear combinations of $H_1, \ldots H_k$. The linear independence of $G_1, G_2, \ldots G_k$ implies that these random linear combinations will be linearly independent with high probability (see Lemma~\ref{invert1}). Since we know the $G_i$, we actually know to coefficients of the random linear combinations. Thus once we learn these combinations, we can invert the transformation and actually get black-box access to each $H_i$ individually. Once we have black-box access to each $H_i$, we can factorize them in a black-box way and hence recover the full underlying circuit. 

Here is a slightly more general case. Imagine that $k = 3$, $G_1$ and $G_2$ are independent, but $G_3 = G_1 + G_2$. Since we actually know the $G_i$s, we can learn their linear dependency structure (for instance by taking enough random evaluations of them and learning the linear dependence structure of the evaluations, see Lemma~\ref{invert}).
Then, $$C_f = G_1 H _1 + G_2 H_2 + (G_1 + G_2)H_3 = G_1(H_1 + H_3) + G_2(H_2 + H_3)$$
Let $H_1 + H_3 = K_1$ and $H_2 + H_3 = K_2$. 
Now just as in the simple case when all the $G_i$s were independent, we can again learn the monomial representation of two distinct random linear combinations of $K_1$ and $K_2$, and then use this to recover the monomial representations of $K_1$ and $K_2$. What remains is to find a representation of $K_1$ which looks like $H_1 + H_3$ and a representation of $K_2$ which looks like $H_2 + H_3$. Individually, each looks like a case of finding a  $\SMLSps(2)$ representation for low degree polynomials, but these two $\SMLSps(2)$ representations are entangled since they must share a multiplication gate. However we can set up one big system of polynomial equations for solving both these reconstruction problems at the same time that takes into account the shared multiplication gate. 

This more general case that we just described contains most of the ideas for the fully general case. For more details, refer to Section~\ref{sec:fullcircuit}.

\subsection{Reconstruction of multilinear $\Sps(k)$ circuits}\label{sec:multilinearoverview}
We now give a proof overview and describe our algorithm for efficiently learning multilinear $\Sps(k)$ circuits. The main goal of this result is to find a procedure which also works over large and infinite fields. 

Variable reduction and setting up and solving polynomial systems of equations again play an important role, especially for the case of low degree multilinear $\Sps(k)$ circuits. However the implementation of this technique and how to set up and solve the system of equations is more subtle. For general high degree multilinear $\Sps(k)$ circuits, we need several other tools such as a clustering procedure (inspired by the work of~\cite{KarninShpilka09}, rank bounds, the notion of rank preserving subspaces, black-box factoring algorithms and an error correcting procedure.  

\subsubsection{Reconstruction of low degree multilinear $\Sps(k)$ circuits} 

Think of $k$ (the top fan-in) and $d$ (the degree) to be constants, and the number of variables, $n$, to be growing.
Let $f \in \Fn$ be a polynomial computed by a degree $d$, multilinear $\Sps(k)$ circuit $C$ of the form 
\begin{equation} \label{d3}
\sum\limits _{i = 1}^k T_i(\bar{x})= \sum_{i= 1}^k
\prod_{j=1}^{d_i}\ell_{i,j}(\bar{x})
\end{equation} 

where for each fixed $i$, the different $\ell_{i,j}$ are supported on disjoint variables.

Let $m$ be the number of essential variables in $f$. Since there at most $kd$ linear forms appearing in $C$, it is easy to see that the number of essential variables in $f$, i.e. $m$, is at most $kd$.

We now apply a variable reduction procedure, and for this we invoke the result by Kayal\cite{Kayal11} and Carlini~\cite{Carlini06} (Lemma~\ref{lem:carlini-black-box}) to efficiently compute an invertible linear transformation $A \in \F^{n \times n}$ such that $f(A \cdot \bar{x})$ only depends on the first $m$ variables. 

Let $g(\bar{x}) = f(A \cdot \bar{x})$. Observe that given black-box access to $f$, one can easily simulate black-box access to $g$. Also since $g(A^{-1} \cdot \bar{x}) = f(\bar{x})$, any algorithm that can efficiently learn $g$ can also efficiently learn $f$
in the following way. For each $i \in [n]$, suppose that $R_i$ denote the $i$th row of $A^{-1}$. Then in the $i$th input to $g$ we simply input the linear polynomial $L_i = \langle R_i, \bar{x} \rangle$, which is the inner product of $R_i$ and the vector $\bar{x}$ of formal input variables. Since $g$ only depends on the first $m$ variables, we only really need to do this operation for $i\in [m]$. 

Since $f$ is computed by a degree $d$ multilinear $\Sps(k)$ circuit, hence $g(X) = f(A\cdot \bar{x})$ also has a natural degree $d$ $\Sps(k)$ circuit representation, where the linear forms of that representation are obtained by applying the transformation $A$ to corresponding linear forms of $C$. Let us call this circuit $C_g$. Notice that $C_g$ {\em may not be multilinear}. However, if were somehow able to learn the precise circuit $C_g$, then by substituting each variable $x_i$ to $L_i$ then we would recover the circuit $C$ which is indeed multilinear. 

Thus our goal is now the following. We have black-box access to $g$ which only depends on $m$ variables. We would like to devise an algorithm for reconstructing $C_g$. 
Now here is a slight issue. $C_g$ is a {\em  particular} degree $d$ $\Sps(k)$ representation of $g$. It has the nice property that when we plug in $x_i = L_i$ (for all $i \in [m]$) in this representation, then we recover a multilinear $\Sps(k)$ representation of $f$. Let us call the new circuit obtained by plugging in $x_i = L_i$ for each $i$, the ``lift" of $C_g$.
Observe that $g$ might have multiple (perhaps exponentially many) representations as a degree $d$ $\Sps(k)$ circuit. If given black-box access to $g$, the reconstruction algorithm finds some other degree $d$ $\Sps(k)$ representation of $g$, call it $C_g'$, then there is no guarantee that when we plug in $x_i = L_i$ in this representation, then we recover a multilinear $\Sps(k)$ representation of $f$. In other words, the lift of $C_g'$ in general {\em may not be multilinear}. 

Although in our algorithm we will not actually be able to guarantee that we learn precisely $C_g$, however the existence of $C_g$ tells us that {\em there exists} a $\Sps(k)$ representation of $g$ whose lift is a multilinear $\Sps(k)$ circuit. We will use this existence to actually find a suitable $\Sps(k)$ representation of $g$ whose lift is multilinear.

In order to learn a degree $d$ $\Sps(k)$ representation of $g$ we will set up a system of polynomial equations such that any solution to it will give as a degree $d$ $\Sps(k)$ representation of $g$. (We do this in a very similar manner to how we did it for $\pow(k)$ circuits and $\SMLSps(k)$ circuits.) We then show how to impose several additional polynomial constraints to this system that will further ensure that whatever $\Sps(k)$ representation is learnt will be such that its lift will be a multilinear $\Sps(k)$ circuit. The details of how we implement this can be found in Lemma~\ref{lem:low degree}.

\subsubsection{Reconstructing general (high degree) multilinear $\Sps(k)$ circuits}
We now describe our algorithm for reconstructing general multilinear $\Sps(k)$ circuits. What we describe here is a bit of a simplification and it avoids some technical issues, but we hope that it provides a high level picture of the algorithm. 

\paragraph{Clustering of gates: } 
We use a very nice and elegant clustering procedure devised in the work of Karnin and Shpilka~\cite{KarninShpilka11} (which they used for reconstructing $\Sps(k)$ circuits over small fields). We will not describe the algorithm here, but describe some nice properties that the clustering satisfies. 
Given as input $C = \sum_i T_i$ where the $T_i$ are the multiplication gates of a degree $d$ multilinear $\Sps(k)$ circuit $C$, the clustering algorithm looks at the $T_i$ and outputs a partition of the the $k$ multiplication gates into a set of clusters $C_1, C_2, \ldots C_r$ (for some $r \in [k]$). Each cluster $C_i$ is some subset of the multiplication gates of $C$, and has the property that any two multiplication gates in a cluster are very ``close" to each other. Suppose that $C_i = \{T_{i_1}, T_{i_2}, T_{i_3}\}$. Then consider the associated circuit $C_i' = T_{i_1}+ T_{i_2}+ T_{i_3}$. The closeness of every two of the multiplication gates will imply that one can write $C_i'$ as $$C_i' = T_{i_1}+ T_{i_2}+ T_{i_3}= \gcd(T_{i_1}, T_{i_2}, T_{i_3}) \times\left(T'_{i_1}+ T'_{i_2}+ T'_{i_3}\right)$$ where $T'_{i_1}+ T'_{i_2}+ T'_{i_3}$ is a {\em low degree} multilinear $\Sps(3)$ circuit. 
Now notice that we don't know what $C$ is (that is what we are trying to learn) and hence we cannot apply any clustering procedure to it. However this clustering exists, and it is canonical. 
We only have black-box access to the original circuit $C$.
Suppose that we could somehow obtain black-box access to each of the clusters (or rather to the circuits corresponding to the clusters). We would then actually be done! Here is why. Suppose we had black-box access to $C_i'$, then we would first apply a black-box factoring algorithm (such as that given by~\cite{KaltofenTrager90}) to compute all the linear factors of $C_i'$ (thus we would obtain $\gcd(T_{i_1}, T_{i_2}, T_{i_3})$) and divide them out. We would then be left with black-box access to $T'_{i_1}+ T'_{i_2}+ T'_{i_3}$ is a {\em low degree} multilinear $\Sps(3)$ circuit. But we already saw how to reconstruct low degree multilinear $\Sps(3)$ circuits! By multiplying it with its linear factors, we we would be able to recover a multilinear $\Sps(3)$ circuit for $C_i$. We would repeat this procedure for each cluster and then put it all together to obtain a multilinear $\Sps(k)$ representation for $C$. 

Thus the goal from now on will be to somehow obtain black-box access to the clusters. The clustering output by the clustering algorithm also has some additional nice properties. It is a ``robust" clustering, that is, if two multiplication gates got assigned to different clusters, then they are quite ``far" from each other (in some well defined sense). This nice property ends up implying the following. 
We start with the circuit $C$ in $n$ variables. Then there is some constant number (about $k^k$) of variables one can keep ``alive" (call these the $\bar{y}$ variables) such that if we set the remaining variables (call these the $\bar{z}$ variables) to random values ($\bar{z} = \bar{\alpha}$), then the new  restricted circuit $C|_{\bar{z}= \bar{\alpha}}$ has the following property. Suppose we applied the clustering algorithm to $C|_{\bar{z}= \bar{\alpha}}$, then the clusters obtained would exactly match up with the clusters output by the clustering algorithm applied to the circuit $C$, and each cluster of $C|_{\bar{z}= \bar{\alpha}}$ would be obtained by the same restriction procedure being applied to the corresponding cluster of $C$.

\paragraph{Obtaining access to evaluations of the clusters at random inputs:}
 Notice that though we do not know what $C$ is, we can know what $C|_{\bar{z}= \bar{\alpha}}$ is. This is because $C|_{\bar{z}= \bar{\alpha}}$ has only about $k^k$ variables and hence is a low degree multilinear $\Sps(k)$ circuit. Hence we can reconstruct it. We have to be a bit careful here since our reconstruction algorithm might not output the precise circuit $C|_{\bar{z}= \bar{\alpha}}$ but some other multilinear $\Sps(k)$ circuit representation of the same polynomial, call it $C'|_{\bar{z}= \bar{\alpha}}$. However the clustering procedure turns out to be robust enough that the clusters of $C|_{\bar{z}= \bar{\alpha}}$ and the clusters of $C'|_{\bar{z}= \bar{\alpha}}$ match up to compute the same polynomials. Hence we can essentially assume that we know what $C|_{\bar{z}= \bar{\alpha}}$ is and hence we can cluster its gates as well. 
By the properties of clustering, the clusters of $C|_{\bar{z}= \bar{\alpha}}$ match up with the clusters of $C$ (after we set the $\bar{z}= \bar{\alpha}$). Thus though we do not as yet have black-box access to the clusters of $C$, we can indeed recover what the clusters look like after setting $\bar{z}= \bar{\alpha}$. Thus if $C'_1, C'_2, \ldots C'_r$ are circuits corresponding to the clusters of $C$, then we can recover their restrictions to $\bar{z}= \bar{\alpha}$. 
Notice that $\alpha$ was any random sample from $\F^m$, where $m$ is the number of $Z$ variables. Thus we can essentially recover black-box evaluations of the clusters at {\em randomly chosen inputs}. If we could do the same for the $Z$ variables being set to any arbitrary adversarially chosen $\beta \in \F^m$ then we would be done. 

There is one issue we have swept under the rug, which is the following. The clusters of $C|_{\bar{z}= \bar{\alpha}}$ match up with the clusters of $C$, but {\em we don't know what this matching is}. In particular, we might be able to learn  $C|_{\bar{z}= \bar{\alpha}}$ as well as $C|_{\bar{z}= \bar{\alpha}'}$ for two distinct $\bar{\alpha}, \bar{\alpha}' \in \F^m$, and we might be able to cluster both of them, and these clusters correspond to the clusters of $C$, but since we don't know the correspondence we cannot really say that
we know the value of $C'_i|{\bar{z}= \bar{\alpha}}$ as well as $C'_i|_{\bar{z}= \bar{\alpha}'}$ for the same $C'_i$. We will refer to this as ``ambiguity issue".

\paragraph{Obtaining the corresponding between two clusterings:} We now address the ambiguity issue. Suppose we know what $C'_i|_{\bar{z}= \bar{\alpha}}$ looks like. We would like to be able to compute $C'_i|_{\bar{z}= \bar{\alpha}'}$ for any other randomly chosen $\alpha' \in \F^m$. Note that we can reconstruct $C|{\bar{z}= \bar{\alpha}'}$ and cluster it and that would give us the set $\{C'_1|_{\bar{z}= \bar{\alpha}'}, C'_1|_{\bar{z}= \bar{\alpha}'}, \ldots C'_r|_{\bar{z}= \bar{\alpha}'}\}$, but we may not know which element of the set corresponds to $C'_i|_{\bar{z}= \bar{\alpha}'}$. In order to do this identification, we first show how to do this when $\bar{\alpha}$ and $\bar{\alpha}'$ differ in only one coordinate, and then we use a hybrid argument to stitch it together for general $\bar{\alpha}$ and $\bar{\alpha}'$ (by considering a sequence of different $\alpha$s going from $\alpha$ to $\alpha'$ and with consecutive elements differing in one coordinate). When $\bar{\alpha}$ and $\bar{\alpha}'$ differ in only one coordinate, we observe that $C'_i|_{\bar{z}= \bar{\alpha}}$ and $C'_i|_{\bar{z}= \bar{\alpha}'}$ are very similar or very ``near each other" in a suitably defined metric. Then using the robustness property of the clustering we show that the identification of $C'_i|_{\bar{z}= \bar{\alpha}'}$ can be done. 

\paragraph{From evaluations at random points to evaluations at worst case points:}   Let $C'_i$ be the circuit corresponding to cluster $C_i$.
Let us assume we know how to compute $C'_i|_{\bar{z}= \bar{\alpha}}$ for any randomly chosen $\alpha \in \F^m$.  Now let $\bar{\beta}$ be some arbitrary point in $\F^m$. We would like to compute $C'_i|_{\bar{z}= \bar{\beta}}$. 
We use Reed-Solomon decoding for this. We consider the line $t \cdot \bar{\alpha} + (1-t) \cdot \bar{\beta}$ passing through $\bar{\alpha}$ and $\bar{\beta}$ in $\F^m$. In order to learn $C'_i|_{\bar{z}= \bar{\beta}}$, we will learn the restriction of $C'_i$ to the full line, which is a polynomial in the $Y$ variables and the additional $t$ variable. Then setting $t= 0$ would give us the value at $\beta$. To learn the restriction to the line, it suffices to learn the restriction on at least $d+1$ points on the line, where $d$ is the degree of the $t$ variable. By evaluating at $d+1$ random points (which can be done since these points look random) on the line, we can accomplish this.

\section{Notations and Preliminaries}

Throughout the paper, we use $X,Y$ uppercase  denote a set of variables, lowercase $x_i$ denotes variables and  $\bar{x}, \bar{y} $ to denote vector/tuple of variables and $\bar{v}$  denotes a  vector/tuple of field constants. We sometimes abuse notations by referring to a circuit as a collection of multiplication $\Sigma\Pi$ gates.  For any circuit $\pow(k)$ or $\SMLSps(k)$ or multilinear $ \Sps(k)$, we say that circuit is optimal circuit computing a particular polynomial(say $f$) if no circuit (in that respective class) can compute $f$ with a smaller fan-in.

\subsection{Algebraic Tool Kit}

\label{sec:prelim}

Let $\F$ denote a field, finite or otherwise, and let $\cF$ denote
its algebraic closure.  

\subsection{Polynomials}

A polynomial $f \in \Fn$ \emph{depends} on a variable
$x_i$ if there are two inputs $\bar\alpha, \bar\beta \in \bar{\F}^n$
differing only in the $i^{th}$ coordinate for which
$f(\bar\alpha) \neq f(\bar\beta)$. Equivalently, 
$f$ \emph{depends} on a variable
$x_i$ if there is a monomial in $f$ which contains $x_i$.

We denote by $\var(f)$ the set of
variables that $f$ depends on. We say that $f$ is $g$ are \emph{similar} and denote by it $f \sim g$ if $f = \alpha g$ for some $\alpha \neq 0 \in \F$.

For a polynomial $f(x_1,\ldots,x_n)$, a variable $x_i$ and
a field element $\alpha$, we denote with $f \restrict{x_i = \alpha}$ the polynomial resulting from substituting $\alpha$ to
$x_i$. Similarly given a subset $I\subseteq [n]$ and an
assignment $\bar{a}$ $\in \F^{n}$, we define $f \restrict{\xb_I = \bar{a}_I}$ to be the polynomial resulting from substituting $a_i$
to $x_i$ for every $i \in I$.

Let $f,g \in \Fn$ be polynomials. We say that $g$ \emph{divides} $f$, or 
equivalently $g$ is a factor of $f$, and denote it by $g \divs f$
if there exists a polynomial $h \in \Fn$ such that $f = g \cdot h$.
We say that $f$ is \emph{irreducible} if $f$ is non-constant
and cannot be written as a product of two non-constant polynomials.

Given the notion of divisibility, we define the gcd of a set of polynomials in the natural way: we define it to be the highest degree polynomial dividing them all (suitably scaled)\footnote{Such a polynomial is unique up to scaling, and one can fix a canonical polynomial in this class for instance by requiring that the leading monomial has coefficient 1. With this definition, two polynomials are pairwise coprime if their gcd is of degree $0$, and in particular the gcd equals $1$.}. A \emph{linear function} is a polynomial of the form $L(X) = \sum \limits _{i=1} ^{n} a_ix_{i}+ a_0$ with
$a_i \in \F$. The following folklore lemma expresses a condition 
for two non-similar linear functions to remain non-similar under a (partial) substitution. 

\begin{lemma}[Folklore]
\label{lem:linear nsim}
Let $L(X) = \sum \limits _{i=1} ^{n} a_ix_{i}+ a_0$ and
$R(X) = \sum \limits _{i=1} ^{n} b_ix_{i}+ b_0$ be two linear functions in $\Fn$ such that $L \nsim R$. Let $$D(L,R)(X) \eqdef \prod \limits_{i=1}^n \left( a_i R(X) - b_i L(X) \right) \text{, where the product is taken only over non-zero elements.}$$
Let $\ub \in \F^n$ such that $D(L,R)(\ub) \neq 0$. Then for every $I \subsetneq [n]$ it holds that $L \restrict{\xb_I = \bar{u}_I} \nsim R \restrict{\xb_I = \bar{u}_I}$.
\end{lemma}

The interested reader can refer to $\cite{SarafVolkovich18}$ for a proof.

\begin{definition}[Hybrids \& Lines]
\label{def:line}
Let $\ab, \bb \in \F^n$ and $0 \leq i \leq n$.
We define the \emph{$i$-th hybrid of $\ab, \bb$} 
as $\Hybrid^i (\ab,\bb) \eqdef (b_1, \ldots, b_i, a_{i+1}, \ldots, a_n)$.
In particular, $\Hybrid^0 (\ab,\bb) = \ab$ and $\Hybrid^n (\ab,\bb) = \bb$. \\
We define a \emph{line} passing through $\ab$ and $\bb$ as $\ell_{\ab,\bb} : \F \to \F^n$,
$\ell_{\ab,\bb}(t) \eqdef (1-t) \cdot \ab + t \cdot \bb$. In particular, $\ell_{\ab,\bb}(0) = \ab$
and $\ell_{\ab,\bb}(1) = \bb$. \\
\end{definition}



We state below a well known result by Berlekamp and Welch which gives an efficient algorithm for noisy polynomial interpolation. 

\begin{lemma}[Berlekamp-Welch Algorithm (for a description see \cite{Sudan98})]
\label{lem:RS}
Let $P(t)$ be a univariate polynomial of degree at most $d$.
There exists a deterministic algorithm that given $m$ evaluations of $P$ with at most $e$ errors outputs $P(t)$, provided that $m - d > 2e+1$. 
\end{lemma}

For two vectors $\ab$ and $\bb \in \F^n$, let $\wh(\ab,\bb)$ denote the Hamming distance between $\ab$ and $\bb$

\subsection{Partial Derivatives}

The concept of a \emph{partial derivative} of a multivariate function and its properties are well-known and well-studied for
continuous domains (such as, $\R$, $\C$ etc.). 
This concept can be extended to polynomials and rational functions over arbitrary fields from a purely algebraic point of view.
For more details we refer to reader to \cite{Kaplansky57}.

\begin{definition}
\label{def: partial derivative}
For a monomial $M = \alpha \cdot x_1^{e_1}  \cdots x_i^{e_i} \cdots x_n^{e_n} \in \Fn$
and a variable $x_i$ we define the \emph{partial derivative} of $M$ with respect to $x_i$, 
as $\Fpar{M}{i} \eqdef \alpha e_i \cdot x_1^{e_1} \cdots x_i^{e_i-1} \cdots x_n^{e_n}$. 
The definition can be extended to $\Fn$ by imposing linearity and to $\RFn$ via the quotient rule.
\end{definition}

Observe that the sum, product, quotient and chain rules carry over. 
In addition, when $\F = \R$ or $\F = \C$ the definition coincides with the analytical one.
The following set of rational function plays an important role.

Inspired by a similar notion of \cite{KarninShpilka09},
we define a distance measure between multiplication gates. 
This measure will play a crucial role in the analysis of our reconstruction algorithm. Roughly speaking, the distance between two polynomials, each of them being product of linear forms, is the number of factors that appear in only one of them. 

\begin{definition}[Distance]
\label{def:distance}
For $f,g \in \Fn$, we define a \emph{distance} function: $$\Delta(f,g) \eqdef \frac{ \max  \set{ \deg(f), \deg(g) }} {\deg(\gcd(f,g))}.$$
\end{definition}

\subsection{Depth-3 Circuits}
\label{sec:depth-3}

In this section we formally introduce the general model of depth-$3$ circuits and specialization of set-multilinear depth-$3$ circuits, which is the focus of our paper. It is to be noted that depth-$3$ circuits were a subject for a long line of study \cite{DvirShpilka07,KayalSaxena07,KayalSaraf09,ShpilkaVolkovich15a,ArvindM10,KarninShpilka11, SaxenaSeshadhri11a,SaxenaSeshadhri12, SaxenaSeshadhri13}.

\begin{definition}

A depth-$3$ $\Sps(k)$ circuit $C$ computes a polynomial of the form
$$C(X) = \sum\limits _{i = 1}^k T_i(X)= \sum_{i= 1}^k
\prod_{j=1}^{d_i}\ell_{i,j}(X),$$
where the $\ell_{i,j}$-s are linear functions;
$\ell_{i,j}(X)= \sum \limits _{t=1} ^{n} a^t_{i,j}x_{t}+ a^0_{i,j}$ with
$a^t_{i,j} \in \F$.
 \\
A \emph{multilinear} $\Sps(k)$
circuit is a $\Sps(k)$ circuit in which each $T_i$ is a
multilinear polynomial. In particular, each such $T_i$ is a product of variable-disjoint linear functions. \\
Given a partition $X = \sqcup _{j \in [d]} X_j$ of $X$, a \emph{set-multilinear} $\SMLSps(k)$ circuit is a further specialization of a multilinear circuit to the case when each $\ell_{i,j}$ is a linear \emph{form} in $\F[X_j]$. That is, each $\ell_{i,j}$ is defined over the variables in $X_j$ and $a^0_{i,j} = 0$.

\bigskip  \noindent We say that $C$ is \emph{minimal} if no subset of the multiplication
gates sums to zero. We define $\gcd(C)$ as the linear product of all the non-constant linear functions that belong to all the $T_i$-s. I.e. $\gcd(C) = \gcd(T_1, \ldots, T_k)$. We say that $C$ is \emph{simple} if $\gcd(C)=1$. The simplification of $C$, denoted by $\simp(C)$, is defined as $C / \gcd(C)$. In other words, the circuit resulting upon the removal of all the linear functions that appears in $\gcd(C)$.
Finally, we say that a $\SMLSps$ circuit 
has \emph{width} $w$, if $\size{X_j} \leq w$ for all $j$. 
\end{definition}

Throughout the paper, we will be referring to this quantity as the \emph{width} of a polynomial, width of a circuit, since our model is is  $\SMLSps$ circuits, it all essentially means the same.   

\subsubsection{Existing Algorithms}

We require the following results. In what follows we focus on multilinear and set-multilinear circuits. We begin with polynomial identity testing algorithms.
We will first state the well-known Schwartz-Zippel lemma followed by a deterministic black-box identity testing algorithm for multilinear depth-$3$ circuits. These algorithms will be used in the testing phase.  A black-box PIT is an algorithm that tests if a given circuit computes the zero polynomial by only evaluating the circuit on points, and not inspecting the internal structure of the circuit. Hence all that a black-box PIT can do is evaluate the circuit
on a small list of points which is guaranteed to have a property that every non-zero circuit produces at least one non-zero evaluation in the list. Such lists are also called hitting sets, the black-box PITs are also called hitting set generators. All the black-box PIT results discussed below can also be interpreted as existence of explicit hitting sets, these hitting sets will  be used in derandomizing our learning algorithms.

\begin{lemma}\cite{Schwartz80,Zippel79, DemilloLipton78} \label{SchwartzZippel}
Let $f (x_1, . . . , x_n )$ be a nonzero polynomial of degree at most $d$, and let $S \subseteq \F$. If we choose $\bar{a} = (a_1 , \ldots , a_n) \in S^n$ uniformly at random, then $Pr[f (\bar{a}) = 0] \leq d/|S |$.

\end{lemma}

\begin{lemma}[\cite{SaxenaSeshadhri12,ShpilkaVolkovich15a}]
\label{lem:pit sps}
There is a deterministic algorithm that given a black-box access to a multilinear $\Sps(k)$ circuit $C$ decides if $C \equiv 0$, in time $n^{\BigO(k)}$.
\end{lemma}

The next result provides a factorization algorithm for multilinear depth-$3$ circuits.
A crucial observation is that factors of a multilinear polynomial must be variable-disjoint. Therefore, each factor of a multilinear polynomial $P$ is obtained by restricting $P$ to an appropriate subset of variables.

\begin{lemma}[\cite{ShpilkaVolkovich10}]
\label{lem:fact sps}
There is a deterministic algorithm that given a black-box access to a multilinear/set-multilinear $\Sps(k)$ circuit $C$, outputs black-boxes for the irreducible factors of $C$, in time $n^{\BigO(k)}$. In addition, each such irreducible factor is computable by a multilinear/set-multilinear $\Sps(k)$ circuit.
\end{lemma}

As a corollary, we can efficiently simulate a black-box access to $\simp(C)$ given a black-box access to $C$. The main observation is that a linear function that appears in $\gcd(C)$ constitutes an irreducible factor of $C$.

\begin{corollary}
\label{cor:simple}
There is a deterministic algorithm that given a black-box access to a multilinear/set-multilinear $\Sps(k)$ circuit $C$ outputs linear functions $L_1, \ldots, L_r$ and black-box access to a simple multilinear/set-multilinear $\Sps(k)$ circuit $\hat{C}$  such that
$C = \prod _{i=1}^r L_i \cdot \hat{C}$, in time $n^{\BigO(k)}$.
\end{corollary}

\begin{proof}
We describe the following algorithm:

\begin{itemize}
    \item Run the algorithm from Lemma \ref{lem:fact sps} to obtain black-boxes $C_1, \ldots, C_m$ for the irreducible factors of $C$.
    
    \item For each $i$, try to learn $C_i$ as a linear function by evaluating it on the standard base vectors and $\bar{0}$. Let $L_i$ denote the resulting purported linear function. 
    
    \item As each $C_i$ is computable by a multilinear $\Sps(k)$ circuit,  use the identity testing algorithm from Lemma \ref{lem:pit sps} on $C_i - L_i$ to determine which $C_i$-s compute linear functions.
    
    \item Wlog, let $C_1, \ldots, C_r$ be the irreducible factors that correspond to linear functions.  Set $V \eqdef \bigcup \limits_{i=1}^r \var(L_i)$;
    
    \item Use Lemma \ref{lem:pit sps} to find an assignment $\ab \in \F^X$ such that $C(\ab) \neq 0$.
    
    \item Output:  $L_1, \ldots, L_r,  \hat{C} \eqdef \frac{C\restrict{X_V=\ab_V}}  {\prod_{i=1}^r L_i(\ab_V)}$
\end{itemize}

\noindent The claim regarding the runtime follows from Lemmas \ref{lem:pit sps} and \ref{lem:fact sps}. For the analysis, observe that $L_1, \ldots, L_r$ are factors of $C$.
In addition, as $C$ computes a multilinear polynomial, its factors are variable-disjoint. Therefore, we can write $$C(V,X \setminus V) = \prod_{i=1}^r L_i \cdot C'( X \setminus V).$$
Consequently:
$$C(\ab_V,X \setminus V) = \prod_{i=1}^r L_i(\ab_V) \cdot C'( X \setminus V).$$
and
$$C(X \setminus V) = \prod_{i=1}^r L_i \cdot \frac{C(\ab_V,X \setminus V)} {\prod_{i=1}^r L_i(\ab_V)} = \prod_{i=1}^r L_i \cdot \hat{C}.$$
Note that $ \prod_{i=1}^r L_i(\ab_V) \neq 0$
as $C(\ab) \neq 0$. Finally, since every factors of $\hat{C}$ constitutes a factor of $C$, and all the linear factors of $C$ has been accounted for it follows that $\hat{C}$ has no linear factors.



\end{proof}

\subsubsection{Structural Results}

In this we discuss a strong structural result about set-multilinear depth-$3$ circuits computing the zero polynomial. 
We note that results of this flavor were  proven before for more general families of depth-$3$ circuits (for more details see e.g. \cite{DvirShpilka07,KayalSaraf09,SaxenaSeshadhri13} and references within).
We prove our result by a reduction to the case where each linear function is, in fact, a univariate polynomial.

\begin{lemma}[\cite{AvMV15}]
\label{lem:uni rank bound}
Let $k \geq 2$ and let $C \equiv \sum \limits_{i=1} ^k T_i = \sum \limits_{i=1} ^k \prod \limits_{j=1} ^{d_i}\ell_{i,j}$ be a simple and minimal multilinear circuit $\Sps(k)$ circuit where each $\ell_{i,j}$ is a univariate polynomial. If $C$ computes the zero polynomial then for all $i \in [k]: \size{\var(T_i)} \leq k-2$.
\end{lemma}

\begin{theorem}
\label{thm:rank bound}
Let $C \equiv \sum \limits _{i=1}^k T_i$ be a simple and minimal $\SMLSps$ circuit computing the zero polynomial. Then for all $i \in [k]: \deg(T_i) \leq k-2$.
\end{theorem}

\begin{proof}
Fix $i\in [k]$. Recall that $T_i$ is of the form  $T_i =  \prod\limits _{j=1}^d \limits\ell_{i,j}(X_j)$. Fix $j \in [d]$. Pick a variable $x_j \in \var(\ell_{i,j})$ and let $\bar{a}_j \in \F^{X_j \setminus \set{x_j}}$ be a random assignment to the variables $X_j \setminus \set{x_j}$.   As $C$ is simple, $\gcd(\ell_{1,j}, \ldots, \ell_{k,j}) = 1$. By the choice of $\bar{a}_j$ we obtain that $\gcd \left( \ell_{1,j}(\ab_j,x_j), \ldots, \ell_{k,j}(\ab_j,x_j) \right) = 1$.  Now consider the circuit $\hat{C} \equiv \sum \limits _{i=1}^k \hat{T}_i$ obtained from $C$ be assigning each set $X_j \setminus \set{x_j}$ to $\ab_j$. Observe that $\hat{C}$ satisfies the premises of Lemma \ref{lem:uni rank bound}. Moreover, as $\hat{T_i}$ is a product of univariate polynomials, $\deg(T_i) = \deg(\hat{T_i}) = \size{\var(\hat{T_i})}$. Therefore, $\deg(T_i)= \size{\var(\hat{T_i})} \leq k-2$, as required.
\end{proof}

This structural result, in turn, implies that the distance (see Definition \ref{def:distance}) between two multiplication gates in a minimal circuit, computing the zero polynomial, is ``small''. 

\begin{lemma}
\label{lem:rank bound}
Let $C \equiv \sum \limits _{i=1}^k T_i$ be a minimal $\SMLSps$ circuit computing the zero polynomial. Then for all $i \neq j \in [k]: \Delta(T_i,T_j) \leq k-2$.
\end{lemma}

\begin{proof}
Consider $\simp(C)$. By definition, $\simp(C)$ is simple. In addition, observe that $\simp(C)$ minimal and computes the zero polynomial, by construction. By Theorem \ref{thm:rank bound}, 
for each $i \in [k]$ we have $$\frac{\deg(T_i)} {\deg(\gcd(C))} = \deg \left( \frac{T_i}{\gcd(C)} \right) \leq k-2.$$
Therefore,
$$\Delta(T_i,T_j) = \frac{ \max  \set{ \deg(T_i), \deg(T_j) }} {\deg(\gcd(T_i,T_j))} \leq \frac{(k-2) \cdot \deg(\gcd(C))}{\deg(\gcd(T_i,T_j))} \leq k-2.$$
The last inequality follows from the fact that $\gcd(T_i,T_j)$ divides $\gcd(C)$.
\end{proof}

The above result implies that any pair of circuits computing the same polynomials must have ``many'' common linear functions. These results are also refereed as rank-bounds in the literature. Also, for our applications we don't need $\Delta(T_i,T_j) \leq k-2$ at this granular detail, we will simply upper bound this by $k$. 

\begin{lemma}
\label{lem:unq}
$C \equiv T_1 + T_2 \ldots T_k $ 
and $C' \equiv T'_1 + T'_2 \ldots T'_{k'}$
be two $\SMLSps$ circuits computing the same polynomial with $k' \leq k$.
Furthermore, suppose $C$ is minimal. Then for each $T_i$ (in $C$) there exists a $T'_j$ (in $C'$)  such that $\Delta(T_i, T'_j) \leq k$. 
\end{lemma}

\begin{proof}
Notice that,  $C-C'  \equiv 0$, that is  $T_1 + T_2 +  \ldots + T_k - T'_1 - T'_2 - \ldots - T_{k'} = 0$. Pick $i \in [k]$ and let $C_i$ be a minimal subcircuit computing the zero polynomial that contains $T_i$.
As $C$ is a minimal circuit, $C_i$ must contain at least one of $T'_j$-s.
The result now follows directly from Theorem~\ref{thm:rank bound}.
\end{proof}

\subsubsection*{Other useful lemmas}
\begin{definition}
Let $\textbf{f}:=(f_1, f_2, \ldots f_m)$, where $f_i(X) \in \F[\textbf{X}]$, be a vector of polynomials over a field $\F$. The set of $\F$-linear dependencies in $\textbf{f}$ , denoted $\textbf{f}^{ \,\perp}$, is the set of all vectors $\textbf{v} \in \F^m$ whose inner product with $\textbf{f}$ is the zero polynomial, i.e., $$\textbf{f}^{ \,\perp} \{ (a_1, \ldots, a_m) \in \F^m: a_1 f_1(\textbf{X}) + \ldots + a_m f_m(\textbf{X})=0\}.$$
\end{definition}

The set $\textbf{f}^{ \,\perp}$ is clearly a linear subspace of $\F^m$. This notion is helpful to state and prove some useful lemmas.  The main observation  here is that given (by arithmetic circuits or black-box access) a collection of polynomials then in randomized polynomial time we can find the $\F$-linear dependencies among these polynomials. In order to show that we will need the following  technical lemma.

\begin{lemma}\cite[Lem 4.1]{Kayal11} \label{invert1}
Let $f_1, f_2, \ldots f_k \in \F[x_1,\ldots, x_n]$ be $\F$-linearly independent polynomials with $|\F| > nkd$ (otherwise  we can work with an extension) and  $a_1, \ldots , a_k \in \F^n$ be $k$ random points\footnote{More precisely, let $S\subseteq \F$ be a set pf size $nkd$ and let each $a_i$ be chosen independently and uniformly at random from $S^n$}  then the following matrix has full rank with high probability.  \[M=\begin{pmatrix}
f_1(a_1) & f_2(a_1) & \cdots & f_k(a_1)\\
f_1(a_2) & f_2(a_2) & \cdots & f_k(a_2)\\
\vdots & \vdots  & \vdots & \vdots \\
f_1(a_k) & f_2(a_k) & \cdots & f_k(a_k)\\ 
\end{pmatrix}\]
\end{lemma}

\begin{proof}
Let $X_1, \ldots, X_k$ be   disjoint sets of variables each of size $n$.    Define,\begin{equation}\label{eq:Q}
 Q=\begin{pmatrix}
f_1(X_1) & f_2(X_1) & \cdots & f_k(X_1)\\
f_1(X_2) & f_2(X_2) & \cdots & f_k(X_2)\\
\vdots & \vdots  & \vdots & \vdots \\
f_1(X_k) & f_2(X_k) & \cdots & f_k(X_k)\\ 
\end{pmatrix}.
\end{equation} 
We will show, via induction on $k$, that $Q$ has full rank, or equivalently, the determinant of $Q$ is a nonzero polynomial. The result will then follow via the Schwartz-Zippel Lemma applied to the determinant of $Q$.

Note that $k=1$ follows directly.  On expanding $Det(Q)$ along the first row we get,
$$ Det(Q)=\sum_{j=1}^{k}(-1)^{j+1}f_j(X_1)Q_{1j}$$ where $Q_{ij}$ is the determinant of the $ij$-th minor. Notice that every $Q_{1j}$ , $j \in [k]$, is a polynomial in the set of variables $X_2 , . . . , X_k$. By induction, every $Q_{1j}$ is a nonzero polynomial (since every subset of a set of $\F$-linearly independent polynomials is also $\F$-linearly independent).  If $Det(Q)$ was the zero polynomial then plugging in random values for $X_2 , . . . , X_k$ would give us a nonzero $\F$-linear dependence among $f_1(X_1), f_2 (X_1 ), . . . , f_k (X_1)$, which is a contradiction. Hence $Det(Q)$
must be nonzero, proving the claim. This along with Lemma~\ref{SchwartzZippel}  gives that $M$ is invertible with high probability. 
\end{proof}

Once we have the above lemma, one can easily use it to determine the linear dependency structure of a set of polynomials as in the next lemma. 

\begin{lemma}\cite[Lem 4.1]{Kayal11} \label{invert}
Given  $m$  polynomials $\textbf{f} =\{f_1, f_2, \ldots f_m\}$, each  in $\F[x_1, \ldots, x_n]$  of degree at most $d$, either by a circuit(or black-box access) \footnote{\label{field} The lemma statement in \cite{Kayal11} just mentions the case when a circuit is given explicitly, however it is easy to observe that even black-box/oracle access suffices.},  s.t.  with rank(maximal number of linearly independent $f_i$-s) of $\textbf{f} = k$, and $|\F|> {m \choose k} \cdot dnk $ (if $|\F| \leq {m \choose k} \cdot dnk $ then we can work with an extension) then: 
\begin{enumerate}
\item  There is a randomized $\poly(m,n,k)$ time algorithm to compute a basis for the space $\F$-span$\{f_1, f_2, \ldots f_m\}$. Along with the basis(say  $f_{i_1}, f_{i_2}, \ldots f_{i_k}$ be a basis of linear space of $f_i$-s),  the aforementioned algorithm also outputs a matrix $M$ s.t. 
$$M\begin{pmatrix}
f_{i_1}\\  \vdots\\ f_{i_k}
\end{pmatrix}= \begin{pmatrix}
f_1\\ \vdots  \\ f_m
\end{pmatrix}.$$
\item Also, there is a randomized $\poly(m,n,k)$ time algorithm that given a vector of these polynomials  $\textbf{f}=(f_1, f_2, \ldots, f_m)$ computes a basis for the space $\textbf{f}^{ \,\perp}$. 
\end{enumerate}
\end{lemma}

\proofsketch
Let $X_1, \ldots, X_k$ be   disjoint sets of variables each of size $n$.    Define,
\[ Q=\begin{pmatrix}
f_1(X_1) & f_2(X_1) & \cdots & \cdots & f_m( X_1)\\
f_1(X_2) & f_2(X_2) & \cdots & \cdots & f_m(X_2)\\
\vdots & \vdots  & \vdots & \vdots & \vdots \\
f_1(X_k) & f_2(X_k) & \cdots & \cdots & f_m(X_k)\\ 
\end{pmatrix}.\] 

The crucial observation here is that linear dependencies of the $f_i$ are exactly captured by linear dependencies of the columns of $Q$ and moreover this continues to hold after substituting random values to the $X_i$-s. We will show in the next paragraph how to prove this. Note that once we have this fact then we have reduced the problem of determining the linear dependencies that hold between the polynomials to determining the linear dependencies that hold between vectors in $\F^m$, and for vectors in $\F^m$ we do know efficient algorithms for computing the basis, the orthogonal subspace and the suitable matrix $M$. 

In order to see that linear dependencies between the $f_i$ are captured by linear dependencies among the columns of $Q$ after the random substitutions, it suffices to show that for any size $k$ subset of $f_i$-s which is linearly independent, the corresponding minor of $Q$ has full rank. Note that there can be at most $m \choose k$ such full rank minors, and we have to ensure that the determinant of each full rank minor stays nonzero, which in-turn boils down to ``hitting" (finding a non-zero assignment) the product of these determinants. Note that the degree of the product of such determinants is bounded by ${m \choose k} \cdot dnk$.
Thus, by Lemma \ref{SchwartzZippel} we get that  random  substitutions (given $|\F|> {m \choose k} \cdot dnk $)  ensure that with high probability, for any subset of $f_i$-s which are linearly independent, the corresponding minor of $Q$ has full rank.  \qed \\


Interestingly, when $f_i$-s are from special classes of polynomials for which deterministic blackbox PIT algorithms (explicit hitting sets) are known, then we can derandomize the previous lemma. Concretely,  if $f_1, \ldots, f_m$ have rank $k=\BigO(1)$ with each $f_i \in \mc{C}$, where $\mc{C}$ is an arithmetic circuit class. Then the randomized algorithms in Lemma \ref{invert} and Lemma \ref{invert1} can be derandomized given polynomial sized hitting sets for the  class $\underbrace{\mc{C} + \ldots +  \mc{C}  }_{k \, \, times}$. Here, $\underbrace{\mc{C} + \ldots +  \mc{C}  }_{k \, \, times}$ is a circuit class which comprises of $\F$-linear combinations of $k$ polynomials in $\mc{C}$.  Also, for our applications $\mc{C}$ will either be $\pow(k')$ circuits, $\SMLSps(k')$ circuits or multilinear $\Sps(k')$ circuits with $k=\BigO(1)$ and $k'=\BigO(1)$, and thus we do have such hitting sets.

We will start by the stating deterministic version of Lemma \ref{invert1}.

\begin{lemma} \label{derand-invert}
Let $f_1, f_2, \ldots f_k \in \F[x_1,\ldots, x_n]$ be $\F$-linearly independent polynomials with $|\F| > nkd$ (otherwise  we can work with an extension) and $f_i \in \mc{C}$, where $\mc{C}$ is an arithmetic circuit class. Furthermore, let $\mathcal{H}_k$ be a hitting set for the  class $\underbrace{\mc{C} + \ldots +  \mc{C}  }_{k \, \, times}$. Then there exist $(a_1, a_2, \ldots, a_k)$ with each $a_i\in \mc{H}_k$  s.t.  the following matrix has full rank.  \[M=\begin{pmatrix}
f_1(a_1) & f_2(a_1) & \cdots & f_k(a_1)\\
f_1(a_2) & f_2(a_2) & \cdots & f_k(a_2)\\
\vdots & \vdots  & \vdots & \vdots \\
f_1(a_k) & f_2(a_k) & \cdots & f_k(a_k)\\ 
\end{pmatrix}\] 
Equivalently, $\underbrace{\mc{H}_k \oplus \mc{H}_k\cdots \oplus \mc{H}_k}_{k \, \, times}$ is a hitting set for $Det(Q)$, where $Q$ is defined by Eq. \ref{eq:Q}.  
\end{lemma}
\begin{proof}
This follows  by essentially the same inductive argument as in the proof of Lemma \ref{invert1}.  The case $k=1$ follows directly. For $k>1$, on expanding $Det(Q)$ along the first row we get,
$ Det(Q)=\sum_{j=1}^{k}(-1)^{j+1}f_j(X_1)Q_{1j}$, where $Q_{ij}$ is the determinant of the $ij$-th minor. By induction, every $Q_{1j}$ is a nonzero polynomial and 
we can deterministically choose $a_2, \ldots, a_n$ using  hitting sets for the  class $\underbrace{\mc{C} + \ldots +  \mc{C}  }_{k-1 \, \, times}$. For brevity we will refer to the substitution $X_2=a_2, \ldots, X_n=a_n$ by $\sigma$.  Thus to find $a_1$ s.t. $M$ is invertible, we have to do identity testing for $ Det(Q)|_{\sigma}=\sum_{j=1}^{k}(-1)^{j+1}f_j(X_1)Q_{1j}|_{\sigma}$ which lies in class $\underbrace{\mc{C} + \ldots +  \mc{C}  }_{k \, \, times}$. This can be done by choosing $a_1$ from a hitting set for the  class $\underbrace{\mc{C} + \ldots +  \mc{C}  }_{k \, \, times}$. 
\end{proof}

Using the above derandomized lemma, we show how to derandomize Lemma \ref{invert} in polynomial time when the number of independent $f_i$-s is constant.

\begin{lemma}\label{derand-invert1}
Given  $m$  polynomials $\textbf{f} =\{f_1, f_2, \ldots f_m\}$, each  in $\F[x_1, \ldots, x_n]$  of degree at most $d$, either by a circuit(or black-box access)   s.t.   rank(maximal number of linearly independent tuples) of $\textbf{f} = k$, and  $f_i \in \mc{C}$, where $\mc{C}$ is an arithmetic circuit class.  Also, let $\mc{H}_k$ be a hitting set for the class $\underbrace{\mc{C} + \ldots +  \mc{C}  }_{k \, \, times}$. Then, 
\begin{enumerate}
\item  There is a deterministic $\poly(|\mc{H}_k| {m \choose k}, n, d,m)$ time algorithm to compute a basis for the space $\F$-span$\{f_1, f_2, \ldots f_m\}$. Along with the basis(say  $f_{i_1}, f_{i_2}, \ldots f_{i_k}$ be a basis of linear space of $f_i$-s),  the aforementioned algorithm also outputs a matrix $M$ s.t. 
$$M\begin{pmatrix}
f_{i_1}\\  \vdots\\ f_{i_c}
\end{pmatrix}= \begin{pmatrix}
f_1\\ \vdots  \\ f_m
\end{pmatrix}.$$
\item Also, there is a deterministic $\poly(|\mc{H}_k| \cdot{m \choose k}, n, d,m)$ time algorithm that given a vector of these polynomials  $\textbf{f}=(f_1, f_2, \ldots, f_m)$ computes a basis for the space $\textbf{f}^{ \,\perp}$. 
\end{enumerate}
\end{lemma}

\proofsketch Let $X_1, \ldots, X_k$ be   disjoint sets of variables each of size $n$.    Define,
\[ Q=\begin{pmatrix}
f_1(X_1) & f_2(X_1) & \cdots & \cdots & f_m( X_1)\\
f_1(X_2) & f_2(X_2) & \cdots & \cdots & f_m(X_2)\\
\vdots & \vdots  & \vdots & \vdots & \vdots \\
f_1(X_k) & f_2(X_k) & \cdots & \cdots & f_m(X_k)\\ 
\end{pmatrix}.\] 

The goal is the derandomize the algorithm from Lemma~\ref{invert}. In order to do this, we have to ensure that for any size $k$ subset of $f_i$-s which are linearly independent, there is a deterministic substitution of the variables the corresponding minor has full rank. Note that there can be at most $m \choose k$ full rank minors and we have to ensure that we find a deterministic substitution of the variables such that the determinant of each  full rank minor stays nonzero, which is equivalent to keeping the product of such determinants nonzero after substitution. 
By Lemma~\ref{derand-invert}, there is a polynomial sized hittting sit for the determinant for each of these full rank minors. Along with standard blackbox PIT trick of working with hitting set ``generators", we can find a hitting set of size  $\poly(|\mc{H}_k| \cdot{m \choose k})$ for the product of the determinants of each full rank minor (see \cite[Sec.~4.1]{SY10} for details).

Once we have the hitting set, we can then choose that element of the hitting set that maximizes the number of $k \times k$ minors whose determinant is nonzero after substitution to find the appropriate substitution such that the columns of $Q$ will have the same linearly dependency structure as the given polynomials. 
\qed

\begin{lemma}
Set-multilinear $\Sps(k)$ circuits are closed under factoring. That is, if $f=g\cdot h$ and $f$ is computed by a \sk \, circuit. Then $g$ (similarly h) is computed by a set-multilinear $\Sps(k)$ circuit. Also, there is a partition of $[d]$ into two disjoint sets $A_1$ and $A_2$ s.t. $g$ is set-multilinear w.r.t to partition $\sqcup_{i \in A_1} X_i$ and $h$ is set-multilinear w.r.t. to partition $\sqcup_{i \in A_2} X_i$. 
\end{lemma}

\begin{proof}

Let $f=g\cdot h$ where has $f$ has  a \sk \, circuit, say $C_f$.

We will first show that $g,h$ are set-multilinear with  two variable disjoint partitions. That is variables from a partition either occur in $g$ or occur in $h$. Formally, for each $X_i$, if a variable $v \in \var(g) \cap X_i $, then $X_i \cup \var(h)= \phi$. Suppose, for contradiction, $\exists X_i$ s.t  $u,v \in X_i$ s.t. $u \in \var(g)$ and $v \in \var(h)$. That is, on writing $g,h$ as univariates in $u, v$ respectively, we get $g=a_du^d+ \ldots + a_o$ and $h=b_{d'}v^{d'}+ \ldots + b_o$ with $d, d' \geq 1$ and $a_d, b_{d'} \neq 0$.  Now, notice that coefficient of $u^dv^{d'} \neq 0$ in $f=g \cdot h$ thus contracting the assumption that $f$ was set-multilinear to begin with. 

The proof concludes by setting all variables in $h$ to random values(s.t. $h$ doesn't evaluate to 0) in $C_f$ and observing that the resulting circuit is a set-multilinear circuit computing  a constant multiple of $g$. 
\end{proof}

\begin{lemma}\label{indep}
Let $C_f \equiv T_1 + T_2 + \cdots+ T_k$ be an optimal degree $d$ $\SMLSps(k)$ circuit computing $f$, where $T_i=\prod_{j \in [d]} \ell_{i,j}(X_j)$. Then  $ \forall \,j \in [d]\, $ the polynomials in the set $ S_j:=\{\frac{T_1}{\ell_{1,j}}, \frac{T_2}{\ell_{2,j}}, \ldots, \frac{T_k}{\ell_{k,j}}\}$ are linearly independent. Note that, $S_j$ is a set of polynomials (not rational functions) because $\ell_{i,j}|T_i \, \forall j \in [k], i \in [d] $. 
\end{lemma}
\begin{proof}
 Assume on contrary that $\exists j$ s.t. the linear forms in $S_j=\{T_1/\ell_{1,j}, T_2/\ell_{2,j}, \ldots T_{k}/\ell_{k,i}\}$  has a linear dependence.  Let $\sum_{i=1}^{k-1} \lambda_i T_i/\ell_{i,j} = T_k/\ell_{k,j} $, be a non-trivial linear dependence. Note that, this can be assumed always by just relabelling the gates. This implies that, \begin{equation}\label{eq:indep}
     f(\xb)= \sum_{i=1}^{k-1} (\ell_{i,j}+ \lambda_i\ell_{k,j})\cdot  T_i/\ell_{i,j}.
 \end{equation} Note that, equation \ref{eq:indep} is a  $\SMLSps(k-1)$ representation of $f$, thus contradicting our assuption that $C_f$ is optimal.
\end{proof}




\subsection{Variable Reduction}
In this section we discuss how to reduce the number of variables in a polynomial. Before describing this procedure we have to formally define the notion of number of essential variables in a polynomial.

\begin{definition}[number of essential variables]
For $f(\bar{x}) \in \F[\bar{x}]$, we will say that the \textit{number of essential variables} in $f(\bar{x})$ is $t$ if there exist an invertible linear transformation $A \in \F^{n \x n}$ s.t. $f(A\bar{x})$ just \textit{depends} on $t$ variables.   
\end{definition}

The next lemma is from the work of Carlini \cite{Carlini06}, adopted by Kayal \cite{Kayal11} in the language of circuits. This lemma eliminates redundant variables from a polynomial and plays a crucial role in our reconstruction results for $\pow$ and multilinear-$\Sps$ circuits.

We state the lemma below in the setting of black-box access to the input polynomial. The original version of the lemma was in the whitebox setting, but by inspecting the proof in \cite{Kayal11} one can see that it works in the black-box setting as well by noting that given black-box access to a circuit computing a polynomial $f$, one can get black-box access to the circuits computing its first order partial derivatives. 

\begin{lemma}\label{lem:carlini-black-box}\cite{Kayal11, Carlini06}
Given black-box access to an $n$-variate polynomial $f(X) \in \F[X]$ of degree $d$ with $m$ essential variables, s.t. $char(\F) > d \mbox{ or } 0$, there is a randomized $\poly(n,d,s)$ time algorithm (where $s$ is the size of the circuit computing $f$) that computes an invertible linear transformation $A \in \F^{(n \times n)}$ such that $f(A \cdot \bar{x})$ depends on the first $m$-variables only. 
\end{lemma}

Interestingly, when $f$ is from a special class of polynomials for which explicit hitting sets are known even for first order partial derivatives of $f$, and additionally if the number of essential variables in $f$ is small,  then we can derandomize the previous lemma. It is worth nothing that many or most interesting classes of circuits are closed under taking partial derivatives.

\begin{lemma}\label{derand-carlini}
Let $\mc{C}$ be  class of arithmetic circuits that is closed under first order partial derivatives. 
We are given black-box access to an $n$-variate polynomial $f(X) \in \F[x_1, \ldots, x_n]$ of degree $d$, computable by a size $s$ circuit in $\mc{C}$, such that  $f$ has $k$ essential variables and  $char(\F) > d \mbox{ or } 0$. Let $\mc{H}_k$ be a hitting set for the class $\underbrace{\mc{C} + \ldots +  \mc{C}  }_{k \, \, times}$. Then there is a deterministic $\poly({n \choose k} ,d,s, |\mc{H}_k|)$ time algorithm  that computes an invertible linear transformation $A \in \F^{(n \times n)}$ such that $f(A \cdot \bar{x})$ depends on the first $k$-variables only. 
\end{lemma}
\proofsketch
The proof of this lemma is obtained by derandomizing the proof of Lemma \ref{lem:carlini-black-box} in the current setting. Since we haven't provided the details of the proof of Lemma \ref{lem:carlini-black-box}, we will only provide a proof sketch here of the changes needed to be made to the proof of Lemma \ref{lem:carlini-black-box} to derandomize it. 
 The only place where randomness is used in the proof of Lemma \ref{lem:carlini-black-box} is in computing a basis of $\mathbf{f}^{\perp}$, where $\mathbf{f}:=\big(\frac{\partial f}{\partial x_1}, \ldots, \frac{\partial f}{\partial x_n}\big)$.
 As in the proof of Lemma \ref{lem:carlini-black-box}, one observes that, rank of $\mathbf{f} =\mbox{number of essential variables } =k$.  
Once we have this, then the assumption that $\mc{C}$ is closed under taking first order partial derivatives, along with  the hitting set $\mc{H}_k$ satisfies all the preconditions for Lemma \ref{derand-invert1}. Thus, we can use  Lemma \ref{derand-invert1} to compute a basis for $\mathbf{f}^{\perp}$ deterministically, which in turn gives a deterministic algorithm for computing $A$ s.t. $f(A \cdot \bar{x})$ depends on the first $k$-variables only. The time complexity follows directly.  \qed

\subsection{Tensors and Set-Multilinear Depth-$3$ Circuits}\label{sec:Equiv-sml}

Tensors,  higher dimensional analogues of matrices, are  multi-dimensional arrays with entries from some field $\F$. For instance, a $3$-dimensional tensor can be written as $\mc{T} = (\alpha_{i,j,k} ) \in \F^{n_1 \times n_2 \times n_3}$ and 2-dimensional tensors simply corresponds to traditional matrices. We will work with general $d$-dimensional tensors $\mc{T}=(\alpha_{j_1,j_2,\ldots, j_d} ) \in \F^{n_1 \times \cdots \times n_d}$, here $[n_1] \times \cdots \times [n_d]$ refers to the shape of the tensor and $n_i$ as length of tensor in $i$-th dimension. Just like any matrix has a natural definition of rank, there is an analogue for tensors as well. 

The rank of a tensor $\mc{T}$ can be defined as the smallest $r$ for which $\mc{T}$ can be written as a sum of $r$ tensors of rank $1$, where a rank-$1$ tensor is a tensor of the form $v_1 \otimes \cdots \otimes v_d$ with $v_i \in \F^{n_i}$. Here $\otimes$ is the Kronecker (outer) product a.k.a \emph{tensor product}. The expression of $\mc{T}$ as a sum of such rank-$1$ tensors, over the field $\F$ is called \emph{$\F$-tensor decomposition} or just tensor decomposition, for short.  The notion of Tensor rank/decomposition has become a fundamental tool in different branches of modern science with applications in  statistics, signal processing, complexity of computation, psychometrics,
linguistics and chemometrics. We refer the reader to a   monograph by  Landsberg \cite{Landsberg12} and the references therein for more details on application of tensor decomposition.

For our application, it would be useful to think of tensors as a restricted form of multilinear polynomials that are called \emph{set-multilinear} polynomials. To this end, let us fix the following notation throughout the paper. \\
Let $d \in \N$. We will refer to $d$ as the \emph{dimension}. For $j \in [d]$ let $X_j = \set{x_{j,1}, x_{j,2}, \ldots, x_{j,n_j}}$, where $n_j = \size{X_j}$. Finally, let $X = \sqcup _{j \in [d]} X_j$. That is, $\set{X_j}_{\set{j \in [d]}}$ form a partition of $X$.

\begin{definition}[Set-Multilinear polynomial] 
A polynomial $P \in \F[X]$ is called \emph{set-multilinear} w.r.t (the partition) $X$, if every monomial that appears in $P$ is of the form $x_{i_1} x_{i_2} \cdots x_{i_d}$
where $x_{i_j} \in X_j$.
\end{definition}

In other words, each monomial of a set-multilinear polynomial picks up exactly one variable from each part in the partition. These polynomial have been well studied in the past \cite{Raz13,ForbesSS14,AKV20} in particular since many natural polynomials like the Determinant, the Permanent, Nisan-Wigderson and others are set-multilinear w.r.t appropriate partitions of variables.  Furthermore, each tensor can be regraded as a set-multilinear polynomial. 

\begin{definition}\label{def:tensor_poly1}
For a tensor $\mc{T}=(\alpha_{j_1,j_2,\ldots, j_d} ) \in \F^{n_1 \times \cdots \times n_d}$  consider the following polynomial 
\begin{equation*}
    f_{\mc{T}}(X) \eqdef \sum_{(j_1,\ldots, j_d) \in [n_1] \times \cdots \times [n_d]} \, \alpha_{j_1,j_2,\ldots, j_d} x_{1,j_1} x_{2,j_2} \cdots x_{d,j_d}.
    \end{equation*}

\end{definition}
Observe that $f_{\mc{T}}(X)$ is a set-multilinear polynomial w.r.t $X$. More interestingly, there is a direct correspondence between tensor decomposition and computing the polynomial $f_{\mc{T}}(X)$ in the model of set-multilinear depth-$3$ circuits. We first define the model formally.

\begin{definition}[Set-Multilinear Depth-$3$ Circuits]
A set-multilinear depth-$3$ circuit w.r.t to (a partition) $X$ with top fan-in $k$, denoted by $\SMLSps(k)$
computes a (set-multilinear) polynomial of the form $$C(X) \equiv  \sum_{i=1}^k \prod_{j=1}^d \ell_{i,j}(X_j)$$
where $\ell_{i,j}(X_j)$ is a linear form in $\F{[X_j}]$.
\end{definition}




To gain some intuition, suppose that $f_{\T}(X) = \ell_{i,1}(X_1)\cdot \ell_{i,2}(X_2)\cdots \ell_{i,d}(X_d)$ for some tensor $\mc{T}$. 
We can observe that in this case $\mc{T}$ is a rank-$1$ tensor. Extending this observation,  the following provides a formal connection between tensor decomposition and computing the polynomial $f_{\mc{T}}(X)$ by set-multilinear depth-$3$ circuits.

\begin{observation}
Let $C(X) = \sum \limits _{i=1}^k \prod \limits_{j=1}^d \ell_{i,j}$ be a set-multilinear depth-$3$ circuit over $\F$ computing $f_{\T}(X)$ for a tensor $\mc{T}=(\alpha_{j_1,j_2,\ldots, j_d} ) \in \F^{n_1 \times \cdots \times n_d}$.
Then $$\T = \sum \limits_{i=1}^k \bar{v}(\ell_{i,1}) \otimes \cdots \otimes \bar{v}(\ell_{i,d})$$ 
where $\bar{v}(\ell_{i,j})$ corresponds to  the linear form $\ell_{i,j}$ as an $n_j$-dimensional vector over $\F$.
\end{observation}
 Note that this connection is, in fact, a correspondence: any $\F$-tensor decomposition of $\mc{T}$ gives a circuit over $\F$. This leads to the following important lemma:
 
\begin{lemma}
\label{lem:decomposition}
A tensor $\mc{T}=(\alpha_{j_1,j_2,\ldots, j_d} ) \in \F^{n_1 \times \cdots \times n_d}$ has rank at most $r$ if and only if $f_{\T}(X)$ can be computed by a $\Sps_{X}(r)$ circuit.
Therefore, rank of $\T$ is the smallest $k$ for which $f_{\T}(X)$ can be computed by a $\Sps_{X}(k)$ circuit.
\end{lemma}
 
 \begin{proof}
 The proof is straightforward. Note that,  $\ell_{i,1}(X_1)\cdot \ell_{i,2}(X_2)\cdots \ell_{i,d}(X_d)$  exactly corresponds to a  rank-1 tensors. Thus, $C_f$ gives a rank $k$ $\F$-tensor decomposition of $\T$ and any $\F$-tensor decomposition  gives a circuit over $\F$.
 \end{proof}

\subsection{Symmetric Tensors and Sum of Power of Linear Forms}

A tensor $\mc{T}$ is called symmetric if $X=X_1=X_2=\cdots =X_d$ and we have $\mc{T}(i_1, i_2, \ldots, i_d ) = \mc{T}(j_1, j_2, \ldots, j_d )$ whenever $(i_1, i_2, \ldots, i_d )$ is a permutation of $(j_1, j_2, \ldots, j_d )$. Thus, a symmetric tensor is a higher order generalization of a symmetric matrix. Analogous to tensor rank, \textit{symmetric rank} is obtained when the constituting rank-1 tensors are imposed to be themselves symmetric, that is $\bar{v}\otimes \bar{v} \cdots \otimes \bar{v}$. 


\begin{definition}
For a symmetric tensor $\mc{T}=(\alpha_{j_1,j_2,\ldots, j_d} ) \in \F^{n \times \cdots \times n}$  consider the following polynomial 
\begin{equation*}
    f_{Sym,\mc{T}}(X) \eqdef \sum_{(j_1,\ldots, j_d) \in [n] \times \cdots \times [n]} \, \alpha_{j_1,j_2,\ldots, j_d} x_{j_1} x_{j_2} \cdots x_{j_d}.
    \end{equation*}

\end{definition}

Just like in case of general tensors, computing the symmetric rank reduces to  finding the optimal top fan-in of a special class of arithmetic circuits, which is sum of power of linear forms ($\pow$) circuits defined below.

\begin{definition}[Sum of power of linear forms] 
The Sum of power of linear forms  with top fan-in $k$  computes a polynomial of the form $f = \ell_1^d + \cdots \ell_k^d$ where each $\ell_i$ is a linear polynomial over the $n$ variables.
\end{definition}

\begin{observation}
Let $C(X) = \sum \limits _{i=1}^k \ell_{i}^d$ be a $\powk{k}$ circuit over $\F$ computing $f_{Sym,\T}(X)$ for a symmetric tensor $\mc{T}=(\alpha_{j_1,j_2,\ldots, j_d} ) \in \F^{n_1 \times \cdots \times n_d}$.
Then $$\T = \sum \limits_{i=1}^k \bar{v}(\ell_{i}) \otimes \cdots \otimes \bar{v}(\ell_{i})$$ 
where $v(\ell_{i})$ is a  $n$-dimensional vector corresponding to  the linear form $\ell_{i,j}$.
\end{observation}
\begin{remark} Both Tensor rank and Symmetric rank are dependent on the underlying field, that is Tensor rank of a tensor $\T$ over $\F$ and  $\G$, an extension of $\F$ can be different, see \cite{Shitov16, SchaeferStefankovic16} for details. The correspondence discussed above, among Tensor rank(symmetric rank) and top fan-in of $\SMLSps$ circuits($\pow$ circuits), respects the dependence of rank on underlying field. That is, in order to find rank of $\T$ over $\G$ we have to find an optimal top fan-in of a  $\SMLSps$ circuit over $\G$ computing $f_{\T}$.
\end{remark}

\subsection{Complexity of Solving a System of Polynomial Equations}
\label{sec:syssolving}

Solving a system of polynomial equations is the following problem: For a field $\F$, we are given $m$ polynomials $f_1, f_2, \ldots, f_m \in \F[x_1, \ldots, x_n]$, each of degree at most $d$. We want to test if there exist a solution (this is the decision version) to $f_1=0, f_2=0, \ldots, f_m=0$ in $\F^n$, or find a solution if it exists (this is the search version).  A straightforward reduction from 3-SAT shows that polynomial system solving is NP-hard in general. This is a fundamental problem in computational algebra, and it has received lot of attention over various fields. To mention a few, system solving is NP-complete for finite fields, in PSPACE over $\R$  \cite{Canny88} and in Polynomial Hierarchy ($\Sigma_2$), assuming GRH \cite{Koiran96}.

Interestingly, for $\F=\Q$ system solving is not even known to be decidable! In fact, if we restrict the question to integral domains (like $\mathbb{Z}$) then the problem is undecidable. This was the well-known  Hilbert’s tenth problem, which asks if a given Diophantine equation has an integral solution, and was famously proved to be undecidable
in the 70’s, see \cite{MR75}.

In this work, we are mainly concerned with polynomial system solving when the number of variables involved is  small (such as a constant). In this case, polynomial system solving turns out is efficient under various settings. We will use the following definitions for describing the complexity of solving a system of equations under various settings.

\begin{definition}[$\Sys_{\F}(n,m,d)$]
Let $\Sys_{\F}(n,m,d)$ denote the randomized time complexity of finding a solution $\in \F^n$ to a  system of $m$ polynomial equations $\in \F[x_1, \ldots, x_n]$ of total degree $d$ (if one exists).

Also, consider a weaker version of the above problem,
let $\widetilde{\Sys}_{\F}(n,m,d)$ denote the randomized time complexity of finding a solution (could be in an extension of $\F$) to a  system of $m$ polynomial equations $\in \F[x_1, \ldots, x_n]$ of total degree $d$ (if one exists).
\end{definition}

\begin{definition}[Det-$\Sys_{\F}(n,m,d)$]
Let Det-$\Sys_{\F}(n,m,d)$ denote the deterministic time complexity of finding a solution $\in \F^n$ to a  system of $m$ polynomial equations $\in \F[x_1, \ldots, x_n]$ of total degree $d$ (if one exists).
\end{definition}

We will now mention various known upper bounds on  $\widetilde{\Sys}_{\F}(n,m,d)$ and  $\Sys_{\F}(n,m,d)$  for various fields. In all these bounds, we have suppressed a $\poly(c)$ dependence in the running time, where $c = \log q$ if $\F = \F_q$ and $c$ is the maximum bit complexity of any coefficient of $f$ if $\F$ is infinite. 

\begin{theorem} 
\label{thm:polysystem}
Let $f_1,f_2, \ldots f_m \in \F[x_1, \ldots, x_n]$  be $n$-variate polynomials of degree at most $d$. Then, the complexity of finding a single solution to the system $f_1(x)=0, \ldots, f_m(x)=0$ (if one exists) over various fields is as follows:
 \begin{enumerate}
    \item For all fields $\F$, $\widetilde{\Sys}_{\F}(n,m,d) = \poly((nmd)^{3^n})$. This follows from standard techniques in elimination theory, see \cite{CLO15} for details. For a detailed sketch of the argument and a bound on the size of the extension, see Appendix~\ref{all-field-system}. 
    \item \cite{HuangWong99}\footnote{the main results of this work is written for the case when $q$ is prime, but the authors observe that it works for general $q$ as well.} For $ \F=\F_{q}$, $\Sys_{\F}(n,m,d)= O(d^{n^n} \cdot (m \log{q}^{O(1)}))$.  
\item \cite{GrigorevVorobjov88} For $\F=\R$, $\Sys_{\F}(n,m,d)=$Det-$\Sys_{\F}(n,m,d)= \poly( (md)^{n^2})$.   Note that in this case the assumption is that the coefficients are integers or rationals\footnote{Here the authors assumed that the constants appearing in the system are integers (or rationals). Note that for all computational applications we can WLOG assume this  by simply approximating/truncating a given real number at some number of bits.}. However, the output might be a tuple of algebraic numbers over $\R$ where the degree of the extension is polynomially bounded when $n$ is a constant. See~\cite{GrigorevVorobjov88} for details. Note that all algebraic algorithms used in this paper will continue to hold when the inputs are algebraic numbers of low/polynomial degree, and we deal with algebraic extensions in the standard way. 
 \item \cite{Ierardi89}   For $ \F=\mathbb{C}$ (or any algebraically closed field), $\Sys_{\F}(n,m,d)=$Det-$\Sys_{\F}(n,m,d)=(mn)^{O(n)} \cdot d^{O(n^2)}$. 
   \end{enumerate}
   \end{theorem}

Note that, for all cases described above, both $\Sys_{\F}(n,m,d)$ and $\widetilde{\Sys}_{\F}(n,m,d)  $ are bounded by $\poly((nmd)^{n^n})$. Thus, when $n=O(1)$,  $\Sys(n,m,d)=\poly({m,d})$.

For clarity in presentation, we \textit{artificially} define $\Sys(n,m,d)$ as the complexity of finding a solution  to a  system of $m$ polynomial equations $\in \F[x_1, \ldots, x_n]$ of total degree $d$ s.t. the solution has to lie in $\F$ if $\F=\R,\C, \F_q$ or an algebraically closed field, and it could be over an algebraic extension for other fields. Clearly, as discussed above $\Sys(n,m,d)=\poly((nmd)^{n^n})$.

\subsubsection{Derandomizing solving system of equations:} \label{Det_system}

Derandomizing solving system of equation in general is considered a hard problem for the following reason. Just solving a univarite quadratic equation over $\F_p$ in \textit{deterministic} $\poly(\log{p})$ time is a notoriously hard open problem, See \cite[Problem 15]{AMcC94}.
Interestingly, this is the only case when low-variate plynomial system solving is hard to derandomize. That is, if the underlying field is not a finite field with large characteristic, then there do exist efficient deterministic algorithms for low-variate system solving. 

Indeed solving systems of polynomial equations is the only place in the paper where randomness is utilized. Thus, all our algorithms can be derandomized over $\R, \C$, since the algorithms mentioned in  Theorem \ref{thm:polysystem}, for polynomial system solving over $\F=\R$ and $\C$ are already deterministic. Though we did not mention it, polynomial system solving (and hence our algorithms) can also be derandomized over $\F_{p^d}$(in time $\poly(p,d)$ time). 


\subsection{Hardness of computing Tensor rank.}\label{sec:tensorhardness}
The first step towards understanding the computational complexity was by H\aa stad  \cite{Hastad90} who showed that determining the tensor rank is an NP-hard  over $\Q$ and  NP-complete  over finite fields. A better way to understand hardness results for computing tensor rank is to study its connection to  solving system of polynomial  equations.

\begin{theorem} \cite{SchaeferStefankovic16}\label{thm:tensorhardness}
For any field $\F$, given a system of $m$ algebraic equations $S$ over $\F$, we can in polynomial time construct a 3 dimension tensor $\T_S$ of shape $[3m]\x [3m] \x [n+1]$  and an integer $k=2m+n$ such that $S$ has a solution $\in \F$ \textit{iff} $\T$ has rank atmost $2m+n$ over $\F$.   
\end{theorem}

This shows equivalence between system solving and computing tensor rank. This along with complexity of system solving (discussed in the previous section) shows that computing tensor rank is NP-complete over finite fields, over $\R$ it is in PSPACE~\cite{Canny88} and is in the Polynomial Hierarchy ($\Sigma_2$), assuming the GRH \cite{Koiran96}.

Similar, reductions also hold for integral domains (e.g. $\Z$) \cite{Shitov16}, thus showing that computing Tensor rank is \textit{undecidable}  over $\Z $ and not known to be decidable over $\Q$. Due to the equivalence between tensor rank computation and learning $\SMLSps$ circuits with optimal top fan-in, we get the corresponding hardness consequences for $\SMLSps$-circuit reconstruction as well.

Such results also hold for symmetric rank computation, see \cite{Shitov16}. Concretely, for 3-dimensional tensors of length $n$, Shitov showed that  we can convert general tensors $\T$ to symmetric tensors $\T_{sym}$ s.t. $rank(\T)+4.5(n^2+n)= \mbox{symmetric-rank}(\T_{sym})$,  thus transferring the results mentioned above for general tensors to symmetric tensors as well.
Again, these hardness results along with equivalence between symmetric  tensor rank computation and reconstructing optimal (w.r.t top fan-in)  $\pow$ circuits implies that proper learning (with optimal top-fan-in) for $\pow$ circuits is as hard as polynomial system solving. In particular, it is NP-hard for most fields and maybe even undecidable over $\Q$.

\section{Reconstruction of $\pow(k)$ circuits (decomposing low rank symmetric tensors) }\label{sec:symmetric}
In this section we will provide a proof of Theorem~\ref{THEOREM:main2}, and we will present our algorithm for reconstructing $\pow(k)$ circuits given black-box access to the polynomial computed by it.
 As discussed in section \ref{sec:Equiv-sml}, this is equivalent to the problem of finding the optimal symmetric tensor decomposition for low rank symmetric tensors. 

In all running times stated in this section, we have suppressed a $\poly(c)$ multiplicative dependence in the running time, where $c = \log q$ if $\F = \F_q$ and $c$ is the maximum bit complexity of any coefficient of $f$ if $\F$ is infinite.

We now restate Theorem~\ref{THEOREM:main2} (only for the randomized algorithm over general fields) and prove it. After the proof we will comment on how the algorithm can be derandomized over $\R$ and $\C$.

\begin{theorem}\label{thm:symmetric}
Given black-box access to a degree $d$ polynomial $f \in \Fn$ such that $f$ is computable by a $\powk{k}$ circuit $C_f$  over field $\F$(characteristic $> d$ or $0$), there is a randomized $\poly{(dk)^{k^{k^{10}}}}$ time algorithm that outputs a $\powk{k}$ circuit computing $f$. 

\end{theorem}
\begin{remark}
 when $\F = \F_q$ for $q> nd$ and when $\F = \R$ or $\C$, the output circuit is over the same underlying field $\F$. In general the output circuit might be over an algebraic extension of $\F$. 
\end{remark}

\begin{proof}
The main observation is that if $f$ can be represented by a $\powk{k}$ circuit, then $f$ has only $k$ essential variables.  Thus by Lemma~\ref{lem:carlini-black-box}, there is an algorithm that given black-box access to $f$, runs in time $\poly(n,d)$ and outputs an invertible linear transformation $A \in \F^{n \times n}$ such that $f(A \cdot X)$ depends only on $k$ variables.  Let $g_A(X) = f(A\cdot X)$. Since we can compute $A$, hence given black-box access to $f$, we can simulate black-box access to $g_A$ in time $\poly(n,d)$.

Notice that $g_A$ is a degree $d$ polynomial in $k$ variables that is also computed by a $\powk{k}$ circuit. We will show how to efficiently learn a $\pow(k)$ representation of $g_A$. Since $g_A(A^{-1}\cdot X) = f(X)$, thus given a $\powk{k}$  representation of $g_A$ we can obtain a $\powk{k}$ representation of $f$.

The algorithm for learning a $\pow(k)$ representation of $g_A$ works as follows. It starts by learning $g_A$ as a sum of monomials (i.e. the sparse polynomial representation of $g_A$). In particular, let S denote the collection of non-negative integer n-tuples summing to $d$. The algorithm finds a collection of coefficients $\{c_{\bar{e}} \in \F | \bar{e} \in S\}$ such that $g_A=\sum_{\bar{e} \in S }c_{\bar{e}} \cdot \bar{x}^{\bar{e}}$.  This can be done  in $\poly({k+d \choose d}) = \poly(d^k)$ using known sparse polynomial reconstruction algorithms~\cite{KlivansSpielman01, Ben-OrTiwari88}.

Let \[C_{g_A} \equiv \sum_{i=1}^{k}  (a_{i,1}x_{1} +  a_{i,2}x_{2} + a_{i,3}x_{3}+ \ldots +  a_{i,k}x_{k})^{d}\] be the $\powk{k}$ circuit computing $g_A$.

Thus, \[\sum_{i=1}^{k}  (a_{i,1}x_{1} +  a_{i,2}x_{2} + a_{i,3}x_{3}+ \ldots +  a_{i,k}x_{k})^{d} = \sum_{\bar{e}  }c_{\bar{e}} \cdot \bar{x}^{\bar{e}}.\] 

Now for each monomial $\bar{x}^{\bar{e}}$ , $\bar{e} \in S$, we can compare the coefficient of $\bar{x}^{\bar{e}}$ on both sides to get a polynomial equation in the variables $a_{i,j}$. Doing this for all monomials gives us
a system of at most ${k+d \choose d}$ polynomial equations in $ k^2$ variables, with $a_{i,j}$ as variables. By Theorem~\ref{thm:polysystem}, this system can be solved in time $\Sys(k^2, {k+d  \choose d }, d)$. Thus, total time complexity is bounded by  $\poly{(dk)^{k^{k^{10}}}}$.
\end{proof}

\textbf{Derandomization:} 
In the above proof, randomness is used in the variable reduction step (Lemma \ref{lem:carlini-black-box}) and polynomial system solving (Theorem~\ref{thm:polysystem}). Over $\R$ and $\C$, Theorem~\ref{thm:polysystem} in fact states that polynomial system solving can be done {\it deterministically} in the same time complexity.

Moreover, for derandomized variable reduction, we can use Lemma \ref{derand-carlini} instead of Lemma \ref{lem:carlini-black-box}. Observe that all the assumptions of Lemma \ref{derand-carlini} are satisfied, since
 as $f$ is computed by a $\pow(k)$ circuit, it has at most k essential variables, and the class is also closed under taking first order partial derivatives. Furthermore, $\underbrace{\mc{C}+\ldots+\mc{C}}_{k\,  times} $ for $\mc{C}=\pow(k)$ is just the class if $\pow(k^2)$ circuits, so by Lemma \ref{lem:pit sps}, there is an efficient hitting set for  $\underbrace{\mc{C}+\ldots+\mc{C}}_{k\,  times} $.


 Thus, putting it all together we see that we can derandomize the algorithm for proper learning algorithm for  $\pow(k)$ circuits over $\R, \C$.

\section{Reconstructing $\SMLSps$ circuits (decomposing low rank tensors).}\label{sec:setmultilinear}

In this section we will provide a proof of Theorem~\ref{THEOREM:main1}, and we will present our algorithm for reconstructing $\SMLSps(k)$ circuits given black-box access to the polynomial computed by it. We will first present all the details for the randomized algorithm and then later comment on how to derandomize it over certain fields. 

In all running times stated in this section, we have suppressed a $\poly(c)$ multiplicative dependence in the running time, where $c = \log q$ if $\F = \F_q$ and $c$ is the maximum bit complexity of any coefficient of $f$ if $\F$ is infinite.

Before we prove Theorem~\ref{THEOREM:main1}, we develop a bunch of lemmas and subroutines that will be used in the final algorithm. 

\subsection{Width reduction for $\SMLSps(k)$ circuits}\label{sec:widthreduction-setmultilinear}
We first show that there is an algorithm for learning $\SMLSps(k)$ circuits of arbitrary width $w$ in roughly the same amount of time it takes to learn $\SMLSps(k)$ circuits of width $k$. 

The algorithm achieves this by a certain ``width reduction" procedure that maps the given circuit to one of low width while ensuring we can still get black-box access to it. The algorithm then learns the low width circuit and inverts the map to recover the original possibly high width circuit. A similar ``width reduction"  technique also appeared in a work of Gupta, Kayal and Lokam \cite{GKL12}, but there it was simpler since it was specialized to the case of top fan-in 2, and hence it avoided some of the subtleties that arise here.

\begin{lemma}\label{lem:widthreduction}
Given black-box access to a degree $d$, $n$ variate polynomial $f \in \F[X]$ such that $f$ is computable by (arbitary width) $\SMLSps(k)$ circuit $C_f$  over the field $\F$ with $|\F| > dn$(otherwise  we can work with an extension), there is a randomized polynomial-time
algorithm that outputs the following:

\begin{enumerate}
\item Black-box access to another related polynomial $h$ which is computed by a width-$k$ $\Sps_{\sqcup_j Y_j}(k)$ circuit $C_{h}$. Each black-box query to $h$ can be simulated in polynomial time by a suitable related query to $f$.
\item The underlying partition of the $Y$-variables, which is of the form $Y := \sqcup_{i=1}^d{Y_i}$ where $Y_j:= \{y_{1,j}, y_{2,j}, \ldots, y_{k_j,j} \}$. For all $j \in [d]$, $k_j\leq k$ thus $width(h)\leq k$.
\item A collection of linear polynomials  $\{P_{i,j}\in \mathbb{F}[X_j]\}_{ j \in [d], i \in [k_j]} $ such that  the following holds: 
Upon substituting $y_{i,j}=P_{i,j}$ into the polynomial $h$ (for each each variable $y_{i,j}$ that appears in $h$), we recover $f$. Moreover $C_h(y_{i,j}=P_{i,j})$ is a $\Sps_{\set{\sqcup_i X_i}} (k)$ representation of $f$.


\end{enumerate}
\end{lemma}


 




\begin{proof}
Let $C_f = \SMLSps(k)$ be the depth-3 set-multilinear circuit representation of $f$. Let
\begin{equation*} \label{smSPS}
    C_f \equiv  \sum_{i=1}^k \prod_{j=1}^d \ell_{i,j}  \qquad \qquad \mbox{where $\ell_{i,j}$ is a linear polynomial in $\bar{X_j}$ variables.  } 
\end{equation*}

Without loss of generality we will assume that $k$ is the smallest integer such that $f$ has such a representation. The reason we can do this because of the following. If there is a representation of $f$ with a smaller top fan-in $k'$, then we can assume that the algorithm knows $k'$ and runs the algorithm with $k'$ instead of $k$. The reason we can assume the algorithm ``knows" $k'$ is that we can run the algorithm for all values of $k'$ from 1 up to $k$ and try to learn the circuit with that top fan-in, and the first time it successfully learns a circuit will correspond the representation of $f$ with the lowest top fan-in. The algorithm knows when it has successfully learnt the circuit because it can test whether or not the output circuit agrees with the input $f$ using polynomial identity testing  (See Lemma~\ref{SchwartzZippel} and Lemma~\ref{lem:pit sps}). 

As a first step, the algorithm will learn the linear span of the set $ \{\ell_{j,i} | j \in [k]\}$ for each $i \in [d]$.



To do this, substitute all variables in $X \setminus X_i$ to independent randomly chosen values from $\F$ and interpolate to get a linear polynomial  $ P_{1,i} \in \F[X_i]$. Then $P_{1,i}=\alpha^{(1)}_{1,i} \ell_{1,i} + \alpha^{(1)}_{2,i} \ell_{2,i}+ \ldots \alpha^{(1)}_{k,i} \ell_{k,i} $ where $\alpha^{(1)}_{j,i} \in \F$. Observe that the algorithm learns $P_{1,i}$ but the $\alpha^{(1)}_{j,i}$ and $\ell_{j,i}$ are unknowns. We repeat this procedure for another $k-1$ random independent substitutions of the variables of $X \setminus X_i$ and upon interpolation recover $k-1$ additional linear combinations of $\ell_{1,i}, \cdots, \ell_{k,i}$. Let the resulting learnt linear polynomials be $P_{2,i}, P_{3,1}, \ldots P_{k,i}$ respectively.

Now, we have that,
\begin{equation} \label{P1}
\begin{pmatrix}
P_{1,i} \\
P_{2,i} \\
\vdots \\
P_{k,i} \\
\end{pmatrix}
=\underbrace{
\begin{pmatrix}
\al^{(1)}_{1, i} & \al^{(1)}_{2,i} & \cdots & \al^{(1)}_{k,i}  \\
\al^{(2)}_{1,i} & \al^{(2)}_{2,i} & \cdots & \al^{(2)}_{k,i}  \\
\vdots & \vdots  & \vdots & \vdots \\
\al^{(k)}_{1,i} & \al^{(k)}_{2,i} & \cdots & \al^{(k)}_{k,i}  \\
\end{pmatrix}}_{A_{{\bar{\alpha_i}}}} \cdot \begin{pmatrix}
l_{1,i} \\
l_{2,i} \\
\vdots \\
l_{k,i} \\
\end{pmatrix}
\end{equation}

Now the great advantage of our assumption that $k$ is the smallest integer such that $f$ has a $\SMLSps(k)$ representation is that we can invoke Lemma~\ref{indep}. In conjunction with Lemma~\ref{invert1}, this implies that $A_{\bar{\alpha_i}}$ is invertible with high probability.

Thus note that for each $j \in [k]$, $\ell_{j,i}$ is in the linear span of $\{P_{1,i}, \ldots, P_{k,i}\}$.

If $\{P_{1,i}, P_{2,i},\ldots, P_{k,i}\}$ are linearly independent polynomials, then the algorithm does the following. It introduces $k$ new formal variables $y_{1,i}, \ldots, y_{k,i}$ and defines $k$ linear functions in these variables as follows. For each $j \in [k]$, it defines $\tilde{l}_{j,i}$ by the following equation.
\begin{equation}\label{new_vars}
\begin{pmatrix}
\tilde{l}_{1,i} \\
\tilde{l}_{2,i} \\
\vdots \\
\tilde{l}_{k,i} \\
\end{pmatrix}= A_{\bar{\alpha_i}}^{-1}\begin{pmatrix}
y_{1,i} \\
y_{2,i}\\
\vdots \\
y_{k,i} \\
\end{pmatrix}    
\end{equation}

If $\{P_{1,i}, \ldots, P_{k,i}\}$ are not linearly independent then find a set $S_i\subseteq [k]$ such that the elements of $B_i = \{P_{j,i} | j \in S_i\}$ form a basis of $\{P_{1,i}, \ldots, P_{k,i}\}$. 
For instance, one can find $S_i$ using Lemma~\ref{invert}.
let $|S_i| = k_i$.

 
 In this case, change Equation~\ref{new_vars} by keeping $y_{j,i}$ for $j \in S_i$ as a formal variables and replacing $y_{j,i}$ for each $j \notin S_i$ in the following way: If $P_{j,i}=\sum_{j \in S_i} \lambda_j P_{j,i}$, then replace $y_{j,i}=\sum_{j \in S_i} \lambda_j y_{j,i}$. Observe that in both cases the width of $\tilde{l}_{j,i}$ is  $\leq k$. 
 
 The algorithm performs the above procedure for each $i \in [d]$, and thus for each $i \in [d]$ and each $j \in [k]$ it obtains a linear polynomial $\tilde{l}_{j,i}$.

Now consider the following polynomial $h=\sum_{j=1}^k\prod_{i=1}^d \tilde{l}_{j,i}$. Notice that $h$ is a $\Sps_{\set{\sqcup_i Y_i}} (k)$ circuit of width at most $k$, with underlying partition $Y:=\sqcup_{i=1}^{d}Y_i$, where $Y_i:= \{y_{1,i}, y_{2,i}, \ldots, y_{k_i,i}  \}$, $k_i = |S_i|$. 

\begin{claim} There is an efficient polynomial-time algorithm for simulating black-box access to $h$.
\end{claim}
\begin{proof}
Suppose the algorithm wants to evaluate $h$ at an input $\beta$, where $\beta=(\beta_{j,i})_{i \in [d], j \in k_i}$, $\beta_{j,i} \in \F$. 

This can be done by solving the following system of linear equations in the $X$ variables, to obtain a solution $X_0$, and then evaluating $f$ at $X_0$.
Suppose that for each $i \in [d]$, $S_i:=\{i_1, i_2, \ldots, i_{k_i}\}$. The for each $i \in [d]$, add the following equations to the system of equations. 

\[ 
\begin{pmatrix}
P_{i_1,i}(X_i) \\
P_{i_2,i}(X_i) \\
\vdots \\
P_{i_{k_i},i}(X_i) \\
\end{pmatrix}
= \begin{pmatrix}
\beta_{i_1,i} \\
\beta_{i_2,i} \\
\vdots \\
\beta_{i_{k_i},i} \\
\end{pmatrix} \]

Existence of a solution is guaranteed, since for each $i \in [d]$, $\{P_{j,i}(X_i) | j \in S_i\}$ is a linearly independent collection of linear polynomials, and hence a solution can be efficiently found by Gaussian elimination.
\end{proof}


By Equations~\ref{P1} and Equation~\ref{new_vars}, we conclude that upon substituting $y_{i,j}=P_{i,j}$ in $h$ we recover $f$. Moreover, $C_h(y_{i,j}=P_{i,j})$ is a $\Sps_{\set{\sqcup_i X_i}} (k)$ representation of $f$. 


\end{proof}

\begin{corollary}\label{cor:widthreduction}
Suppose $A$ is an algorithm that has the following behavior. On input black-box access to degree $d$, $n$-variate polynomial $f \in \F[X]$ such that $f$ is computable by a width $k$ $\SMLSps(k)$ circuit $C_f$  over the field $\F$, runs in randomized time $\mathcal{A}(n,d,k)$ and outputs a $\SMLSps(k)$ circuit computing $f$. Then there is another algorithm $A'$ that has the following behavior. On input black-box access to degree $d$, $n$-variate polynomial $f \in \F[X]$ such that $f$ is computable by an arbitrary width $\SMLSps(k)$ circuit $C_f$  over the field $\F$, runs in randomized time $\poly(n,d,k) \cdot \mathcal{A}(n,d,k)$ and outputs a $\SMLSps(k)$ circuit computing $f$.
\end{corollary}
\begin{proof}
The algorithm $A'$ works as follows. Using the procedure described in Lemma~\ref{lem:widthreduction}, it uses black-box queries to $f$ to simulate black-box queries to another polynomial $h$ which is computed by a degree $d$, width $k$, $\SMLSps(k)$ circuit in at most $n$ variables. It then uses algorithm $A$ to obtain a $\Sps_{\sqcup_j Y_j}(k)$  representation of $h$ in time $\poly(n,d,k) \cdot \mathcal{A}(n,d,k)$ (since each query to $h$ takes time $\poly(n,d,k)$) and then makes the suitable substitution of linear polynomials into the variables of $h$ to recover $f$.
\end{proof}

\subsection{Reconstructing low degree $\SMLSps(k)$ circuits }\label{sec:lowdegree-setmultilinear}
As a basic step in reconstructing general $\SMLSps(k)$ circuits, we show how to reconstruct $\SMLSps(k)$ circuits  efficiently when the degree of the computed polynomial (which corresponds to the dimension of the underlying tensor) is small. 

We will obtain a bound as a function of the width-$w$, but when we use the lemma later, we will assume the width is $k$ (due to our width reduction lemma). 
(Recall, the width $w$ is an upper bound on the number of variables in each $X_i$.)
The notation we use in the running time of the lemma below is from Section~\ref{sec:syssolving}.

\begin{lemma}\label{lem:setmultilinear-lowdegree}
Given black-box access to a degree $d$ polynomial $f \in \F[X]$ such that $f$ is computable by a width $w$ $\SMLSps(k)$ circuit $C_f$  over the field $\F$, there is a randomized $\Sys(kwd,w^d,d) + \poly(w,k,d)$  time algorithm that outputs a $\SMLSps(k)$ circuit computing $f$. 
\end{lemma}
\begin{remark}
 when $\F = \F_q$  and when $\F = \R$ or $\C$, the output circuit is over the same underlying field $\F$. In general the output circuit might be over an algebraic extension of $\F$. 
\end{remark}


\begin{proof}
Start by learning $f$ as a sparse polynomial, that is find $c_{\bar{e}} \in \F$ such that $f=\sum_{\bar{e} \in  X_1 \times \cdots \times X_d} c_{\bar{e}} \cdot \bar{x}^{\bar{e}}$.  This can be done  in $\poly(n,d)$ using \cite{KlivansSpielman01, Ben-OrTiwari88} sparse polynomial reconstruction.

Also, let \[C_f \equiv \sum_{i=1}^{k} \prod_{j=1}^{d} (a_{i,j,1}x_{j,1} +  a_{i,j,1}x_{j,1} + a_{i,j,2}x_{j,2}+ \ldots +  a_{i,j,1}x_{j,1})\] be a \sk circuit computing $f$.

Thus, \[\sum_{i=1}^{k} \prod_{j=1}^{d} \sum_{\ell=1}^{w} a_{i,j,\ell}\cdot x_{j,\ell} =  \sum_{\bar{e} \in  X_1 \times \cdots \times X_d } c_{\bar{e}} \cdot \bar{x}^{\bar{e}}\] gives us a system of polynomial equations with $a_{i,j,\ell}$ as variables. Notice that, this system has $w^d$ polynomial equations  of degree d  supported on $kwd$ variables. Thus, can be solved in  $\Sys_{\F}(kwd,w^d,d)$. Since, system solving is the most expensive step of this process thus the overall time complexity is bounded by $\Sys(kwd,w^d,d) \poly(w,k,d) = \poly(kwd \cdot w^d \cdot d \cdot c)^{(kwd)^{(kwd)}}$.
\end{proof}



\subsection{Learning 2 linear forms appearing in the circuit}
As a first step towards learning a general (high degree) $\SMLSps(k)$ circuit $C$, our algorithm will try to learn a single linear form appearing in the circuit $C$. In fact it will be convenient to learn 2 linear forms appearing in $C$ such that each multiplication gate of $C$ contains at most one of them. 
\begin{lemma}\label{lem:twolinearforms}
Let $C \equiv \sum_{i=1}^k T_i$ be a width-$k$, simple and minimal $\SMLSps(k)$ circuit computing a nonzero polynomial $f$,
such that $k \geq 2$ and $d > 2k^2$.
Then, given black-box access to $f$, there is a randomized algorithm that runs in time $\Sys(2k^4 ,k^{2k^2},2k^2)+ \poly(k,d)$, and does the following. It outputs a set $L$ of $2k^3$ pairs of linear forms which has the following property. One of the pairs $(\ell_1, \ell_2)$ in $L$ is such that for some $i \in [d]$,  $\ell_1$ and $\ell_2$ are supported on the variables of $X_i$, and both $\ell_1$ and $\ell_2$ appear in $C$. 
\end{lemma}

\begin{proof}
Let $\sigma = \sqcup_{2k^2+1}^d \bar \alpha_i$,  be a random assignment from the underlying field $\F$ to the variables in $\sqcup_{i = 2k^2+1}^{d} X_i$ performed in the following way: we pick $S \subseteq \F$ such that $|S| > 2^{k+1} nd$. If $\F$ is not large enough, then we can pick $S$ from an extension field. Then each element of $\sqcup_{i = 2k^2+1}^{d} X_i$ is independently set to a uniformly random element of $S$, and let the resulting polynomial be $f|_{\sigma}$.

Then observe that after setting these variables, with high probability, the restricted circuit $C|_{\sigma}$ is still nonzero, simple and minimal. Simplicity follows directly. To observe non-zeroness and minimality, consider the following family of polynomials, $\forall U \subset [k], f_U:=\sum_{i \in U} T_i$ . Note that, using minimality of $f$, we get that   $\forall U \subset [k], \, f_U \neq 0$, thus on applying   Lemma \ref{SchwartzZippel} on $\prod_{\forall U}f_U$ gives that for aforementioned random $\sigma$, $f_U|_{\sigma}\neq 0$ implying non-zeroness and minimality of $C|_{\sigma}$.

Now we use Lemma~\ref{lem:setmultilinear-lowdegree} to learn a $\SMLSps(k)$ circuit with variable partition $\sqcup_{i = 1}^{2k^2} X_i$ computing $f|_{\sigma}$. Let it be $\tilde C = \sum_{i=1}^{k'} B_i$, with $k' \leq k$ and WLOG $\tilde C$ is minimal (else we remove all identically zero subcircuits). 

The algorithm then outputs a set $L$ which comprises of all pairs of linear forms appearing in $\tilde C$ that are supported on the same set of variables $X_i$ (for each $X_i$). Let us now see why this set $L$ has the desired property.

By Lemma~\ref{lem:unq}, we get that for each $T_i\rs$ there exists $j_i \in [k]$ such that $\Delta(T_i\rs,B_{j_i} )\leq k$. Let us call the set of $X_i$'s on which  $\frac{T_i\rs}{\gcd(T_i\rs, B_{j_i})}$  is supported ``bad". Then, the number of sets that are bad for at least one choice of $T_i$ is bounded by $k^2$. Thus, since $d >  2k^2$, there exists at least one set (say $X_a$) such that if a linear form $\ell \in \F[X_a] $ is such that $\ell$ divided $T_i\rs$, then $\ell$ also divides $B_{j_i}$. Also, since $C|_{\sigma}$ is simple, there are at least two distinct linear forms $\ell_1, \ell_2 \in \F[X_a]$ that each divide some multiplication gate of $C|_{\sigma}$.
Clearly this pair of linear forms is also included in $L$. The time complexity estimate is immediate. 
\end{proof}



\subsection{Learning most of the linear forms appearing in the circuit}
In this section we will see how to use the two linear forms learnt in the previous subsection to learn a small set $S$ of multiplication gates, such that each multiplication gate $T_i$ of $C$ is very ``close" to some gate of $S$. From this we will then see how to essentially learn {\it most} of the linear forms appearing in each gate of $C$.
Our approach is recursive, so we will assume using an induction hypothesis that there is $\mc{A}(n,k-1,d)$ time randomized algorithm  for reconstructing degree $d$, width $k$ $\SMLSps(k-1)$ circuits. Note that, for base case  of $k=1$ follows directly from black-box factoring result of \cite{KaltofenTrager90} i.e. $\mc{A}(n,1,d) =\poly(n,d)$. 

\begin{lemma}\label{lem:mostlinearforms1}
Let $C \equiv \sum_{i=1}^k T_i$ be a minimal $\SMLSps(k)$ circuit of degree $d$ and let $\ell_1$ and $\ell_2$ be two distinct linear forms supported on variables of $X_i$ for some $i \in [d]$ such that each of $\ell_1$ and $\ell_2$ appears in $C$. Then there is a randomized algorithm that given black-box access to $C$, given $\ell_1$ and $\ell_2$, and two oracle calls  to algorithm  for learning $\SMLSps(k-1)$ circuits, runs in time 
atmost $2\mc{A}(n,k-1,d)+ \poly(n,k,d)$ 
and outputs a set $S = \{M_1, M_2 \ldots, M_{|S|}\} $ of at most $2k-2$ $\Pi\Sigma$ circuits of degree $d-1$ such that with high probability, for all $i \in [k]$, there exists $j \in [2k-2]$ such that $\Delta(T_i, M_{j}) < 2k$. 
\end{lemma}

\begin{proof}
Let $\bar \alpha_1$ be a setting of the variables of $X_i$ (chosen from a suitably large domain) on which $\ell_1$ vanishes. Notice that with high probability, $\ell_1$ is the only linear form appearing in $C$ that vanishes, and the restricted circuit $C|_{X_i =  \bar \alpha_i}$ is a $\SMLSps(k-1)$ circuit. Indeed it is also possible that it might have much few than $k-1$ gates. Using the induction hypothesis, there is an efficient algorithm to learn a $\SMLSps(k-1)$ representation of $f|_{X_i= \bar \alpha_1}$. We run this algorithm and let the output be $C_1$.  
 
Similarly we let $\bar \alpha_2$ be a setting of the variables of $X_i$ (chosen from a suitably large domain) on which $\ell_2$ vanishes. We learn a $\SMLSps(k-1)$ representation of $f|_{X_i=\bar \alpha_2}$ and let the output be $C_2$.

Let $S$ denote the set of union multiplication gate from  $C_1$ and $C_2$. Note, $|S| \leq 2(k-1)$. Observe that for each $i \in [k]$, with high probability, one of the substitutions $X_i = \bar \alpha_1$ of $X_i = \bar \alpha_2 $ must keep $T_i$ nonzero, since at most one of $\ell_1, \ell_2$ divides $T_i$. Thus, by Lemma~\ref{lem:unq}, for all $i \in [k]$, $T_i$ must be ``close" to some multiplication gate of $C_1$ or $C_2$. More precisely, there exists $j \in [2k-2]$ such that $\Delta(T_i, M_{j}) < 2k$. The time complexity estimates is immediate.
\end{proof}

Thus we can now assume that our algorithm can compute a set $S$ consisting of multiplication gates  such that $|S| \leq 2(k-1)$, and each gate of the circuit $C$ that we are trying to learn is close to some element of $S$.

\begin{lemma}\label{lem:mostlinearforms2}
Let $C \equiv \sum_{i=1}^k T_i$ be a $\SMLSps(k)$ circuit of degree $d$
and let $S = \{M_1, M_2 \ldots, M_{|S|}\} $ be a set of at most $2k-2$ $\Pi\Sigma$ circuits of degree $d-1$ such that for all $i \in [k]$, there exists $j \in [2k-2]$ such that $\Delta(T_i, M_{j}) < 2k$. Then there is a $\poly{((kd^{k^3}))}$-time algorithm for computing another set $\bar S$ which has the following properties.
\begin{enumerate}
    \item $|\bar S| \leq \left(|S|\cdot {d-1 \choose 2k^2}\right)^k$
    \item The elements of $\bar S$ are $k$-tuples of $\Pi\Sigma$ circuits of degree $d-1-2k^2$ 
    \item One of the elements of $\bar S$ is of the form $(G_1, G_2, \ldots G_k)$ where for each $i \in [k]$, $G_i$ divides $T_i$. Moreover all of the $G_i$s are set-multilinear $\Pi\Sigma$ circuits sharing the same variable partition. 
    \end{enumerate}
\end{lemma}

\begin{proof}

Consider the set $\hat{S}=\{G \;| \,\exists M \in S \text{ s.t. }  G \mbox{ divides } M, \mbox{ and } \deg(G) = d-1-2k  \}$. Notice that, $|\hat{S}| \leq |S| \cdot {d-1 \choose 2k}$. 

By the property of the set $S$, this implies that for all $i \in [k]$, there is some element $G_i' \in\hat{S}$ such that $G_i'$ divides $T_i$.
There may be multiple elements of $\hat{S}$ that divide $T_i$, but we fix any one and call it $G_i'$.
Now consider the set $\{G_1', G_2', \ldots, G_k'\}$. 

Ideally we would like these $G_i'$ to depend on the same set of variables. Here is a modification of them that will result in somewhat lower degree polynomials, but they would be supported on the same variable partition. First recall that all the $G_i'$s are a product of set disjoint linear forms, and each linear form is supported on one of the parts of the underlying partition  $X=\sqcup_{i=1}^d X_i$. Now we perform the following procedure. For each $i \in [k]$ and for each $j\in [d]$ such that $G_i'$ does not have a factor $\in \F[X_j]$ then remove(divide out) any linear factor $\in \F[X_j]$  from all the other elements of the set $\{G_1', G_2', \ldots, G_k'\}$. 
At the end of this process, the new polynomials have degree at least $d-1 - 2k^2$ and they all are supported on the same parts of the partition $X=\sqcup_{i=1}^d X_i$. We can also ensure that they all have degree exactly $d-1 - 2k^2$, as we can divide out linear forms all depending on the same set of variables from all of these polynomials till they have degree exactly $d-1 - 2k^2$. Call these new polynomials $G_1, G_2, \ldots, G_k$. Observe that they all have the same variable partition, they all have degree $d-1 - 2k^2$, and they each divide some multiplication gate $T_i$ of the circuit $C$ as well as some element of $S$.

We will now define the set $\bar S$. First consider the set $Q =  \{G | \exists M \in S \text{ s.t. }  G \mbox{ divides } M, \mbox{ and } \deg(G) = d-1-2k^2  \}$.
Let $\bar S$ be the set of all $k$-tuples of elements of $Q$. Then observe that $(G_1', G_2', \ldots, G_k')$ is an element of $\bar S$.
Then $|\bar S| \leq \left(|S|\cdot {d-1 \choose 2k^2}\right)^k$ and it satisfies all other required properties as well. 

\end{proof}

  
 

\subsection{Learning the full circuit}\label{sec:fullcircuit}
\begin{lemma}\label{lem:fullcircuit}
Let $C \equiv \sum_{i=1}^k T_i$ be a $\SMLSps(k)$ circuit of degree $d$ and width $w$ computing a polynomial $f$. Let $(G_1, G_2, \ldots G_k)$ be a $k$-tuple where for each $i \in [k]$, $G_i$ divides $T_i$. Moreover all of the $G_i$s are set-multilinear $\Pi\Sigma$ circuits of degree $d-1-2k^2$, sharing the same variable partition. Then, given black-box access to $C$ and given the $k$-tuple $(G_1, G_2, \ldots G_k)$, there is a randomized $\Sys(k(2k^2+1)w,(2k^2w)^{2k^2+1},2k^2+1) + \poly((2k^2w)^{2k^2})$ time algorithm that outputs a $\SMLSps(k)$ representation of $f$. 
\end{lemma}

\begin{proof}

We are given $(G_1, G_2, \ldots G_k)$ and our aim is to find $H_i$-s such that $f=\sum_{i \in [k]} G_i H_i$ is a $\SMLSps(k)$ representation of $f$. 
Notice that, if we have black-box access to the individual $H_i$-s, then we can learn them just by black-box factorization followed by sparse reconstruction of the linear factors. We will  achieve something close to this in principle. 

As a first step, we find the linear dependency structure among the $G_i$-s using Lemma~\ref{invert}. 

Let  $G_1, G_2, \ldots G_c$ be a basis of linear space of $G_i$-s (we can always ensure this by relabelling of gates). Also, let $M$ be a $k \times c$ matrix which is the corresponding linear dependence matrix we get from  Lemma~\ref{invert}, that is, \begin{equation} \label{M}
M_{k \times c}\begin{pmatrix}
G_1\\  \vdots\\ G_c
\end{pmatrix}= \begin{pmatrix}
G_1\\ \vdots \\ G_c\\ \vdots\\ G_k
\end{pmatrix}.\end{equation}
 Note that, \begin{align*}
    f&=\big(H_1, \, \cdots, \, H_k 
    \big)\cdot\begin{pmatrix}
    G_1 \\ \vdots \\ G_k 
    \end{pmatrix}= \big(H_1,\, \cdots, \, H_k \big)\cdot M \cdot \begin{pmatrix} G_1 \\ \vdots \\ G_c    \end{pmatrix}
\end{align*}

For $i \in [c]$, define $\tilde{H}_i$ by the following equality  \begin{equation}\label{tildeH}
\begin{pmatrix}
    \tilde{H}_1\\  \vdots \\ \tilde{H}_c 
    \end{pmatrix}:= M^{T}\begin{pmatrix}
    {H}_1\\  \vdots \\.\\ {H}_k \end{pmatrix}.
\end{equation}

Thus, $$f=\sum_{i \in [k]}H_iG_i= \sum_{i \in [c]}\tilde{H}_iG_i. $$

We will now show how to  obtain black-box access to the $\tilde{H_i}$ and then later use the $\tilde{H_i}$ to find the $H_i$-s.

Now observe that $f=\sum_{i=1}^{c} \tilde{H_i} G_i$ where $G_1, \ldots G_c $ are linearly independent (and we know what $G_1, \ldots G_c $ are). Let $Z \subset X$ be the set of variables on which the $G_i$'s are supported. Then for each $i$, the $\tilde{H_i}$ and the $H_i$ belong to $\F[X \setminus Z] $.

We will now show how to get get black-box access to the $\tilde{H}_i$ by using black-box access to $f$ and to the $G_i$-s.


Let $\beta_1,  \beta_2, \ldots, \beta_c $ be random independent substitution of $Z$ variables. Let $f_{\be_i} \in \F[X \setminus Z]$ be the polynomial obtained by substituting $\beta_i$ into the $Z$ variables of $f$. Then $f_{\be_i}$ is a $(2k^2+1)w$  variate multilinear polynomial (think of $w$ being small, such as $k$) and thus  $f_{\be_i}$  can be represented efficiently as a sparse polynomial (and in fact we can learn the monomial representation using black-box sparse polynomial interpolation in $\poly((2k^2w)^{2k^2})$ \cite{KlivansSpielman01}) . 

Now observe that 
\[
\begin{pmatrix}
f_{\beta_1} \\
f_{\beta_2} \\
\vdots \\
f_{\beta_c} \\
\end{pmatrix}
=\underbrace{
\begin{pmatrix}
G_1(\beta_1) & G_2(\beta_1) & \cdots & G_c(\beta_1)\\
G_1(\beta_2) & G_2(\beta_2) & \cdots & G_c(\beta_2)\\
\vdots & \vdots  & \vdots & \vdots \\
G_1(\beta_c) & G_2(\beta_c) & \cdots & G_c(\beta_e)\\ 
\end{pmatrix}}_{B_\be} \cdot \begin{pmatrix}
\tilde{H}_1 \\
\tilde{H}_2 \\
\vdots \\
\tilde{H}_c \\
\end{pmatrix}
.\]

Since $G_1, \ldots G_c $ are linearly independent, by Lemma~\ref{invert1} we see that $B_{\be}$ is invertible with high probability. Thus, \[\begin{pmatrix}
\tilde{H}_1 \\
\tilde{H}_2 \\
\vdots \\
\tilde{H}_c \\
\end{pmatrix} = B_{\be}^{-1}
\begin{pmatrix}
f_{\beta_1} \\
f_{\beta_2} \\
\vdots \\
f_{\beta_c} \\
\end{pmatrix} \] and in this manner we obtain access to the $\tilde{H}_i$-s. Indeed we can recover the monomial representation of the $\tilde{H}_i$-s.

Note that each $\tilde{H}_i$ is a polynomial is at most $(2k^2+1)w$ variables and is of degree $ \leq 2k^2+1$.

For ease of presentation, assume each $H_i \in \F[X_1, X_2, \ldots, X_{2k^2+1}]$, which can ensured by relabeling of variables.

Recall, \begin{equation} \label{patch}
    M^{T} \cdot \begin{pmatrix}
H_1\\ .\\ \vdots\\ H_k
\end{pmatrix}= \begin{pmatrix}
\tilde{H}_1\\ \ \vdots\\ \tilde{H}_c
\end{pmatrix}
\end{equation} and let $H_i=\prod_{j \in [2k^2+1]} (a_{i,j,1}x_{j,1}+ a_{i,j,2}x_{j,2}  \ldots a_{i,j,w}x_{j,w})$, where $a_{i,j,k}$ are unknowns that we intend to find. On expanding equation \ref{patch} and  comparing coefficient of monomials on both sides, we will get  at most $  (2k^2w)^{2k^2+1}$ polynomial equations of degree at most $  2k^2+1$ in at most $ k(2k^2+1)w$ variables. Thus we can solve this in $\Sys(k(2k^2+1)w,(2k^2w)^{2k^2+1},2k^2+1)$ time by Theorem~\ref{thm:polysystem}.


\end{proof}

\subsection{Putting it all together}
We now show how to combine all the lemmas and subroutines developed so far to get the full reconstruction algorithm for $\SMLSps(k)$ circuits. The theorem below is basically a restatement of Theorem~\ref{THEOREM:main1} ((only for the randomized algorithm over general fields). After the proof we will comment on how the algorithm can be derandomized over $\R$ and $\C$.

\begin{theorem}
Given black-box access to a degree $d$, $n$-variate polynomial $f \in \F[X]$ such that $f$ is computable by a $\SMLSps(k)$ circuit $C_f$  over the field $\F$, there is a randomized $\poly(d^{k^3}, k^{k^{k^{10}}}, n)$  time algorithm that outputs a $\SMLSps(k)$ circuit computing $f$. 
\end{theorem}

\begin{remark}
 when $\F = \F_q$  and when $\F = \R$ or $\C$, the output circuit is over the same underlying field $\F$. In general the output circuit might be over an algebraic extension of $\F$. 
\end{remark}

\begin{proof}
By Corollary~\ref{cor:widthreduction}, it suffices to assume that the width of $f$ is at most $k$.

Our algorithm is recursive and we assume that we have an efficient algorithm for reconstructing $\SMLSps(k-1)$ circuits that runs in time $\mc{A}(n,k-1,d)$. For the base case of $k=1$, the reconstruction algorithm follows directly from black-box factoring result of \cite{KaltofenTrager90}. 

Now assume $k \geq 2$. By Corollary~\ref{cor:simple} we can assume that $C_f$ is simple. We can also assume that $C_f$ is minimal. This is because if $f$ has a $\SMLSps(k)$ representation with a non-minimal circuit, then it also has $\SMLSps(k)$ representation with a minimal circuit, which is obtained by just deleting any subset of multiplication gates in the non-minimal circuit which sums to zero.

If $d \leq 2k^2$, then we invoke the algorithm in Lemma~\ref{lem:setmultilinear-lowdegree} to learn a $\SMLSps(k)$ representation of $C$.

If $d > 2k^2$, we invoke the algorithm from Lemma~\ref{lem:twolinearforms} to compute the set $L$ of pairs of linear forms. For each pair $(\ell_1, \ell_2) \in L$ we do the following: We invoke Lemma~\ref{lem:mostlinearforms1} to compute a set $S$ of at most $2k-2$ $\Pi\Sigma$ circuits and then invoke Lemma~\ref{lem:mostlinearforms2} to compute a set $\bar S$ of $k$-tuples. Let us call this final set $\bar{S}_{(\ell_1, \ell_2)}$. For each $k$-tuple $(G_1, G_2, \ldots, G_k) \in \bar{S}_{(\ell_1, \ell_2)}$ we invoke the algorithm of Lemma~\ref{lem:fullcircuit} with $w=k$ to output a circuit. We then verify that the output circuit indeed has the $\SMLSps(k)$ format and then we check (by running a polynomial identity testing algorithm) if it computes $f$. If it passes both theses verification steps then the algorithm halts and outputs that circuit. By Lemmas~\ref{lem:twolinearforms},~\ref{lem:mostlinearforms1},  ~\ref{lem:mostlinearforms2} and ~\ref{lem:fullcircuit}, we do know that for some choice of $(\ell_1, \ell_2)$ and for some choice of $(G_1, G_2, \ldots, G_k) \in \bar{S}_{(\ell_1, \ell_2)}$, the algorithm will succeed with high probability.

\textbf{Time complexity analysis}:
The  upper bound on the time complexity is $poly(n,d,k) \cdot  \mc{A}(n,k,d)$. Recall, $\mc{A}(n,k,d)$ is the time complexity of learning degree $d$, width $k$ $\SMLSps(k)$ circuit computing $f$. We now upper bound $\mc{A}(n,k,d)$. Note that,
\begin{align*}
\mc{A}(n,k,d) \leq & \, 2 \cdot  \mc{A}(n,k-1,d) +  \Sys(k(2k^2+1)k,(2k^2k)^{2k^2+1},2k^2+1) + \poly(n,d) + \poly(d^{k^3}).\\
 \leq & \, 2^k \cdot  \mc{A}(n,1,d) +  k \cdot \Sys(k(2k^2+1)k,(2k^2k)^{2k^2+1},2k^2+1) + k\poly(d^{k^3}) \leq \poly(d^{k^3}, k^{k^{k^{10}}}, n).
\end{align*}

So, the total time complexity is also bounded by $\poly(d^{k^3}, k^{k^{k^{10}}}, n)$.
\end{proof}

\textbf{Derandomization:} 
We list below the places in the proof where randomization is used, and state how to derandomize them. 
\begin{itemize}
    \item {\it Polynomial system solving:} This is used in two different places in the proof and in both cases can be substituted with deterministic algorithms for the same when the underlying field is $\R$ or $\C$ (See Theorem~\ref{thm:polysystem}). It is worth noting that this is the only step where the derandomization does not work over all fields. 
    \item {\it Blackbox factoring:} For this step we had used the randomized blackbox factoring algorithm by Kaltofen and Trager. To derandmize this step one can use the deterministic factoring algorithm for multilinear polynomials given in \cite{ShpilkaVolkovich10} along with a hitting set of $\SMLSps(k)$ circuits.
    \item{\it Variable or ``width" reduction:} In this step we had used Lemma~\ref{invert1} to find an assignment  s.t. $A_{\alpha_i}$ matrix is invertible. However in our setting we can use the deterministic version of this lemma instead, i.e. Lemma~\ref{derand-invert} since we have efficient hitting sets for $\SMLSps(k)$ circuits and for sums of constantly many $\SMLSps(k)$ circuits. 
    \item {\it Computing linear dependence among the $G_i$s:} In this step, instead of using the randomized algorithm from Lemma~\ref{invert}), we can instead use the derandomized version of it, i.e. Lemma~\ref{derand-invert1} since we have efficient hitting sets for $\SMLSps(k)$ circuits and for sums of constantly many $\SMLSps(k)$ circuits. 
    

\end{itemize}





\section{Multilinear Depth-3 Circuits}\label{sec:multilinear}
In this section, we will provide a proof of Theorem~\ref{THEOREM:main3}. In all running times stated in this section, we have suppressed a $\poly(c)$ multiplicative dependence in the running time, where $c = \log q$ if $\F = \F_q$ and $c$ is the maximum bit complexity of any coefficient of $f$ if $\F$ is infinite.

Let us start by revisiting some core definition related to depth-3 circuits. 

\begin{definition}
\label{def:depth-3}

A depth-$3$ $\Sps(k)$ circuit $C$ of degree (at most) $d$
computes a polynomial of the form
\begin{equation*}
C \equiv \sum\limits _{i = 1}^k T_i(X)= \sum_{i= 1}^k
\prod_{j=1}^{d_i}\ell_{i,j}(X),
\end{equation*}
where $d_i \leq d$ and the $\ell_{i,j}$-s are linear functions;
$\ell_{i,j}(X)= \sum \limits _{t=1} ^{n} a^t_{i,j}x_{t}+ a^0_{i,j}$ with
$a^t_{i,j} \in \F$.
 \\
A \emph{multilinear} $\Sps(k)$ circuit is a $\Sps(k)$ circuit in which each $T_i$ is a
multilinear polynomial. In particular, each such $T_i$ is a product of variable-disjoint linear functions. \\

Following notations will be useful throughout the paper.

\begin{enumerate}
\item For each $A \subseteq [k]$, $C_A$ is defined as a subcircuit of C supported on $A$, formally, $C_A= \sum_{i \in A} T_i$.
\item $\gcd(C) \eqdef \gcd(T_1, T_2, \ldots, T_k)$.
\item $\rank(C)= \dim(\mathrm{span}\set{\ell_{i,j}})$.
\end{enumerate}
\end{definition}

Observe that for a multilinear circuit: $d \leq \rank(C)$. \\

Based on the notion of rank, in \cite{KarninShpilka09},
Karnin and Shpilka defined a ``distance function” for depth-3 circuits.

\begin{definition}[\cite{KarninShpilka09}]
\label{def:distancenew}
For two $\Sps$ circuits $C_1, C_2$, we define a \emph{distance} function: $$\Drank(C_1, C_2) \eqdef  \rank \left( \frac{C_1 + C_2}{\gcd(C_1 + C_2)} \right).$$

For a single circuit $C$, we define the \emph{GCD-free}-rank as:
$$\Drank(C) \eqdef \Drank(C,0) = \rank \left( \frac{C}{\gcd(C)} \right).$$

\end{definition}





The following result known as the \emph{Rank Bound} provides a structural property for multilinear depth-3 computing the zero polynomial, under some technical conditions.

\begin{theorem}[\cite{SaxenaSeshadhri09}]
\label{thm:rank bound ml}
There exists a monotone function $R_M(k) \leq \BigO(k^3 \log k)$ such that any simple and minimal, multilinear $\Sps(k)$ circuit $C$, computing the zero polynomial satisfies $rank(C) \leq R_M(k)$.  
\end{theorem}

\subsection{Learning Low-Degree Multilinear $\Sps(k)$ Circuits}
\label{sec:low degree}

In this section we will show how to reconstruct a low-degree multilinear $\Sps(k)$ circuit from black-box samples. (For now we only state the randomized version and later point out how to derandomize it over $\R$ and $\C$). We state the lemma below for general $k$ and $d$, but think of $k$ and $d$ to be constants, and the number of variables, $n$, to be growing.

\begin{lemma}
\label{lem:low degree}
Let $f \in \Fn$ be a polynomial computed by a degree $d$, multilinear $\Sps(k)$ circuit $C_f$ of the form 
\begin{equation*}
\sum\limits_{i = 1}^k T_i(X)= \sum_{i= 1}^k
\prod_{j=1}^{d_i}\ell_{i,j}(X)
\end{equation*} 
Then there is a randomized algorithm that given $k,d$ and black-box access to $f$ outputs a multilinear $\Sps(k)$ circuit computing $f$, in time $\poly{(n,\Sys(d^2k^2, kd^2n+ {dk+d\choose k},d))} \leq (dkn)^{{(d^2k^3)}^{\BigO(d^2k^2)}}$.
\end{lemma}

\begin{proof}
Let $m$ be the number of essential variables in $f$. Since there at most $kd$ linear forms appearing in $C$, this it is easy to see that $m \leq kd$.

By Lemma~\ref{lem:carlini-black-box}, there is a polynomial-time randomized algorithm that given black-box access to $f$, computes an invertible linear transformation $A \in \F^{n \times n}$ such that $f(A \cdot \bar{x})$ only depends on the first $m$ variables. 

Let $g(X) = f(A \cdot \bar{x})$.
Observe that given black-box access to $f$, one can easily simulate black-box access to $g$, since in order to evaluate $g$ at any input $\alpha \in \Fn$, one has to simply evaluate $f$ at $A \cdot \alpha$. 

Also observe that $g(A^{-1} \cdot \bar{x}) = f(\bar{x})$. 
Thus any algorithm that can efficiently learn $g$ can also efficiently learn $f$ in the following way. For each $i \in [n]$, suppose that $R_i$ denote the $i$th row of $A^{-1}$. Then in the $i$th input to $g$ simply input the linear polynomial $L_i = \langle R_i, \bar{x} \rangle$, which is the inner product of $R_i$ and the vector $\bar{x}$ of formal input variables. Since $g$ only depends on the first $m$ variables, we only really need to do this operation for $i\in [m]$. 

Since $f$ is computed by a degree $d$ multilinear $\Sps(k)$ circuit, hence $g(\bar{x}) = f(A\cdot \bar{x})$ is also has a natural degree $d$ $\Sps(k)$ circuit representation, where the linear forms of that representation are obtained by applying the transformation $A$ to corresponding linear forms of $C$. Let us call this circuit $C_g$. Notice that $C_g$ may not be multilinear. However, if were somehow able to learn the precise circuit $C_g$, then by substituting each variable $x_i$ to $L_i$ then we would recover the circuit $C_f$ which is indeed multilinear. 

Thus our goal is now the following. We have black-box access to $g$ which only depends on $m$ variables. We would like to devise as algorithm for reconstructing $C_g$. 
Now here is a subtle point. 
$C_g$ is a particular degree $d$ $\Sps(k)$ representation of $g$. It has the nice property that when we plug in $x_i = L_i$ in this representation, then we recover a multilinear $\Sps(k)$ representation of $f$. Let us call the new object obtained by plugging in $x_i = L_i$ for each $i$, the ``lift" of $C_g$
However, $g$ might have multiple representations as a degree $d$ $\Sps(k)$ circuit. If given black-box access to $g$, the reconstruction algorithm finds some other degree $d$ $\Sps(k)$ representation of $g$, call it $C_g'$, then there is no guarantee that when we plug in $x_i = L_i$ in this representation, then we recover a multilinear $\Sps(k)$ representation of $f$. In other words, the lift of $C_g'$ may not be multilinear. 

Now, we will not actually be able to guarantee that we learn $C_g$. However the existence of $C_g$ tells us that there exists a $\Sps(k)$ representation of $g$ whose lift is a multilinear $\Sps(k)$ circuit. Can we find such a representaion of $g$?

We will now see that we can actually do this.
In order to learn a degree $d$ $\Sps(k)$ representation of $g$ we will set up a system of polynomial equations whose solution will give as a degree $d$ $\Sps(k)$ representation. We will be able to impose additional polynomial constraints to this system that will further ensure that whatever $\Sps(k)$ representation is learnt will be such that its lift will be a multilinear $\Sps(k)$ circuit. 

The algorithm first learns $g$ as a sum of monomials. Since $g$ is of degree at most $d$ and depends on at most $kd$ variable, such a representation of $g$ can be found in time $\poly\left({kd+d \choose d }\right)$ using known sparse polynomial reconstruction algorithms~\cite{KlivansSpielman01, Ben-OrTiwari88}. Let $S$ be the set of $m$-tuples of non-negative integers that sum to $d$. Then the algorithm finds a collection of coefficients $\{c_{\bar e} \in \F | \bar e \in S\}$ such that $g = \sum_{ \bar e \in S} c_{\bar e}\cdot \bar x^{\bar e}.$

Any degree $d$ $\Sps(k)$ representation of $g$ looks like the following:
$$ \sum_{i=1}^k \prod_{j=1}^d (a^{(i)}_{j,1}x_1 + a^{(i)}_{j,2}x_2 + \ldots + a^{(i)}_{j,m}x_m + a^{(i)}_{j, m+1}) = \sum_{ \bar e \in S} c_{\bar e}\cdot \bar x^{\bar e}.$$

The algorithm already knows the set of coefficients $\{c_{\bar e} \in \F | \bar e \in S\}$. In order to learn a $\Sps(k)$ representation it needs to learn values for the coefficients in the LHS, i.e. the  $a^{(i)}_{j,r}$ for various choices of $i,j,r$. These  $a^{(i)}_{j,r}$ are the unknown variables. 

Now for each monomial $x^{\bar e}$ that appears in $g$, we can compare the coefficient of it on the LHS and RHS of the above expression, set them equal to each other and get a polynomial equation in the unknown variables. We do this for all the monomials and hence set up a system of polynomial equations in the unknown variables. Each solution to this system of equations corresponds to a degree $d$ $\Sps(k)$ representation of $g$ and vice versa. 

We are looking for a degree $d$ $\Sps(k)$ representation whose lift it multilinear. To ensure this, we will add some additional polynomial constraints to our system of polynomial equations. 

Now suppose that $$ \sum_{i=1}^k \prod_{j=1}^d (a^{(i)}_{j,1}x_1 + a^{(i)}_{j,2}x_2 + \ldots + a^{(i)}_{j,m}x_m + a^{(i)}_{j, m+1}) $$ represents some degree $d$ $\Sps(k)$ representation of $g$. (We still treat the $a^{(i)}_{j,r}$ as unknown variables). In order for its lift to be multilinear, we would need to look at the expression $$ \sum_{i=1}^k \prod_{j=1}^d (a^{(i)}_{j,1}L_1 + a^{(i)}_{j,2}L_2 + \ldots + a^{(i)}_{j,m}L_m + a^{(i)}_{j, m+1}) $$
and in the above expression, any two linear polynomials appearing in the same multiplication gate should be variable disjoint. Now consider a linear polynomial $$L^{(i)}_j = (a^{(i)}_{j,1}L_1 + a^{(i)}_{j,2}L_2 + \ldots + a^{(i)}_{j,m}L_m + a^{(i)}_{j, m+1})$$ appearing in the expression. Each $L_i$ is a linear form in $x_1, \ldots, x_n$ and the algorithm knows what these $L_i$ are. 
Thus upon expanding and collecting terms, we see that $L^{(i)}_j$ is a linear polynomial in the $x_1, \ldots, x_n$, with coefficients being linear combinations of $a^{(i)}_{j,1}, a^{(i)}_{j,2}, \ldots, a^{(i)}_{j, m+1}$.
Now for the lift to be multilinear, we need that for each $i\in [k]$, $L^{(i)}_1, L^{(i)}_2, \ldots, L^{(i)}_d$ are mutually variable disjoint. 
In order for $L^{(i)}_j$ and $L^{(i)}_r$ to be variable disjoint, we need to ensure that for each $t\in [n]$, one of the coefficients of $x_t$ in $L^{(i)}_j$ and $L^{(i)}_r$ is zero. Equivalently, it suffices that the product of the coefficient of $x_t$ in $L^{(i)}_j$ and the coefficient of $x_t$ in $L^{(i)}_r$ is zero. This equality is in fact a polynomial constraint in the $a^{(i)}_{j,1}, a^{(i)}_{j,2}, \ldots, a^{(i)}_{j, m+1}$ and the $a^{(i)}_{r,1}, a^{(i)}_{r,2}, \ldots, a^{(i)}_{r, m+1}$ variables. 

We add this polynomial equation to our system of polynomial equations. We do this for each $i \in [k]$, for each $j, r \in [d]$ where $j \neq r$, and each $t \in [n]$. Thus we add about $k\cdot d^2 \cdot n$ additional polynomial equations. 

Then observe that any solution to the new system will have the property that the lift will be multilinear. Moreover the existence of $C_g$ guarantees that the system will have at least one solution, and hence it solvable in time $\Sys(mdk, kd^2n+ {dk+d\choose k},d)$. And the overall time complexity of the algorithm is bounded by $\poly{(n,\Sys(d^2k^2, kd^2n+ {dk+d\choose k},d))} \leq  (dkn)^{{\BigO(d^2k^3)}^{(d^2k^2)}}$.
\end{proof}

We observe that Lemma \ref{lem:low degree} can be extended in two aspects: first, as $d \leq \rank(C)$ one can immediately extend the algorithm to the case when the rank ($\rank(C)$) is ``small'' . It turns out, though, that we can extend the algorithm further to the case when gcd-free-rank 
($\Drank(C)$) is small. Note that this is not an immediate extension as one can have a high-degree circuit with constant $\Drank(C)$. One such example would be a circuit in which all the multiplication gates are equal. To avoid such situations we use the algorithm in Corollary \ref{cor:simple} to factor out $\gcd(C)$ (which is a product of linear functions).
Second, we can find a circuit with the smallest possible fan-in, by starting with $k=1$ and increasing it, until we can found a valid circuit. Note that, we can  verify the correctness of our output using Lemma \ref{lem:pit sps}.
The above discussion gives rise the following lemma, the proof of which is left as an easy exercise to the reader.

\begin{lemma}
\label{lem:low rank exact}
Let $f \in \Fn$ be a polynomial computed by multilinear $\Sps(k)$ circuit $C$ with $\Drank(C) \leq r$.
Then there is a randomized algorithm that given $k,r$ and black-box access to $f$ outputs a multilinear $\Sps(k')$ circuit computing $f$, where $k' \leq k$ is the smallest possible fan-in,
in time \sloppy$\poly{(n,\Sys(r^2k^2, kr^2n+ {rk+r\choose k},r))} \leq (rkn)^{{\BigO(r^2k^3)}^{(r^2k^2)}}$.
\end{lemma}

 \textbf{Derandomization:} The only steps where randomization is required for Lemma \ref{lem:low degree}(learning low-degree multilinear $\Sps(k)$ circuit) are in the variable reduction step (Lemma \ref{lem:carlini-black-box}) and polynomial system solving (Theorem~\ref{thm:polysystem}). Over $\R$ and $\C$, Theorem~\ref{thm:polysystem} in fact states that polynomial system solving can be done {\it deterministically} in the same time complexity.
 
 Moreover, for derandomized variable reduction, we can use Lemma \ref{derand-carlini} instead of Lemma \ref{lem:carlini-black-box}. Observe that all the assumptions of Lemma \ref{derand-carlini} are satisfied, since
 a low degree multilinear $\Sps(k)$ circuit only has constantly many linear forms and hence constantly many essential variables. Moreover the class is  closed under taking first order partial derivatives. Furthermore, $\underbrace{\mc{C}+\ldots+\mc{C}}_{k\,  times} $ for $\mc{C}$ being the class of multilinear $Sps(k)$ circuits is just the class if multilinear $\Sps(k^2)$ circuits, so by Lemma \ref{lem:pit sps}, there is an efficient hitting set for  $\underbrace{\mc{C}+\ldots+\mc{C}}_{k\,  times} $.

 
Note that, we can also derandomize Lemma \ref{lem:low rank exact}. The only place where we need randomness is used is in the step requiring gcd extraction. Using the deterministic factoring  algorithm in \cite{ShpilkaVolkovich10} along with a hitting set of multilinear $\Sps(k)$ circuits (instead of using the  Kaltofen-Trager \cite{KaltofenTrager90} algorithm) gives us a deterministic algorithm for this step.  




\subsection{Learning High-Degree Multilinear $\Sps(k)$ Circuits}

We show that high-degree case reduces to the low-degree case. More precisely, the high-degree case reduces to the low-gcd-free-rank case, which in turn reduces to the low-degree case. Algorithmically, we invoke Lemma \ref{lem:low rank exact} together with Lemma \ref{lem:eval}, that simulates a black-box access to all low-gcd-free-rank components of a circuit.

\subsubsection{Clustering Algorithm}

In \cite{KarninShpilka09}, a ``clustering'' algorithm for $\Sps(k)$ circuits was proposed.
Intuitively speaking, this algorithm merges multiplication gates with ``high'' GCD into clusters. One can also think of these clusters as finding a partition of $[k]$, where all the gates $T_1$ to $T_k$ which are ``close'' together according to the $\Drank$ distance function merge to form a partition of $[k]$. We formalize this  notion below:

\begin{definition}[\cite{KarninShpilka09}]
Let C be a multilinear $\Sps(k)$ circuit and $I = A_1 \cupdot \ldots \cupdot A_s = [k]$ be some partition of $[k]$. For each $i \in [s]$, define $C_i \eqdef C_{A_i}$. The set $\{C_i\}^s_{i=1}$  is called 
\emph{a partition} of $C$. For $\kappa, r \in \N$, we say a partition $\{C_i\}^s_{i=1}$ is \emph{$(\kappa, r)$-strong} when the following conditions hold:
\begin{itemize}
\item $\forall i \in [s], \; \Drank(C_i) \leq r$. 
\item  $\forall i \neq j  \in [s], \; \Drank(C_i, C_j) \geq \kappa \cdot r$.
\end{itemize}
\end{definition}

We now give the main relevant result.

\begin{lemma}[Clustering Algorithm of \cite{KarninShpilka09}]
\label{lem:clustering}

Let  $n,k, \ri, \kappa \in \N$.
There exists an algorithm that given $\ri, \kappa$ and $n$-variate multilinear $\Sps(k)$ circuit $C$ as input, outputs $r \in \N$ such that $\ri \leq r \leq k^{(k-2) \cdot \log_k(\kappa)} \cdot \ri$ and a $(\kappa, r)$-strong partition of $[k]$, in time $\BigO(\log(\kappa) \cdot n^3k^4),$
\end{lemma}

A key corollary of this result is that for sufficiently (yet, still modestly) large parameters, any two clustered representations of (possible even different) circuits computing the same polynomial are identical (up to a permutation). In that sense, we can say that the clustered representation is unique!

\begin{corollary}[Implicit in \cite{KarninShpilka09}]
\label{lem:unique clusters}
Let  $n,k, \ri, \kappa \in \N$ such that
$\ri \geq R_M(2k)$ \footnote{$R_M(k)$ is the so-called ``Rank Bound'' from Theorem \ref{thm:rank bound}.} and $\kappa > k^2$,
and let $C$ and $C'$ be two minimal multilinear
$\Sps(k)$ circuits computing the same non-zero polynomial. Furthermore, let $C_1,\ldots,  C_s$ and $C'_1,\ldots,  C'_{s'}$  be the the partitions of $C$ and $C'$, respectively, found by the clustering algorithm on inputs $\kappa, \ri$ together with $C$ and $C'$, respectively.
Then $s'=s$ and there exist a permutation $\pi:[s] \to [s]$ such that $ \forall i: C_i \equiv C_{\pi(i)}'$.
\end{corollary}

Given the above, we can define a \emph{canonical partition} of a circuit.

\begin{definition}
\label{def:cannon}
Let $C$ be a minimal multilinear
$\Sps(k)$ circuit computing a non-zero polynomial. We define $\Con(C) \eqdef (C_1, \ldots C_s)$ as the output of the clustering algorithm, given $\ri= R_M(2k), \kappa = k^3$ and $C$ as input.
\end{definition}

Observe that for all $i \in [s]: \Drank(C_i) \leq k^{\BigO(k)}$.
Nonetheless, $\Con(C)$ will never be explicitly computed as it requires the hidden circuit $C$ itself as an input, finding which is the very purpose of the reconstruction algorithm! 
Yet, a further key observation utilized in \cite{KarninShpilka09} is that the uniqueness of clustered representation still holds true if we restrict the circuits to a well-chosen, yet low-dimensional affine space. These are referred to as \emph{rank-preserving subspaces} (a formal definition is given in Definition \ref{def:rank preserving space}). We first state the aforementioned uniqueness property and then discuss rank-preserving subspaces and their constructions in Section \ref{sec:rank preserving}.

\begin{lemma}[Implicit in \cite{KarninShpilka09}]
\label{lem:rank preserving uniqueness}
Let $C$ be a minimal multilinear
$\Sps(k)$ circuit computing a non-zero polynomial, let $\Con(C) = (C_1, \ldots C_s)$
and let $V$ be $k^{\BigO(k)}$-multilinear-rank-preserving for $C$.
Furthermore, let $C' \equiv C \restrict{V}$ and  $C'_1,\ldots,  C'_{s'}$  be the the partition of $C'$ found by the clustering algorithm on inputs $\ri= R_M(2k), \kappa = k^3$ and $C'$.
Then $s'=s$ and there exist a permutation $\pi:[s] \to [s]$ such that $ \forall i: C'_i = C_{\pi(i)} \restrict{V}$.

\end{lemma}

\subsubsection{Rank Preserving Subspaces}
\label{sec:rank preserving}

In this section we formalize the notion of multilinear rank-preserving subspaces introduced in \cite{KarninShpilka08,KarninShpilka09} and show new constructions. We begin with a definition.

\begin{definition}[\cite{KarninShpilka08,KarninShpilka09}]
\label{def:rank preserving space}
Let $C \equiv \sum \limits_{i=1} ^k T_i = \sum \limits_{i=1} ^k \prod \limits_{j=1} ^{d_i}\ell_{i,j}$ be a multilinear $\Sps(k)$ circuit and $V$ an affine subspace. We say that $V$ is \emph{$r$-multilinear-rank-preserving} for $C$ if the following properties hold:

\begin{enumerate}
\item For any two linear functions $\ell_{i,j} \nsim \ell_{i',j'}$ appearing in $C$, we either have that $\ell_{i,j} \restrict{V} \nsim \ell_{i',j'} \restrict{V}$ or that both
$\ell_{i,j} \restrict{V}, \ell_{i',j'} \restrict{V}$ are constant functions.

\item  $\forall A \subseteq [k]$, $\rank(\simp(C_A ) \restrict{V} ) \geq \min \{\rank(\simp(C_A)), r\}$.
 
\item  No multiplication gate $T_i$ vanishes on $V$. In other words, for all $i \in [k]: T_i \restrict{V} \nequiv 0$.

\item The circuit $C\restrict{V}$ is a multilinear circuit.
\end{enumerate}
\end{definition}
 
In \cite{KarninShpilka08}, a construction of such subspaces was given. Unfortunately, we cannot use this construction directly as it is very ``rigid''; we will need something ``less structured''. Nonetheless, we will build on (and, in fact, generalize) this construction to fit our needs.  

\begin{definition}[\cite{KarninShpilka08}]
For a set $B \subseteq [n]$, we define
$V_B \eqdef \mathrm{span} \condset{e_i}{i \in B}$, where $e_i \in \bools{n}$ denotes the $i$-th standard basis vector.
\end{definition}
 
This definition was used as the first step of the construction of \cite{KarninShpilka08}. Indeed, it was shown that it ``almost'' works.
 
 \begin{lemma}[\cite{KarninShpilka08}]
\label{lem:B}
 Let $C$ be a multilinear $\Sps(k)$  circuit and $r \in \N$. Then there exists a subset $B \subseteq [n]$ of size $\size{B} = 2^k \cdot r$ such that for every $B' \supseteq B$ and
 $\ub \in \F^n$:
 \begin{enumerate}
\item  $\forall A \subseteq [k]$, $\rank(\simp(C_A ) \restrict{V_B + \ub} ) \geq \min \{\rank(\simp(C_A)), r\}$.

\item The circuit $C\restrict{V_B + \ub}$ is a multilinear circuit.
\end{enumerate}
 \end{lemma}

The next (and the final) step of the construction of \cite{KarninShpilka08} was to show that for a particular shift $\ub \in \F^n$, the space $V_B + \ub$ satisfies \textbf{all} the requirements of Definition \ref{def:rank preserving space}. In what follows, we generalize this steps by expressing a general condition for $\ub \in \F^n$ under which $V_B + \ub$ satisfies \textbf{all} these conditions.
Indeed, our result is a direct application of Lemma \ref{lem:linear nsim}. Furthermore, we show a somewhat stronger statement, which in the terminology of \cite{KarninShpilka08,KarninShpilka09} is referred to as ``liftable'' rank-preserving subspace.
 
\begin{lemma}
\label{lem:rank preserving}
Let $C \equiv \sum \limits_{i=1} ^k T_i = \sum \limits_{i=1} ^k \prod \limits_{j=1} ^{d_i}\ell_{i,j}$ be a multilinear $\Sps(k)$ circuit and let $r \in \N$.
 Let $B$ be the subset from Lemma \ref{lem:B}. Then there exist a polynomial $\Phi_C(\xb)$ (independent of $r$ and $B$) of degree less than $2n^3k^2$ such that if $\Phi_C(\ub) \neq 0$ then $V_{B'} + \ub$ is $r$-multilinear-rank-preserving space
 for $C$ for every $B' \supseteq B$.
\end{lemma}
 
 \begin{proof}
 Consider the polynomial $$\Phi_C(\xb) \eqdef \prod _{i=1}^n T_i \cdot \prod _{(i,j) \neq (i',j')} D(\ell_{i,j},  \ell_{i',j'}).$$
 Here $D(R,L)$ is given by Lemma \ref{lem:linear nsim}. Indeed, Properties $1$ and $3$ of Definition \ref{def:rank preserving space} follow from Lemma \ref{lem:linear nsim}. In terms of the degree, observe that there are at most $nk$ linear forms. Therefore, $$\deg(\Phi_C) \leq nk + {nk \choose 2} \cdot n  < nk + n^3k^2 < 2n^3k^2.$$
 \end{proof}

We conclude this section with two observations. The first observation was implicitly made in \cite{KarninShpilka08} and was, in fact, used in their construction of rank-preserving subspaces.

\begin{observation}
\label{obs:Phi}
For every $C: \Phi_C\left( 1,y,y^2, \ldots, y^{n-1} \right)$ is non-zero univariate polynomial in $y$ of degree less than $2n^4k^2$.
\end{observation}

The next observation follows immediately from the definition. 

\begin{observation}
\label{obs:V}
Let $f \in \Fn$ be a multilinear polynomial, $B \subseteq [n]$ and $\ab,\bb \in \F^n$ two assignments such that $\wh(\ab,\bb) = 1$. Finally, suppose that $\ab$ and $\bb$ differ (only) in the $i$-th coordinate. Then:
\begin{enumerate}
    \item $\left(f \restrict{V_{B \cup \set{i}}+\ab} \right) \restrict{x_i = 0} = f \restrict{V_{B}+\ab}$
    
    \item $\left(f \restrict{V_{B \cup \set{i}}+\ab} \right) \restrict{x_i = b_i - a_i} = f \restrict{V_{B}+\bb}$
    
    \item $\left(f \restrict{V_{B} + \ab} \right) \restrict{x_B = \bar{0}_B} = f(\ab)$
\end{enumerate}
\end{observation}

\subsubsection{Cluster Evaluation}

For a circuit $C$, let $\Con(C) = (C_1, \ldots C_s)$ be its canonical partition (see Definition \ref{def:cannon}). Recall that by design each $C_i$ is a ``low-rank'' circuit. That is, $\Drank(C_i) \leq  k^{\BigO(k)}$.
Therefore, if we could evaluate each such $C_i$ on an arbitrary point $\bb \in \F^n$, we could invoke the learning algorithm from Lemma \ref{lem:low rank exact} and reconstruct it.  We show how to achieve this goal via a technique similar to the one used in \cite{BSV20}. \\

We first observe that the uniqueness property w.r.t to rank-preserving spaces (Lemma \ref{lem:rank preserving uniqueness}) will allow us to evaluate the $C_i$-s on the space ${V_{B}+\ab}$ (and thus on $\ab$) for a ``random'' point $\ab$ (or a point $\ab$ with a particular structure). 
Our next step will be to change one (arbitrary) coordinate of such an $\ab$.

\begin{lemma}
\label{lem:jump1}
Let $C$ be a multilinear $\Sps(k)$ circuit and
$\Con(C) = (C_1, \ldots C_s)$ be its canonical partition.
Then there exists an algorithm that given:
\begin{itemize}
 \item Assignments: $\ab, \bb \in \F^n$ such that $\wh(\ab,\bb) = 1$ and $\Phi_C(\ab) \neq 0$
    
    \item The subset $B \subseteq [n]$ of size $k^{\BigO(k)}$ guaranteed by Lemma \ref{lem:B} for $r=k^{\BigO(k)}$.
    
    \item Ordered tuple $\left( C_1\restrict{V_{B}+\ab}, \ldots, C_s\restrict{V_{B}+\ab}\right) $
\end{itemize}
outputs the ordered tuple $\left( C_1\restrict{V_{B}+\bb}, \ldots, C_s\restrict{V_{B}+\bb}\right)$, in time $n^{k^{k^{\BigO(k)}}}$.
\end{lemma}

\begin{proof} Let $i \in [n]$ be coordinate where $\ab$ and $\bb$ differ. The algorithm operates as follows:
\begin{itemize}
   \item Learn $C' \eqdef C \restrict{V_{B \cup \set{i}}+\ab}$  using Lemma \ref{lem:low degree} with $d = \size{B}+1$.
    
    \item Run the Clustering  Algorithm from Lemma \ref{lem:clustering} on inputs $C'$, $\ri= R_M(2k), \kappa = k^3$. \\ Let $C'_1, \ldots, C'_{s'}$ be the output of the algorithm.
    
    \item For $j=1..s':$ find a coordinate $k$ such that $C'_j \restrict{x_i = 0} = C_k\restrict{V_{B}+\ab}$; Set $\sigma(k) \leftarrow j$
    
    \item Output $\left(  C'_{\sigma(1)} \restrict{x_i = b_i - a_i},
     C'_{\sigma(2)} \restrict{x_i = b_i - a_i}, \ldots,  C'_{\sigma(s')} \restrict{x_i = b_i - a_i} \right)$
\end{itemize}

We now argue correctness.  By Lemma \ref{lem:rank preserving}, ${V_{B \cup \set{i}}+\ab}$ is $k^{\BigO(k)}$-multilinear-rank-preserving for $C$. Consequently, by Lemma \ref{lem:rank preserving uniqueness},
$s = s'$ and there exists a permutation $\pi:[s] \to [s]$ such that $\forall j \in [s]: C'_j = C_{\pi(j)} \restrict{V_{B \cup \set{i}}+\ab}$.
By Observation \ref{obs:V}, $$\forall j \in [s]: C'_j \restrict{x_i=0} = C_{\pi(j)} \restrict{V_{B}+\ab}.$$
As the clusters in the partition are different, we obtain that $\forall j \in [s]: \pi(j) = k, \sigma(k) = j$ which implies that 
$$\forall k \in [s]: \pi(\sigma(k)) = k.$$
Finally,  by Observation \ref{obs:V}:
$$\forall k \in [s]: C'_{\sigma(k)} \restrict{x_i = b_i - a_i} = \left( C_{\pi(\sigma(k))} \restrict{V_{B \cup \set{i}}+\ab} \right) \restrict{x_i = b_i - a_i} = C_{\pi(\sigma(k))} \restrict{V_{B}+\bb} = C_{k} \restrict{V_{B}+\bb}.$$

The  running time follows from Lemmas \ref{lem:low degree} and \ref{lem:clustering}.
\end{proof}

Next, as in \cite{BSV20}, by applying the lemma iteratively we can extend the evaluation algorithm to handle assignments with \emph{arbitrary} Hamming distance, yet under some technical conditions. This can be considered as a grass-hopper jump.
To formulate these conditions, we will use the notations from  Definition \ref{def:line}.

\begin{corollary}
\label{cor:eval}
Let $C$ be a multilinear $\Sps(k)$ circuit and
$\Con(C) = (C_1, \ldots C_s)$ be its canonical partition.
Then there exists an algorithm that given:
\begin{itemize}
 \item Assignments: $\ab, \bb \in \F^n$ such that for all $0 \leq i \leq n-1$, $\Phi_C(\Hybrid^i(\ab,\bb)) \neq 0$.
    
    \item The subset $B \subseteq [n]$ of size $k^{\BigO(k)}$ guaranteed by Lemma \ref{lem:B} for $r=k^{\BigO(k)}$.
    
    \item Ordered tuple $\left( C_1\restrict{V_{B}+\ab}, \ldots, C_s\restrict{V_{B}+\ab}\right) $
\end{itemize}
outputs the ordered tuple $\left( C_1\restrict{V_{B}+\bb}, \ldots, C_s\restrict{V_{B}+\bb}\right)$ and hence $\left( C_1(\bb), \ldots, C_s(\bb)\right) \in \F^s$, in time $n^{k^{k^{\BigO(k)}}}$.
\end{corollary}

\begin{proof}
 Apply Lemma \ref{lem:jump1} iteratively, using the ordered tuple
 $\left( C_1\restrict{V_{B}+\Hybrid^i(\ab,\bb)}, \ldots, C_s\restrict{V_{B}+\Hybrid^i(\ab,\bb)}\right)$ to compute the ordered tuple  $\left( C_1\restrict{V_{B}+\Hybrid^{i+1}(\ab,\bb)}, \ldots, C_s\restrict{V_{B}+\Hybrid^{i+1}(\ab,\bb)}\right)$, for $0 \leq i \leq n-1$,  recalling that $\ab = \Hybrid^0(\ab,\bb)$ and $\bb = \Hybrid^n(\ab,\bb)$. 
The last part follows from Observation \ref{obs:V}. 
\end{proof}

Now we show how to evaluate $\Con(C)$ on $V\restrict{B+\bb}$ for an \emph{arbitrary} $\bb \in \F^n$ and hence $\Con(C) \restrict{\xb=\bb}$
. In order to do this, we will consider the line $\ell_{\ab,\bb}(t)$ through $\ab$ and $\bb$ and show that ``most'' points $\bar{u}$ on this line do satisfy the 
condition that $\Phi_{C}(\bar{u}) \neq 0$. Once we have this, by
Corollary~\ref{cor:eval}, we will show that for most points $\bar{u}$ on the line, 
$\Con(C) \restrict{B+\ub}$
 can be computed accurately. We then apply noisy polynomial interpolation (for instance the Berlekamp-Welch algorithm for decoding Reed-Solomon Codes) to recover the entire univariate polynomial which is $\Con(C)$ restricted to $V \restrict{B +\ell_{\ab,\bb}(t)}$, and from this we can 
 recover $\Con(C)$  on $V\restrict{B+\bb}$, and hence on $\bb$.

\begin{lemma}
\label{lem:eval}
Let $C$ be a multilinear $\Sps(k)$ circuit and
$\Con(C) = (C_1, \ldots C_s)$ be its canonical partition.
Then there exists an algorithm that given:
\begin{itemize}
 \item Assignments: $\ab, \bb \in \F^n$ such that $\Phi_C(\ab) \neq 0$.
    
    \item The subset $B \subseteq [n]$ of size $k^{\BigO(k)}$ guaranteed by Lemma \ref{lem:B} for $r=k^{\BigO(k)}$.
    
    \item Ordered tuple $\left( C_1\restrict{V_{B}+\ab}, \ldots, C_s\restrict{V_{B}+\ab}\right) $
\end{itemize}
outputs the ordered tuple  $\left( C_1(\bb), \ldots, C_s(\bb)\right) \in \F^s$ in time $n^{k^{k^{\BigO(k)}}}$.
\end{lemma}

 
 The algorithm and the proof mimic Lemma $5.4$ from \cite{BSV20}.

 \begin{proof}
  Let $W \subseteq \F$ be a subset
  of size $\size{W}=5n^4k^2$ and let $f = (f_1, \ldots, f_s) : \F \to \F^s$ be a function to be specified later.  The algorithm operates as follows:

\begin{itemize}
    \item For each $\alpha \in W$, define $f(\alpha)$ as the output of the algorithm in Corollary \ref{cor:eval} for the input assignments $\ab$ and $\ub = \ell_{\ab,\bb}(\alpha)$.
    
    \item For $i \in [s]$: use noisy polynomial interpolation (Lemma \ref{lem:RS})
    on $f_i$ to recover a polynomial    $\hat{f}_i(t)$ of degree at most $n$
    
    \item Output $\left( \hat{f}_1(1), \ldots, \hat{f}_s(1)
    \right)$
    
\end{itemize}

We now analyse the algorithm. Consider the following polynomials:
$$Q(t) \eqdef \prod \limits _{i=1}^{n-1} \Phi_C \left( \Hybrid^i(\ab,\ell_{\ab,\bb}(t)) \right) \; , \;
P_i(t) \eqdef C_i(\ell_{\ab,\bb}(t)) \text { for } i \in [s].$$

Observe that by Corollary \ref{cor:eval}, if $Q(\alpha) \neq 0$ then $\forall i \in [s]: f_i(\alpha) = P_i(\alpha)$. We will now bound the number of roots of $Q(t)$. 
By Lemma \ref{lem:rank preserving},
 $Q(t)$ is a univariate polynomial of degree less that $2n^4k^2$. In addition, $Q \nequiv 0$ since
 $Q(0) = (\Phi_C(\ab))^n \neq 0$. Consequently, $Q(t)$ has less that
 $2n^4k^2$ roots. 
 On the other hand, for every $i \in [s]: P_i(t)$ is a univariate polynomial of degree at most $n$. By Lemma \ref{lem:RS}, for each $i \in [s]: \hat{f_i}(t) \equiv P_i(t)$. In particular, $\hat{f_i}(1) = P_i(1) = C_i(\bb)$.
\end{proof}

\subsubsection{Putting all together}

We can finally prove Theorem \ref{THEOREM:main3}.


\begin{proof}(of Theorem \ref{THEOREM:main3}.)
  Let $W \subseteq \F$ be a subset
  of size $\size{W}=2n^4k^2$.  The algorithm operates as follows: \bigskip \\
  Repeat the following steps for each $\alpha \in W$ and a subset $B \subseteq [n]$ of size $\size{B} = k^{\BigO(k)}$:  

\begin{itemize}
    \item  Let $\ab \eqdef (1,\alpha,\alpha^2, \ldots, \alpha^{n-1})$
    
      \item Learn $C' \eqdef C \restrict{V_{B}+\ab}$ using Lemma \ref{lem:low degree} with $d=\size{B}$
    
    \item Run the Clustering  algorithm from Lemma \ref{lem:clustering} on inputs $C'$, $\ri= R_M(2k), \kappa = k^3$. \\ Let $C'_1, \ldots, C'_{s'}$ be the output of the algorithm.
    
    \item For each $i \in [s']$ use the  algorithm from Lemma \ref{lem:low rank exact} on $C_i$ with $r=k^{\BigO(k)}$
     to output a circuit $\hat{C}_i$. \\
    Use Lemma \ref{lem:eval} with $\ab,B$ and $\left(C'_1, \ldots, C'_{s'}\right)$, as inputs to simulate black-box access to $C_i$
    
    \item Let $\hat{C} \eqdef \hat{C}_1 + \ldots + \hat{C}_{s'}$
    
    \item If $C \equiv \hat{C}$ (using Lemma \ref{lem:pit sps}) and $\hat{C}$ has top fan-in $\leq k$, output $\hat{C}$; \\ otherwise, continue to the next iteration.
\end{itemize}

We now analyze the algorithm. First, observe that the algorithm can only output a multilinear $\Sps(k)$ circuit that is equivalent to $C$. We will now argue that there exist at least one iteration when such a circuit is computed.

Let $B$ be the set guaranteed by Lemma \ref{lem:B} for $r = k^{\BigO(k)}$. Furthermore, by Observation \ref{obs:Phi}, 
there exists $\alpha \in W$ such that $\Phi_C(\ab) \neq 0$ for
$\ab = (1,\alpha,\alpha^2, \ldots, \alpha^{n-1})$. By Lemma \ref{lem:rank preserving}, ${V_{B}+\ab}$ is $k^{\BigO(k)}$-multilinear-rank-preserving for $C$. Thus by Lemma \ref{lem:rank preserving uniqueness},
$s = s'$ and there exists a permutation $\pi:[s] \to [s]$ such that $\forall i \in [s]: C'_i = C_{\pi(i)} \restrict{V_{B }+\ab}$. Assume WLOG that $\forall i: \pi(i) = i$ \footnote{One could define $\Con(C)$ up to a permutation. All the previous analyses would carry over for a fixed permutation.}.
Given that, Lemma \ref{lem:eval} guarantees black-box access to $\Con(C)$. Recall that $\forall i \in [s]: \Drank(C_i) \leq k^{\BigO(k)}$. Consequently, by Lemma \ref{lem:low rank exact} $\forall i \in [s]: \hat{C_i} \equiv C_i$ and hence $\hat{C} =  \hat{C}_1 + \ldots + \hat{C}_{s} \equiv C_1 + \ldots C_s = C$. Finally, by the minimality property of Lemma \ref{lem:low rank exact}, for each $i \in [s]$
the fan-in of $\hat{C_i}$ is at most the fan-in of $C_i$. hence, the fan-in of $\hat{C}$ is at most $k$.
Consequently, for the above choices of $\alpha$ and $B$ the algorithm will output a circuit, as required.

\end{proof}
 
 \textbf{Derandomization:} The only place in the entire algorithm where randomness was used was in Lemmas~\ref{lem:low degree} and~\ref{lem:low rank exact}. At the end of those lemmas we already commented on how they can be derandomized over $\R$ and $\C$.
 
\bibliographystyle{alpha}
\bibliography{bibliography}

 \appendix
 \section{Solving system of algebraic equations using elimination theory} \label{all-field-system}

The aim of this section is to show that given a  system of $m$ polynomial equations $\{f_1=0,\, f_2=0, \ldots, \, f_m =0\}$ where  $f_i \in \Fn$  of degree $d$, there exists a  randomized $\poly((dnm)^{3^n})$-time that outputs a point $\ab \in \bar{\F}$ if it exists, and otherwise outputs ``no solution". Also, the degree of extension of the  solution(outputted our algorithm ) is bounded by $ \poly((dn)^{3^n})$.

The algorithm we present in this section is based on elimination theory and extension theorem which are well known in algebraic geometry literature \cite{CLO15}. This algorithm is recursive in nature: essentially we reduce a system on $n$ variate polynomial to $n-1$ variate polynomial system and so on. When we reach $n=1$, we can use the fact that univariate systems are easy to solve\footnote{Solving a univariate system is equivalent to factoring the univariate polynomials and finding a non-trivial gcd.}, thus concluding our algorithm.

The algorithm we describe below has some corner cases, essentially to ensure that we get a non-trivial resultant. We will elaborate on each of them below:\begin{enumerate}
    \item \label{cc1}  $m=1$: Then we can't take resultant as it requires atleast 2 polynomials. However, we can just substitute n-1 variable to random values and find the solution. Note that, the solution we found here will be over an extension of degree atmost $d$.

\item \label{cc2}  $\gcd(f_1, f_2, \ldots, f_m) \neq 1$: If $\gcd(f_1, f_2, \ldots, f_m) \neq 1$ then resultant as needed in our algorithm will turn out to be identically 0, which gives a trivial System \ref{system2}. To overcome this, strip off any common gcd and then solve two separate systems, $\{f_1'=0,\, f_2'=0, \ldots, \, f_m' =0\}$ and $\{gcd(f_1, f_2, \ldots, f_m)=0 \}$, where $f_i'=\frac{f_i}{gcd(f_1, f_2, \ldots, f_m)}$. Note that, a solution to either system will give us a solution to $\{f_1=0,\, f_2=0, \ldots, \, f_m =0\}$.

\item \label{cc3} $m > {n+d \choose d}$: Note that, the dimension of the space of $n$ variate degree $d$ polynomials is ${n+d \choose d}$. Thus, we can always ensure that $m \leq {n+d \choose d}$ by removing any redundant/dependent $f_i$, in $\poly {n+d \choose d}$ time. 
  
\end{enumerate}

Now that we have discussed all the corner cases, we will assume that $ {n+d \choose d} \geq m> 1$ and $\gcd(f_1, f_2, \ldots, f_m) = 1$, we can proceed to discussing the core idea of this approach. That is, when we convert our $n$-variate system to $n-1$ variate such that it preserves the solutions. And,  that each solution to $n-1$ variate system is extendable.  This is exactly what we show in next lemma.

The lemma also assumes that $f_i$-s are monic in $x_1$. Note that, this can be ensured by the following shift in variables $x_i=x_i+ a_ix_1$ for random $a_i$-s.

\begin{equation}\label{system1}
\forall i \in [m], \qquad f_i(x)=0, \qquad     
\end{equation}

Define $h \in \mathbb{F}[\bar{x}, \bar{u}]:=Res_{x_1}(f_1, u_2f_2+ \ldots + u_m f_m)= \sum_{\alpha}h_{\alpha} u^{\alpha}$.

 Since we have already taken care of the case when $m=1$ (\ref{cc1}), we can assume that $m>1$ thus ensuring that $h$ is well defined. Also,   $h \neq 0$ because of corner-case \ref{cc2}.

\begin{equation} \label{system2}
\forall \alpha \qquad h_{\alpha}(x)=0
\end{equation}

\begin{lemma}
System \ref{system1} has solutions iff System \ref{system2} has solutions.
\end{lemma} 
\begin{proof}
System \ref{system1} $\implies$ System \ref{system2}: Let $(a, \bar{c})$ be a solution to \ref{system1}, since all $f_i$'s are monic, we get that there exist a non-trivial gcd among $f_1(x_1, \bar{c}),f_2(x_1, \bar{c}), \ldots, f_m(x_m, \bar{c}) $. This in-turn implies $Res_{x_1}(f_1(x_1, \bar{c}), u_2f_2(x_1, \bar{c})+ u_3f_3(x_1, \bar{c})+ \ldots u_m f_m(x_1, \bar{c})) \equiv 0 $.  Thus  system \ref{system2} will have solutions. Note that if $(a, \bar{c}) \in \mathbb{G}^n$ then system \ref{system2} also has a solution over $\mathbb{G}$, where  $\mathbb{G}$ is some extension of $\F$.

System \ref{system2} $\implies$ System \ref{system1}:
Let $\bar{c'} \in \mathbb{F} \, s.t. \, h_{\alpha}(\bar{c'})=0 \, \forall \bar{\alpha}.$ Note that, due to monicness assumption, $h(\bar{c'},u_2, \ldots u_m) = Res_{x_1}(f_1(x_1, \bar{c'}), u_2f_2(x_1, \bar{c'})+ u_3f_3(x_1, \bar{c'})+ \ldots u_m f_m(x_1, \bar{c'}))$.

Since,  $h_{\alpha}(\bar{c'})=0 \, \forall \bar{\alpha}=0$, we get that, $$Res_{x_1}(f_1(x_1, \bar{c'}), u_2f_2(x_1, \bar{c'})+ u_3f_3(x_1, \bar{c'})+ \ldots u_m f_m(x_1, \bar{c'}))=0.$$ This implies, there exist a common factor $\mathcal{F}$ with positive degree in $x_1$ of $f_1(x_1, \bar{c'})$ and  $u_2f_2(x_1, \bar{c'})+ u_3f_3(x_1, \bar{c'})+ \ldots u_m f_m(x_1, \bar{c'})$. 

Since $\mathcal{F}|f_1(x_1, \bar{c'})$ we get that $\mathcal{F} \in \mathbb{F}[x_1]$. By comparing coefficients of $u_i$-s on both sides in the following equation, $$\mathcal{F}(x_1)A(x_1, u_2, \ldots, u_m) = u_2f_2(x_1, \bar{c'})+ u_3f_3(x_1, \bar{c'})+ \ldots u_m f_m(x_1, \bar{c'}),$$ thus we get that $\mathcal{F}|f_i(x_1, \bar{c'})$, for $i > 1$. As a direct consequence, we that $(\zeta, \bar{c'})$ is a solution of system \ref{system1}, where $\zeta$ is a root of $\mathcal{F}$. Note that if $ \bar{c}' \in \mathbb{G}^{n-1}$ then $\zeta$ lies in just in a degree $2d^2$ extension of $\mathbb{G}$, where  $\mathbb{G}$ is some extension of $\F$.
\end{proof}

We will now write the skeleton for a recursive algorithm to solve system \ref{system1}. Its correctness follows from what we have discussed above.

\begin{algorithm}[H]\label{alg-system}
\KwIn{$f_1, f_2, \ldots f_m \in \Fn$ s.t each $f_i$ is monic in $x_1$ and total degree atmost $d$.}
\KwOut{ \textit{Simultaneous solution to $f_i(\xb)=0$ }}
\begin{algorithmic}
\IF  {$m=1$:}  
\STATE See corner-case \ref{cc1}.

\ELSE

    \STATE Check for non-trivial gcd using factoring. See  corner-case \ref{cc2}.
    \STATE  Compute $Res_{x_1}(f_1, u_2f_2+ \ldots + u_m f_m)= \sum_{\alpha}h_{\alpha} u^{\alpha}$
    \STATE Find $\cb'$ s.t.  $\forall \alpha\,  h_{\al}(\cb')=0$ using $ALG( \{ h_{\al}\},n-1, 2d^2)$.
\STATE  If $ALG( \{ h_{\al}\},n-1, 2d^2)$ \textit{fails} then output no solution, else solve for $x_1$ in  $f_1(x_1, \cb')=0, \ldots, \, f_i(x_1, \cb')=0$. Let the output be $\zeta$.
\RETURN $(\zeta, \cb')$.
\ENDIF

 \caption{$Algorithm_{\F}(\{f_1, f_2, \ldots f_m\}, n ,d)$}
\end{algorithmic}
\end{algorithm}

\textbf{Time complexity:}
Note that, we have dropped $m$ from the list of parameters as $m \leq {n+d \choose d}$.
Note that, the time complexity $T(n,d)$ of Algorithm \ref{alg-system} satisfies the following inequality.
 $T(n,d) \leq  poly({n+d \choose d}) + T(n-1,2d^{2}).$

Thus, $T(n,d) \leq \poly((dnm)^{3^n}) $.  Similar analysis also gives that the degree of extension of the  solution(outputted by algorithm \ref{alg-system}) is bounded by $ \poly((dn)^{3^n}).$

\end{document}